\documentclass[
  reprint,
  amsmath,amssymb,
  aps,pra,
  superscriptaddress,
  footinbib,
  floatfix,
]{revtex4-2}

\usepackage[T1]{fontenc}
\usepackage[utf8]{inputenc}
\usepackage{CJKutf8}
\makeatletter
\ifdefined\pdfnobuiltintounicode
\AddToHook{cmd/CJK@addcmap/after}{\pdfnobuiltintounicode\font@name}
\fi
\makeatother
\usepackage{lmodern}
\usepackage{amsthm}
\usepackage{bm}
\usepackage{braket}
\usepackage{graphicx}
\usepackage{capt-of}
\usepackage{booktabs}
\usepackage{microtype}
\usepackage[dvipsnames]{xcolor}
\usepackage[normalem]{ulem}
\usepackage{tikz}
\usetikzlibrary{arrows.meta,backgrounds,calc,fit,positioning}
\usepackage{xr-hyper}
\usepackage[colorlinks=true,allcolors=NavyBlue]{hyperref}

\newtheorem{theorem}{Theorem}[section]
\newtheorem{proposition}[theorem]{Proposition}
\newtheorem{lemma}[theorem]{Lemma}
\newtheorem{corollary}[theorem]{Corollary}
\newtheorem{definition}[theorem]{Definition}

\newcommand{\Tr}{\operatorname{Tr}}
\newcommand{\supp}{\operatorname{supp}}
\newcommand{\diam}{\operatorname{diam}}
\newcommand{\dist}{\operatorname{dist}}
\newcommand{\id}{\mathbb{I}}
\newcommand{\Had}{\mathrm{Had}}
\newcommand{\cH}{\mathcal{H}}
\newcommand{\cA}{\mathcal{A}}
\newcommand{\cK}{\mathcal{K}}

\newcommand{\ee}{\mathrm{e}}

\newcommand{\rhoG}[1]{\rho_{\beta}(#1)}
\newcommand{\norm}[1]{\left\lVert #1\right\rVert}
\newcommand{\abs}[1]{\left|#1\right|}

\numberwithin{equation}{section}
\allowdisplaybreaks
\graphicspath{{figures/}}

\newcommand{\manuscripttitle}{Spectral Core--Tail Architecture for Locally Certified Gibbs-State Preparation}
\hypersetup{unicode=true,pdfauthor={Rui-Hao Li (李瑞浩)}}

\newcommand{\smref}[2]{\begingroup\let\,=.\ref{#1}\endgroup}
\newcommand{\maineqref}[1]{\eqref{#1}}
\newcommand{\maintextlink}[2]{\hyperref[#1]{#2}}

\makeatletter
\let\arxiv@newlabel\newlabel
\renewcommand*\ExplSyntaxOn
\clist_new:N \l_arxiv_citation_keys_clist
\tl_new:N \l_arxiv_citation_key_tl
\cs_new_eq:NN \arxiv_original_sort_cites:n \NAT@sort@cites
\cs_set_protected:Npn \NAT@sort@cites #1
{
  \clist_clear:N \l_arxiv_citation_keys_clist
  \clist_map_inline:nn {#1}
  {
    \tl_set:Nn \l_arxiv_citation_key_tl {##1}
    \regex_replace_once:nnN {\A main-} {} \l_arxiv_citation_key_tl
    \clist_put_right:NV \l_arxiv_citation_keys_clist \l_arxiv_citation_key_tl
  }
  \exp_args:NV \arxiv_original_sort_cites:n \l_arxiv_citation_keys_clist
}
\cs_new_protected:Npn \arxiv_enable_main_links:
{
  \cs_set_eq:NN \arxiv@refstar \@refstar
  \cs_set_protected:Npn \@refstar ##1
  {
    \str_if_eq:eeTF {\str_range:nnn {##1} {1} {5}} {main-}
    {\T@ref{##1}}
    {\arxiv@refstar{##1}}
  }
}
\AtBeginDocument{\arxiv_enable_main_links:}
\cs_new_eq:NN \arxiv_original_addtocontents:nn \addtocontents
\RenewDocumentCommand{\addtocontents}{m m}
{
  \str_if_eq:nnF {#1} {toc}
  {\arxiv_original_addtocontents:nn {#1} {#2}}
}
\cs_new_protected:Npn \arxivRestoreTOCWrites
{\cs_set_eq:NN \addtocontents \arxiv_original_addtocontents:nn}
\ExplSyntaxOff
\makeatother
\hypersetup{
  pdftitle={\manuscripttitle},
  pdfsubject={Operational spectral core--tail construction and local Gibbs-state certification}
}

\begin{document}
\begin{CJK*}{UTF8}{gbsn}

  \title{\manuscripttitle}
  \author{Rui-Hao Li (李瑞浩)}
  \email[Contact author: ]{lir9@ccf.org}
  \affiliation{Computational Life Sciences, Cleveland Clinic Research, Cleveland, Ohio 44195, USA}

  \begin{abstract}
    We develop a constructive method for designing Gibbs-state preparation circuits near tractable reference models, together with an explicit local error analysis.
    The spectral core--tail architecture (SCTA) combines a tractable thermal core and a unitary tail, defining a core-frame residual that measures the Hamiltonian mismatch from the target.
    We derive a local Gibbs-state error bound from the residual using the Kubo--Mori response theory.
    This bound is volume uniform when the response meets a shell summability condition and the relevant local data remain uniform with system size.
    We present three anchor Hamiltonian classes that admit exact core--tail constructions with zero residuals.
    For small deformations around an exact anchor, we employ a Schrieffer--Wolff-type reduction procedure that updates the core and constructs a correction to the tail.
    We prove that, under locality and solvability assumptions, the first-order reduction leaves a residual whose strength is quadratic in the deformation strength, with bounds independent of system size.
    A formal recursion extends the construction to arbitrary perturbative order at fixed volume, conditional on order-by-order solvability.
    Finite-size numerical tests on deformed graph-stabilizer Hamiltonians verify the quadratic suppression of both the residual and the local Gibbs-state error in the perturbative regime.
    The correction generators also provide a useful variational ansatz beyond perturbative control.
  \end{abstract}

  \maketitle
\end{CJK*}

\section{Introduction}
\label{sec:introduction}

{Gibbs states encode the equilibrium properties of quantum many-body systems~\cite{Alhambra2023}.}
{Preparing these states on a quantum computer provides a route to studying finite-temperature phenomena through measurements of the corresponding observables~\cite{ZhuEtAl2020,SagastizabalEtAl2021,ThanEtAl2025}.}
{Thermal-state preparation also serves as a building block for quantum algorithms that estimate partition functions, linking state preparation to the calculation of thermodynamic quantities~\cite{PoulinWocjan2009,FrancisEtAl2021}.}

For a many-body system with Hamiltonian $H$, the Gibbs state at inverse temperature $\beta>0$ is
\begin{equation}
  \rho_{\beta}(H)
  =
  \frac{\ee^{-\beta H}}{Z_{\beta}(H)},
  \label{eq:intro-gibbs-state}
\end{equation}
where $Z_{\beta}(H) = \Tr\!\left(\ee^{-\beta H}\right)$ is the partition function.
Preparing Eq.~\eqref{eq:intro-gibbs-state} on a quantum computer is difficult in general, and broadly applicable algorithms typically avoid constructing the complete eigendecomposition of $H$.
Coherent Hamiltonian-based methods use phase estimation or linear-combination-of-unitaries spectral filtering, together with amplitude amplification, to prepare a Gibbs-state purification~\cite{PoulinWocjan2009,ChowdhurySomma2017}.
Their costs can depend strongly on inverse temperature, target precision, spectral resolution, and the normalization or success amplitude of the filtered state.
Another broad family of algorithms exploits spatial locality.
One approach reconstructs the Gibbs state from overlapping local channels, with decay of correlations and an approximate quantum Markov property providing the sufficient structure for efficient preparation~\cite{BrandaoKastoryano2019}.
A second approach prepares the state as the fixed point of an engineered thermalization process.
Davies-generator-based and Kubo--Martin--Schwinger detailed-balance samplers have been developed with explicit quantum implementations~\cite{RallWangWocjan2023,ChenEtAl2025,DingLiLin2025,HahnEtAl2026}.
These results address how to construct and implement the thermalizing dynamics, but efficient Gibbs-state preparation additionally requires the dynamics to be ergodic and rapidly mixing.
Such mixing guarantees are only known for specific classes of Hamiltonians and under certain conditions, for example at sufficiently high temperature~\cite{RouzeFrancaAlhambra2026,RouzeFrancaAlhambraPRL2026} or for one-dimensional systems with short-range interactions~\cite{BergamaschiChen2026}.
Another high-temperature approach avoids quantum thermalization: above an interaction-dependent temperature threshold, a classical algorithm efficiently samples product states whose mixture approximates the Gibbs state.
Each sampled state can be prepared using a depth-one quantum circuit~\cite{BakshiEtAl2024}.

A complementary model-adapted construction starts from the spectral representation of the Gibbs state in Eq.~\eqref{eq:intro-gibbs-state}.
Writing $H=\sum_n E_n\ket{E_n}\!\bra{E_n}$ in an orthonormal energy eigenbasis gives
\begin{equation}
  \begin{aligned}
    \rho_{\beta}(H)
    &=
    \sum_n
    \frac{\ee^{-\beta E_n}}{Z_{\beta}(H)}
    \ket{E_n}\!\bra{E_n},
    \\
    Z_{\beta}(H)
    &=
    \sum_n
    \ee^{-\beta E_n}.
  \end{aligned}
  \label{eq:intro-spectral-gibbs-state}
\end{equation}
The energies $E_n$ determine the normalized Boltzmann weights $\ee^{-\beta E_n} / Z_{\beta}(H)$, whereas the eigenbasis determines how those thermal populations are embedded into the physical degrees of freedom.
Eq.~\eqref{eq:intro-spectral-gibbs-state} represents a structural identity, which by itself supplies neither an efficient way to sample the Boltzmann weights across exponentially many levels nor an efficient implementation of the exact basis transformation.
When both pieces admit exploitable structure, however, the equation suggests a physical circuit-design principle: first load a tractable population distribution in a reference basis, and then apply a unitary that maps that basis into the physical frame.
We call the population-loading component the spectral \emph{core}, and the subsequent basis transformation the \emph{tail}.
This population-first, basis-second construction is the physical starting point of the spectral core--tail architecture (SCTA).

The population--basis factorization underlying SCTA is already explicit in various variational constructions that combine classical population models with unitary basis transformations~\cite{MartynSwingle2019,VerdonEtAl2019,LiuEtAl2021BetaVQE,SewellWhiteSwingle2022}.
Related circuits generate the populations through mid-circuit measurements or ancilla registers before applying a separate basis rotation~\cite{SeliskoEtAl2024,StengerEtAl2024,ConsiglioEtAl2024,SambasivamEtAl2025}.
More recently, a nonvariational approach prepares the Gibbs state of the two-dimensional toric code by sampling classical Ising chains and applying a known duality circuit to transform the state back to the physical frame~\cite{PaezVelascoEtAl2026}.

{
  Building on this established architecture, we formalize the SCTA and develop a systematic method for designing Gibbs-state preparation circuits near tractable thermal reference models.
  These reference models admit exact Gibbs-state preparation with explicit core--tail constructions, and they serve as anchors for the perturbative circuit construction for nearby deformed Hamiltonians, making the approach more broadly applicable to many-body systems with local interactions.
  Specifically, the chosen core family and the deformation's local transition channels identify the core updates and tail corrections needed to reduce the impact of the deformation on the prepared state.
  The second major contribution of this work is a local error analysis that determines, for a given core--tail construction, how the error in the prediction of local thermal observables depends on the Hamiltonian mismatch between the target and the reference model.
  For observables that have bounded support, such as energy densities and short-range correlations, this analysis identifies conditions under which preparation accuracy can remain controlled as the system grows.
  Furthermore, the correction generators identified through the perturbative construction can be used to define variational ans\"atze that are physically motivated and whose parameters can be optimized beyond the perturbative regime.
}

The paper is organized as follows.
{Sec.~\ref{sec:architecture-certification} introduces the SCTA and develops a local error analysis for the Hamiltonian mismatch left by a chosen core and tail.}
{We use the Kubo--Mori response theory to bound the resulting local Gibbs-state error and discuss conditions under which this bound can be made volume uniform, i.e., independent of system size.}
Sec.~\ref{sec:exact-anchors} gives exact core--tail constructions for three representative Hamiltonian classes, including the graph-stabilizer Hamiltonians, independent Rydberg dimers, and quadratic fermions.
These exact anchors have zero modeling error and serve as reference points in circuit design for more general Hamiltonians.
Specifically, Sec.~\ref{sec:local-reduction} considers weak deformations around an exact anchor and uses a Schrieffer--Wolff-type reduction to update the core Hamiltonian and construct a correction to the tail circuit.
Under uniform locality and solvability assumptions, the first-order construction yields a residual whose local strength is quadratic in the deformation strength, with bounds independent of system size.
Turning this residual estimate into a local Gibbs-state guarantee additionally requires the thermal response control discussed in Sec.~\ref{sec:architecture-certification}.
In Sec.~\ref{sec:numerics}, we numerically test the first-order correction mechanism on a deformed graph-stabilizer chain and explore variational optimization of the correction circuit beyond the perturbative regime.
Sec.~\ref{sec:discussion} discusses the limitations and future directions.
Detailed proofs, exact-anchor constructions, and numerical methodology are provided in the Supplemental Material (SM)~\cite{SupplementalMaterial}.
\section{Operational architecture and local certification}
\label{sec:architecture-certification}

\subsection{Core, tail, and operational resources}
\label{sec:core-tail-architecture}

For a normalized distribution $p$ over orthogonal labels $x$, we define the spectral core and the corresponding SCTA state by
\begin{equation}
  \sigma_C=\sum_xp(x)\ket{x}\!\bra{x},
  \qquad
  \rho_{\mathrm{SCTA}}=U\sigma_CU^\dagger,
  \label{eq:scta-state}
\end{equation}
where $U$ is the tail unitary that maps the core basis to the physical basis.
The choice of $x$ is model-dependent and can be any convenient orthogonal basis for the system.
More generally, each core label may represent a fixed-size local subsystem with its own normalized internal state, provided that state can be prepared with controlled resources.

{An operational SCTA supplies a procedure for preparing $\sigma_C$ and a tail circuit with specified physical geometry, together with the resource requirements of both components.}
For example, a randomized loader samples $x\sim p(x)$ and prepares $\ket{x}$.
Averaging repeated preparations then realizes $\sigma_C$ for every system-only measurement.
A coherent loader instead prepares a pure state $\ket{\Psi_C} = \sum_x\sqrt{p(x)}\ket{x}_A\ket{x}_S$ on an enlarged Hilbert space containing an ancillary register $A$ and the system register $S$, with $\sigma_C=\Tr_A\ket{\Psi_C}\!\bra{\Psi_C}$.
These two modalities produce effectively the same core but differ in how that state can be used and may have very different resource costs.
For the rest of the discussion, we assume the core is prepared through purification unless otherwise stated, because it allows for a coherent access to the state $\sigma_C$, which may be useful for subsequent quantum processing.
Fig.~{\ref{fig:scta-architecture}} illustrates the corresponding purification-based implementation.

\begin{figure}[t]
  \centering
  \begin{tikzpicture}[
      x=1cm,
      font=\footnotesize,
      >=Latex,
      wire/.style={line width=0.65pt},
      circuitblock/.style={
        draw=black,
        fill=white,
        rounded corners=2pt,
        minimum height=0.72cm,
        align=center
      }
    ]
    \node[anchor=east] (ain) at (0.75,1.38) {$A:\ \ket{0}^{\otimes N}$};
    \node[anchor=east] (sin) at (0.75,0) {$S:\ \ket{0}^{\otimes N}$};

    \node[circuitblock,minimum width=1.35cm,minimum height=0.90cm]
    (core) at (1.80,1.38) {core loader\\[-0.15em]$C$};
    \node[circuitblock,minimum width=1.55cm,minimum height=2.05cm]
    (entangler) at (3.70,0.69) {$\mathrm{CNOT}_{A\to S}^{\otimes N}$};
    \node[circuitblock,minimum width=1.42cm,minimum height=0.90cm]
    (tail) at (5.95,0) {basis rotation\\[-0.15em]$U$};

    \draw[wire] (ain.east) -- (core.west);
    \draw[wire] (core.east) -- ($(entangler.west)+(0,0.69)$);
    \draw[wire] (sin.east) -- ($(entangler.west)+(0,-0.69)$);
    \draw[wire] ($(entangler.east)+(0,0.69)$) -- (7.35,1.38);
    \draw[wire] ($(entangler.east)+(0,-0.69)$) -- (tail.west);
    \draw[wire,-{Latex[length=2.2mm]}] (tail.east) -- (7.35,0);

    \node[below=0.10cm of core] {\scriptsize{$\sqrt{p(x)}$}};
    \node[above=0.08cm] at (6.75,1.38) {$\Tr_A$};
    \node[anchor=west] at (7.3,0) {$\rho_{\mathrm{SCTA}}$};

    \begin{scope}[on background layer]
      \node[
        draw=NavyBlue,
        fill=NavyBlue!6,
        rounded corners=5pt,
        fit=(core)(entangler),
        inner xsep=0.12cm,
        inner ysep=0.20cm
      ] (coregroup) {};
      \node[
        draw=BurntOrange,
        fill=BurntOrange!8,
        rounded corners=5pt,
        fit=(tail),
        inner xsep=0.12cm,
        inner ysep=0.20cm
      ] (tailgroup) {};
    \end{scope}
    \node[anchor=south,text=NavyBlue,font=\footnotesize\bfseries]
    at ($(coregroup.north)+(0,0.08)$) {spectral core: populations};
    \node[anchor=north,text=BurntOrange,font=\footnotesize\bfseries]
    at ($(tailgroup.south)-(0,0.12)$) {tail: eigenbasis};
  \end{tikzpicture}
  \caption{The spectral core--tail architecture in a purification-based implementation.
    The spectral core first loads amplitudes $\sqrt{p(x)}$ on the ancilla register $A$.
    The block $\mathrm{CNOT}_{A\to S}^{\otimes N}$ denotes $N$ parallel bitwise CNOT gates, with $A_j$ as control and $S_j$ as target.
    These gates correlate the registers so that tracing out $A$ leaves $\sigma_C=\sum_xp(x)\ket{x}\!\bra{x}$.
  The tail $U$ acts only on the system register $S$ and rotates the core basis, giving $\rho_{\mathrm{SCTA}}=U\sigma_CU^\dagger$.}
  \label{fig:scta-architecture}
\end{figure}
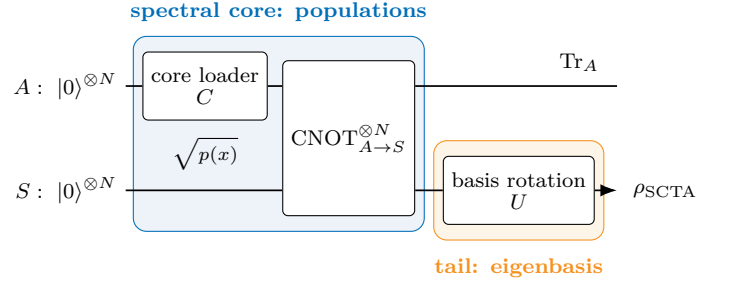

Below we distinguish constructive, variational, and hybrid implementations according to how the core and tail are specified.

\paragraph{Constructive mode.}
A constructive SCTA additionally specifies a core Hamiltonian $H_C$ satisfying $\sigma_C=\rhoG{H_C}$ and the classical and quantum resources needed to prepare the core and implement the tail to the stated accuracy.
For a diagonal core, one may write $H_C=\sum_xE_C(x)\ket{x}\!\bra{x}$.
The identity $\sigma_C=\rhoG{H_C}$ then fixes the probabilities in Eq.~\eqref{eq:scta-state} as $p(x)=\ee^{-\beta E_C(x)}/Z_C$, where $Z_C=\sum_x\ee^{-\beta E_C(x)}$.
Efficient pointwise evaluation of $E_C(x)$ determines the unnormalized Boltzmann factor for a specified label, but it does not by itself provide an efficient way to compute the partition function and the normalized probabilities, or provide a coherent loader.
In the constructive route, the main classical cost is therefore associated with the preprocessing needed to supply the core probabilities.
The quantum costs include the coherent core-loading procedure and the implementation of the tail unitary on a physical quantum device.

For a given core--tail pair, the target Hamiltonian $H$ is generally not exactly reproduced in the physical frame.
To quantify this mismatch, we distinguish between the physical frame and the core frame.
They are two choices of basis related by the tail unitary $U$: a physical operator $O$ is represented in the core frame by $U^\dagger O U$.
In general, the constructed core Hamiltonian and tail satisfy the core-frame identity
\begin{equation}
  U^\dagger H U = H_C+R,
  \label{eq:intro-core-tail-identity}
\end{equation}
where the residual $R$ records the Hamiltonian mismatch from the target $H$ after the core and tail are applied.
{The local error analysis developed in this section aims to bound the modeling error of the ideal SCTA preparation $U\rhoG{H_C}U^\dagger$ from this residual.}
{Errors in core loading and tail implementation can be incorporated through the additive error budget in Sec.~\smref{sec:supp-error-budget-proof}{S1.B} of the SM.}

\paragraph{Variational mode.}
In a variational SCTA, the core and tail are promoted to parameterized families,
\begin{equation}
  \sigma_\theta
  =
  \sum_xp_\theta(x)\ket{x}\!\bra{x},
  \qquad
  \rho_{\theta,\phi}
  =
  U_\phi\sigma_\theta U_\phi^\dagger,
  \label{eq:variational-scta-state}
\end{equation}
where $\theta$ and $\phi$ denote the core and tail parameters, respectively.
Variational Gibbs-state preparation seeks parameters that minimize the free energy $F_{\beta,H}(\rho):=\Tr(H\rho)-\beta^{-1}S(\rho)$ of the target Hamiltonian $H$, since the exact Gibbs state $\rhoG{H}$ is the unique free-energy minimizer.
The energy term can be estimated from expectation values of the terms in $H$, which can be straightforwardly measured on a quantum device.
The von Neumann entropy, $S(\rho) = -\Tr(\rho\log\rho)$, is a nonlinear function of the state and not the expectation value of a fixed observable that can be measured directly.
While estimating $S(\rho)$ does not necessarily require full state tomography, it does require information about the eigenvalue distribution of $\rho$.
General protocols obtain this information from repeated state preparations or controlled access to the preparation circuit, whose cost can grow with the Hilbert-space dimension or state rank and can worsen when the smallest nonzero eigenvalue becomes small~\cite{AcharyaEtAl2020,WangZhaoWang2023}.
Therefore, entropy estimation is generally not efficient for a generic variational state.

The core--tail separation can avoid this generic entropy-estimation problem when the core distribution is tractable.
Because the tail is unitary, it changes the eigenbasis of the state but not its eigenvalues, and hence the von Neumann entropy is invariant under the tail transformation:
\begin{equation}
  S(\rho_{\theta,\phi})
  =S(\sigma_\theta)
  =-\sum_xp_\theta(x)\log p_\theta(x).
  \label{eq:entropy-invariance}
\end{equation}
In the purification-based implementation of Fig.~\ref{fig:scta-architecture}, the same probabilities are encoded in the computational-basis statistics of the ancilla register immediately after the core-loading circuit.
Entropy evaluation can thus be done by measuring the ancilla register alone, without reconstructing the tail-rotated physical state.
This simplification does not, however, automatically make the sum in Eq.~\eqref{eq:entropy-invariance} efficient for an arbitrary distribution.
Entropy tractability requires a polynomial-resource method for evaluating or estimating the entropy of the core distribution.
For example, if the core is a product state that factorizes into single-site states, $\sigma_C = \bigotimes_{j=1}^{N} \sigma_j$, then its entropy is the sum of the entropies of the individual single-site states: $S(\sigma_C) = \sum_{j=1}^{N} S(\sigma_j)$.
It can therefore be computed from the $N$ individual single-site density operators rather than from the full distribution over $2^N$ basis states.
More generally, a block-diagonal core has an entropy equal to the classical entropy of its block probabilities plus the probability-weighted entropies within the blocks.
This remains tractable when both contributions can be evaluated with polynomial resources.

Therefore, the core--tail separation in SCTA could be used to confine the entropy calculation to a chosen tractable core family, while the tail adjusts the physical eigenbasis without changing the entropy.
This provides a natural interface for efficient variational Gibbs-state preparation.
Sec.~\ref{sec:exact-anchors} presents several core families in which this requirement is satisfied.

\paragraph{Hybrid mode.}
A hybrid SCTA combines analytically constructed and variationally optimized components within the same core--tail parameterization.
In the variational mode above, one first selects parameterized core and tail ansatz families and determines their adjustable parameters by minimizing the free energy.
The ansatz families may be chosen to have a physically motivated structure or to be hardware-efficient, but they are not guaranteed to be optimal for the target Hamiltonian.
In the hybrid mode, an analytical construction additionally determines part or all of the circuit structure.
It may also supply initial parameter values for the subsequent optimization.
Thus, the constructive input can shape both the circuit architecture and its parameterization for variational optimization.

\subsection{Local error metric and backward tail geometry}

The output of an SCTA implementation is intended to reproduce thermal observables on physically local regions, even when its global state differs appreciably from the target.
We therefore first specify the relevant local error metric and determine which core-frame region can influence a chosen output region through the tail propagation.
Let $\Lambda$ be a finite lattice with finite-dimensional onsite Hilbert spaces.

\begin{definition}[Local trace distance]
  \label{def:local-trace-distance}
  For a nonempty region $A\subseteq\Lambda$, let $\rho_A=\Tr_{A^c}\rho$ be the reduced density matrix on $A$, where $A^c=\Lambda\setminus A$ is the complement of $A$.
  At length scale $r\ge0$, define the local trace distance between two states $\rho$ and $\sigma$ on the full lattice by
  \begin{equation}
    D_r(\rho,\sigma)
    =
    \sup_{\substack{\varnothing\ne A\subseteq\Lambda\\\diam(A)\le r}}
    \norm{\rho_A-\sigma_A}_1.
    \label{eq:local-trace-distance}
  \end{equation}
\end{definition}

Here the trace norm is defined as $\norm{X}_1 :=\Tr\sqrt{X^\dagger X}$, and the diameter of a region $A$ is $\diam(A):=\max_{x,y\in A}\dist(x,y)$, with $\dist(x,y)$ the distance between sites $x$ and $y$ on a prescribed metric.
The quantity $D_r(\rho,\sigma) \in [0,2]$ measures the largest distinction between the states $\rho$ and $\sigma$ visible on any region of diameter at most $r$.
Indeed, trace-norm duality gives, for every Hermitian operator $X$,
\begin{equation}
  \norm{X}_1
  =
  \sup_{\substack{O=O^\dagger\\\norm{O}_\infty\le1}}
  \abs{\Tr(OX)}.
  \label{eq:trace-norm-duality-main}
\end{equation}
When $X = \rho_A - \sigma_A$, the quantity $D_r$ measures the largest difference between the predictions of the two states for any normalized observable confined to a region of diameter at most \(r\).

On the other hand, the tail changes which part of the core can affect one of these local measurements.
{Thus, a physical observable $O_A$ supported on $A$ is represented in the core frame by $U^\dagger O_AU$, whose support may extend beyond $A$.}
Let $\cA_A$ denote the full operator algebra supported on $A$, i.e., the set of all operators that act nontrivially only on $A$.

\begin{definition}[Backward causal region]
  \label{def:backward-causal-region}
  For an ideal tail $U$ and an output region $A$, the backward causal region $B_U(A)$ is the smallest region $B\subseteq\Lambda$ such that
  \begin{equation*}
    U^\dagger\cA_AU\subseteq\cA_B.
  \end{equation*}
  Equivalently, $B_U(A)$ is the union of the supports of $U^\dagger O_AU$ over all operators $O_A$ supported on $A$.
\end{definition}

Thus $B_U(A)$ contains precisely the core-frame degrees of freedom that can influence measurements on $A$ through the tail $U$.
To quantify its spatial enlargement, we define a neighborhood of a region $A$ by $A^{(L)}:=\{x\in\Lambda:\dist(x,A)\le L\}$, where $\dist(x,A):=\min_{y\in A}\dist(x,y)$ and $L\ge0$ specifies the neighborhood radius.

\begin{definition}[Backward tail radius]
  \label{def:backward-tail-radius}
  For a nonempty output region $A$, define
  \begin{equation}
    \ell_U(A)
    :=
    \inf\{L\ge0:B_U(A)\subseteq A^{(L)}\}.
  \end{equation}
  The worst-case radius at output scale $r$ is
  \begin{equation}
    \ell_U(r)
    :=
    \sup_{\substack{\varnothing\ne A\subseteq\Lambda\\\diam(A)\le r}}
    \ell_U(A).
  \end{equation}
\end{definition}

$\ell_U(A)$ is thus the smallest radius by which $A$ must be enlarged to contain the entire backward causal region $B_U(A)$, and $\ell_U(r)$ is the largest such radius over all regions of diameter at most $r$.
The following proposition transports core-frame distinguishability through the tail and makes the required change of length scale explicit.

\begin{proposition}[Backward-region propagation]
  \label{prop:backward-region-propagation}
  For any states $\rho$ and $\sigma$, ideal unitary $U$, and nonempty output region $A$,
  \begin{equation}
    \norm{[U\rho U^\dagger-U\sigma U^\dagger]_A}_1
    \le
    \norm{[\rho-\sigma]_{B_U(A)}}_1.
    \label{eq:backward-region-propagation}
  \end{equation}
  Consequently,
  \begin{equation}
    D_r(U\rho U^\dagger,U\sigma U^\dagger)
    \le
    D_{r+2\ell_U(r)}(\rho,\sigma).
    \label{eq:finite-tail-propagation}
  \end{equation}
  If $U$ is a depth-$d$ circuit whose gate supports have diameter at most $a$, then $\ell_U(r)\le ad$.
\end{proposition}

\begin{proof}
  See SM Sec.~\smref{sec:supp-propagation-proof}{S1.A}.
\end{proof}

Therefore, to bound the output error on $A$, it suffices to compare the core states on $B_U(A)$.
The scale change in Eq.~\eqref{eq:finite-tail-propagation} follows because enlarging $A$ by a radius $\ell_U(r)$ can increase its diameter by that amount on both sides.
For a bounded-depth circuit of bounded-range gates, the resulting core-frame scale remains independent of the total volume.


{
  For the constructive preparation $\rho_{\mathrm{SCTA}}=U\rhoG{H_C}U^\dagger$, unitary covariance and Eq.~\eqref{eq:intro-core-tail-identity} give the target Gibbs state:
  \begin{equation}
    \rhoG{H}=U\rhoG{H_C+R}U^\dagger.
  \end{equation}
  If $\ell_U(r)\le\ell$, Proposition~\ref{prop:backward-region-propagation} therefore yields
  \begin{equation}
    \begin{split}
      D_r\!\left(\rho_{\mathrm{SCTA}},\rhoG{H}\right)
      \le D_{r+2\ell}\!\left(\rhoG{H_C},\rhoG{H_C+R}\right).
    \end{split}
    \label{eq:ideal-modeling-error-main}
  \end{equation}
}

\subsection{From residual strength to term-resolved Gibbs response}
\label{sec:local-residual-response}

{Eq.~\eqref{eq:ideal-modeling-error-main} reduces the local error in an SCTA output to a comparison between $\rho_\beta(H_C)$ and $\rho_\beta(H_C+R)$ on the enlarged core-frame region.}
Whether the difference between these two states admits a useful volume-uniform bound depends not only on the magnitude of the residual $R$, but also on how that magnitude is distributed across the system.
As a result, three residual structures should be distinguished.

First, a globally small residual is controlled directly by its spectral half-width
\begin{equation}
  \begin{aligned}
    w(R)
    &:=
    \inf_{c\in\mathbb R}\norm{R-c\id}_\infty
    \\
    &=
    \frac{\lambda_{\max}(R)-\lambda_{\min}(R)}{2}.
  \end{aligned}
  \label{eq:half-width}
\end{equation}
Here $\norm{X}_\infty=\sup_{\norm{\psi}=1}\norm{X\ket{\psi}}$ is the operator norm.
Second, a residual supported on one bounded region whose diameter and cardinality are uniformly bounded in volume can be treated by the same global bound through $w(R)$ or by a local-perturbations-perturb-locally result~\cite{CapelEtAl2025}.
{The third class covers small local contributions distributed throughout the bulk, whose combined spectral half-width $w(R)$ can grow with volume.
  To treat such an extensive but locally weak residual, we introduce a term-resolved decomposition of the residual and a corresponding local strength measure.
}
We start with a declared Hermitian decomposition
\begin{equation}
  R=\sum_Xr_X,
  \qquad
  \supp(r_X)\subseteq X.
  \label{eq:local-residual-decomposition}
\end{equation}
Thus each $r_X$ is a local Hermitian operator supported on a bounded neighborhood $X$.
Since scalar shifts do not affect a Gibbs state, it is useful to center each local term according to
\begin{equation}
  \bar r_X:=r_X-c_X^\star\id_X,
  \label{eq:centered-residual-terms}
\end{equation}
where
\begin{equation}
  c_X^\star=\frac{\lambda_{\max}(r_X)+\lambda_{\min}(r_X)}2
\end{equation}
is the local shift that minimizes the operator norm of $r_X$.
Then we define the local half-width of each residual term by
\begin{equation}
  J_X
  =\norm{\bar r_X}_\infty
  =\frac{\lambda_{\max}(r_X)-\lambda_{\min}(r_X)}2,
  \label{eq:local-residual-width}
\end{equation}
and the maximum incident strength over all sites by
\begin{equation}
  \varepsilon_R
  =\sup_{x\in\Lambda}\sum_{X\ni x}J_X.
  \label{eq:incident-residual-strength}
\end{equation}

Our objective is to convert the Hamiltonian mismatch $R$ into a bound on the difference between the Gibbs states of $H_C$ and $H_C+R$.
Because the normalized Gibbs state depends nonlinearly on its Hamiltonian, this endpoint difference cannot be decomposed directly into independent contributions from the terms $r_X$.
We therefore connect the two Hamiltonians by the comparison path
\begin{equation}
  H_s=H_C+sR,
  \qquad
  \rho_s=\rhoG{H_s},
  \label{eq:comparison-path}
\end{equation}
where $s\in[0,1]$ is an auxiliary interpolation parameter.
Its endpoints are $\rho_0=\rhoG{H_C}$ and $\rho_1=\rhoG{H_C+R}$.
For a fixed observable $O$, we write its expectation value in the interpolated Gibbs state as $\langle O\rangle_s:=\Tr(\rho_sO)$.
The fundamental theorem of calculus gives
\begin{equation}
  \Tr\!\left[(\rho_1-\rho_0)O\right]
  =
  \int_0^1ds\,\frac{d}{ds}\langle O\rangle_s.
  \label{eq:integrated-expectation-main}
\end{equation}
For Hermitian observables, the exact equilibrium response is~\cite{Kubo1957,Mori1965}
\begin{equation}
  \frac{d}{ds}\langle O\rangle_s
  =-\beta\cK_s(R,O),
  \label{eq:kubo-mori-derivative}
\end{equation}
where
\begin{equation}
  \cK_s(A,O)
  =\int_0^1d\tau\,
  \left[
    \Tr(\rho_s^\tau A\rho_s^{1-\tau}O)
    -\langle A\rangle_s\langle O\rangle_s
  \right]
  \label{eq:kubo-mori-kernel}
\end{equation}
is the Kubo--Mori covariance, a noncommutative generalization of the ordinary equilibrium covariance.
We state its properties in the following proposition.

\begin{proposition}[Exact Gibbs response]
  \label{prop:exact-gibbs-response-main}
  For every Hermitian observable $O$, the path in Eq.~\eqref{eq:comparison-path} satisfies Eq.~\eqref{eq:kubo-mori-derivative}.
  On Hermitian operators, $\cK_s$ is a real symmetric positive-semidefinite bilinear form, is invariant under scalar shifts of either argument, and obeys
  \begin{equation}
    \abs{\cK_s(A,O)}
    \le
    w(A)w(O),
    \label{eq:kubo-mori-width-main}
  \end{equation}
  where $w(\cdot)$ denotes the spectral half-width defined in Eq.~\eqref{eq:half-width}.
\end{proposition}

\begin{proof}
  See SM Sec.~\smref{sec:supp-response-identity-proof}{S2.A}.
\end{proof}

The proposition is exact on every finite-dimensional system.
{Physically, the Kubo--Mori covariance $\cK_s(R, O)$ determines the instantaneous response of the expectation value of $O$ to a residual perturbation $R$ at the point $s$ along the comparison path.
  Integrating this response over the complete path thus gives the total change in the expectation value of $O$ between the two Gibbs states.
  As such, Proposition~\ref{prop:exact-gibbs-response-main} already yields a baseline estimate controlled only by the spectral half-width of the complete residual.
}

\begin{corollary}[Unconditional global bound]
  \label{cor:global-width-main}
  For every region $B$,
  \begin{equation}
    \norm{[\rho_1-\rho_0]_B}_1
    \le
    \norm{\rho_1-\rho_0}_1
    \le
    \min\{2,\beta w(R)\}.
    \label{eq:global-width-main}
  \end{equation}
\end{corollary}

\begin{proof}
  See SM Sec.~\smref{sec:supp-width-bound-proof}{S2.B}.
\end{proof}

For a globally small residual, the bound in Eq.~\eqref{eq:global-width-main} is sufficient to control the local error when $\beta w(R)\ll1$.
When a fixed local mismatch is repeated throughout the bulk, however, $w(R)$ generally grows with volume.
The locally weak regime therefore requires the declared decomposition in Eq.~\eqref{eq:local-residual-decomposition}.
By bilinearity and scalar-shift invariance of the Kubo--Mori covariance,
\begin{equation}
  \cK_s(R,O)
  =
  \sum_X\cK_s(\bar r_X,O),
\end{equation}
so the total response can be resolved into contributions from the centered local terms.

For a nonempty core-frame observation region $B$, define the normalized response of one centered residual term by
\begin{equation}
  \chi_s(X\!\to\!B)
  :={}
  \sup_{\substack{O_B=O_B^\dagger\\\norm{O_B}_\infty\le1}}
  \frac{\abs{\cK_s(\bar r_X,O_B)}}{J_X},
\end{equation}
and the corresponding path-averaged response by
\begin{equation}
  \bar\chi(X\!\to\!B)
  :={}
  \int_0^1\chi_s(X\!\to\!B)\,ds.
  \label{eq:path-averaged-response-main}
\end{equation}
Both response quantities are set to zero when $J_X=0$.
$\chi_s(X\!\to\!B)$ measures the maximum instantaneous response of any Hermitian observable on $B$ to the local mismatch on $X$, normalized by the magnitude of that mismatch.
Accordingly, $\bar\chi(X\!\to\!B)$ measures the worst-case influence of $\bar r_X$ on observables in $B$, accumulated over the complete path.
Proposition~\ref{prop:exact-gibbs-response-main} implies $0\le\chi_s,\bar\chi\le1$.

\begin{theorem}[Term-resolved local Gibbs bound]
  \label{thm:response-profile-main}
  For every nonempty region $B$ in a finite lattice,
  \begin{equation}
    \norm{[\rho_1-\rho_0]_B}_1
    \le
    \min\!\left\{2,\,
      \beta\sum_XJ_X\bar\chi(X\!\to\!B)
    \right\}.
    \label{eq:response-profile-main}
  \end{equation}
\end{theorem}

\begin{proof}
  See SM Sec.~\smref{sec:supp-localized-residual-proof}{S2.C}.
\end{proof}

The above theorem is an exact consequence of the response identity and the declared decomposition, but the right-hand side of Eq.~\eqref{eq:response-profile-main} may still grow with the total volume.
To make the nontrivial bound volume-uniform, we need to further control the path-averaged response $\bar\chi(X\!\to\!B)$.

\subsection{Shell summability and volume-uniform control}
\label{sec:shell-summability}

{
  At fixed temperature and observation scale, a useful volume-uniform bound would allow a prescribed local accuracy to be certified for arbitrarily large systems.
  A volume-dependent estimate can become uninformative for local observables as the system grows, even when their actual errors remain small.
  We therefore ask when the termwise response in Theorem~\ref{thm:response-profile-main} can be accumulated without producing a bound that grows with the total volume.
}
To supply the required locality input, suppose that for disjoint $X$ and $B$ ($X\cap B=\varnothing$),
\begin{equation}
  \bar\chi(X\!\to\!B)
  \le
  C_\beta(B)g\!\left(\dist(X,B)\right),
  \label{eq:path-averaged-response-envelope-main}
\end{equation}
where $\dist(X,B):=\min_{x\in X,y\in B}\dist(x,y)$ is the distance between the two regions, $g:[0,\infty)\to[0,\infty)$ is a nonincreasing function, and $C_\beta(B) \geq 0$ is a common response prefactor for every residual support $X$ in the declared decomposition that is disjoint from $B$.

We note that Eq.~\eqref{eq:path-averaged-response-envelope-main} is an additional analytical input and does not follow from the locality of the residual decomposition alone.
Such an envelope is expected away from thermal criticality in regimes with finite correlation length and suitable mixing or stability properties, where the influence of a local perturbation on distant observables decreases with distance.
It holds exactly for factorized paths and can be established in classical Dobrushin regimes and suitable noncommuting high-temperature phases, as discussed in Sec.~\ref{sec:physical-response-regimes}, but it may fail in low-temperature phases with long-range order or near critical points.

Eq.~\eqref{eq:path-averaged-response-envelope-main} controls the response of only one residual support, but the right-hand side of Eq.~\eqref{eq:response-profile-main} contains all supports in the lattice.
The remaining question is therefore geometric: for a fixed region $B$ in a growing lattice, does the decay of response with separation compensate for the growing number of locations at which residual terms may occur?
To express this competition, we partition the sites outside $B$ into the unit-width spatial shells labeled by $n\in\mathbb Z_{\ge0}$,
\begin{equation}
  \mathcal S_B(n)
  :=
  \{x\in\Lambda\setminus B:n\le\dist(x,B)<n+1\}.
  \label{eq:distance-shell-main}
\end{equation}
This definition applies to any graph metric.
The corresponding shell population is
\begin{equation}
  N_B(n)
  =
  \abs{\mathcal S_B(n)},
  \label{eq:shell-population-main}
\end{equation}
which counts the number of sites in the $n$th distance shell around $B$.
Finally, we define the response-weighted shell sum
\begin{equation}
  \Sigma_g(B;\Lambda)
  :=
  \sum_{n=0}^{\infty}N_B(n)g(n).
  \label{eq:shell-accumulation-main}
\end{equation}
Because $g$ is nonincreasing, every site in $\mathcal S_B(n)$ contributes at most $g(n)$, so $\Sigma_g(B;\Lambda)$ bounds the site sum $\sum_{x\in\Lambda\setminus B}g(\dist(x,B))$.
Intuitively, $N_B(n)$ measures geometric growth, $g(n)$ measures response decay, and $\Sigma_g(B;\Lambda)$ records the net effect of the two competing factors.
The following theorem combines this shell sum with the local incident-strength bound $\varepsilon_R$ to control the full response sum.

\begin{theorem}[Shell-resolved local Gibbs bound]
  \label{thm:shell-response-main}
  {Let the residual have a declared Hermitian decomposition and incident strength as in Eqs.~\eqref{eq:local-residual-decomposition}--\eqref{eq:incident-residual-strength}.}
  {Suppose the response envelope in Eq.~\eqref{eq:path-averaged-response-envelope-main} holds for every nonempty region $B$ and every residual support $X$ disjoint from $B$, with a common prefactor $C_\beta(B)$ and a nonincreasing profile $g$.}
  Then every nonempty region $B$ satisfies
  \begin{equation}
    \norm{[\rho_1-\rho_0]_B}_1
    \le
    \min\!\left\{2,\,
      \beta\varepsilon_R
      \left[\abs{B}+C_\beta(B)\Sigma_g(B;\Lambda)\right]
    \right\}.
    \label{eq:shell-response-main}
  \end{equation}
\end{theorem}

\begin{proof}
  See SM Sec.~\smref{sec:supp-shell-bound-proof}{S2.D}.
\end{proof}

The two terms in Eq.~\eqref{eq:shell-response-main} arise from distinct classes of residual supports.
For supports intersecting $B$ ($X\cap B\ne\varnothing$), the universal response bound $\bar\chi(X\!\to\!B)\le1$ and the incident-strength bound give a contribution no greater than $\varepsilon_R\abs{B}$, which produces the first term inside the square brackets.
For supports disjoint from $B$, the spatial response envelope and shell summation give the second contribution $\varepsilon_R C_\beta(B)\Sigma_g(B;\Lambda)$.

On any fixed finite lattice, $\Sigma_g(B;\Lambda)$ is finite because all sufficiently distant shells are empty.
Shell summability becomes a substantive condition only when the lattice grows.
{To formulate this condition, consider a sequence of finite metric lattices $\{\Lambda_N\}_{N\ge1}$ with $\abs{\Lambda_N}\to\infty$ on which the preceding decomposition and response-envelope assumptions hold with a common upper bound $\varepsilon_R$ on the incident strength and one response profile $g$.}
For a fixed observation scale $r$, we call the decay profile $g$ shell summable at scale $r$ when
\begin{equation}
  \Sigma_r
  :=
  \sup_{\substack{N\ge1\\\varnothing\ne B\subseteq\Lambda_N,\ \diam(B)\le r}}
  \Sigma_g(B;\Lambda_N)
  <
  \infty.
  \label{eq:uniform-shell-summability-main}
\end{equation}
This condition requires the envelope weight accumulated over all shells to remain bounded uniformly over the system size and every observation region admitted by $D_r$.

A useful sufficient test follows from polynomial shell growth.
Suppose that for some $d\ge1$ and $c_r<\infty$,
\begin{equation}
  N_B(n)
  \le
  c_r(1+n)^{d-1}
  \label{eq:uniform-polynomial-shell-growth-main}
\end{equation}
for every $N$ and every nonempty $B\subseteq\Lambda_N$ with $\diam(B)\le r$.
Here $d$ is the exponent governing polynomial shell growth.
On a regular $d$-dimensional lattice, the number of sites in a radius-$n$ ball grows as $n^d$, while the number in its unit-width boundary shell grows as $n^{d-1}$, which motivates Eq.~\eqref{eq:uniform-polynomial-shell-growth-main}.
The constant $c_r$ allows the bound to depend on the fixed observation scale $r$.
Then Eq.~\eqref{eq:uniform-shell-summability-main} follows whenever
\begin{equation}
  \sum_{n=0}^{\infty}(1+n)^{d-1}g(n)
  <
  \infty,
  \label{eq:weighted-shell-summability-main}
\end{equation}
because $\Sigma_r$ is at most $c_r$ times this series.
Finite-range, exponential, and stretched-exponential decay profiles satisfy Eq.~\eqref{eq:weighted-shell-summability-main} for every finite $d$, while a power-law envelope $g(n)\sim (1+n)^{-\eta}$ satisfies it only when $\eta>d$.

Finally, while shell summability controls the accumulated response of supports disjoint from $B$, it does not control the local term $\abs{B}$ or the response prefactor $C_\beta(B)$ in Eq.~\eqref{eq:shell-response-main}.
For the volume-uniform consequence, we additionally define
\begin{equation}
  v(r)
  :=
  \sup_{N\ge1}\sup_{x\in\Lambda_N}
  \abs{\{y\in\Lambda_N:\dist(x,y)\le r\}},
  \label{eq:uniform-ball-population-main}
\end{equation}
which bounds the number of sites in every observation region of diameter at most $r$, and
\begin{equation}
  C_{\beta,r}
  :=
  \sup_{\substack{N\ge1\\\varnothing\ne B\subseteq\Lambda_N,\ \diam(B)\le r}}
  C_\beta(B),
  \label{eq:uniform-response-prefactor-main}
\end{equation}
which bounds the response prefactor uniformly over those regions and system sizes.

\begin{corollary}[Volume-uniform local Gibbs bound]
  \label{cor:volume-uniform-shell-main}
  Write $\rho_s^{(N)}$ for the comparison-path Gibbs state on lattice $\Lambda_N$.
  If $v(r)$, $C_{\beta,r}$, and $\Sigma_r$ are finite, then
  \begin{equation}
    \sup_{N\ge1}
    D_r\!\left(\rho_1^{(N)},\rho_0^{(N)}\right)
    \le
    \min\!\left\{2,\,
      \beta\varepsilon_R
      \left[v(r)+C_{\beta,r}\Sigma_r\right]
    \right\}.
    \label{eq:volume-uniform-shell-main}
  \end{equation}
\end{corollary}

\begin{proof}
  See SM Sec.~\smref{sec:supp-volume-uniform-proof}{S2.E}.
\end{proof}

The corollary above provides a volume-uniform certificate of the Gibbs-state difference when the response is shell summable and the residual strength, response prefactor, and cardinality of bounded-diameter observation regions all remain uniformly controlled.
{For regular lattices of fixed dimension and spacing, $v(r)<\infty$ holds automatically at fixed observation scale $r$.}
{Uniform bounds on $C_{\beta,r}$ and $\Sigma_r$ can be supported in physical regimes with finite correlation length and suitable mixing properties, as discussed in Sec.~\ref{sec:physical-response-regimes}.}

\subsection{Physical regimes supporting shell summability}
\label{sec:physical-response-regimes}

We now provide examples of physical regimes in which a suitable response envelope, or a direct local-stability estimate, can be established.

\paragraph{Factorized paths.}
Suppose that the complete path $H_s$ factorizes over a fixed partition into disjoint blocks of uniformly bounded size and diameter, with arbitrary interactions within each block.
The Kubo--Mori response between a residual term in one block and an observable supported on other blocks then vanishes exactly, so only the blocks intersecting the observation region can contribute.
This gives a finite-range response envelope and hence shell summability whenever fixed-radius neighborhoods have uniformly bounded cardinality.
The complete derivation, including a stronger root-sum-square bound for the contributing block responses, is given in Sec.~\smref{sec:supp-factorized-paths}{S2.F} of the SM.

\paragraph{Classical paths.}
When every $H_s$ is diagonal in one fixed product basis, the comparison path defines a classical Gibbs distribution that may still contain spatial interactions and correlations.
The Kubo--Mori response then reduces to ordinary classical covariance.
A path-uniform Dobrushin condition bounds the direct conditional influence between sites and requires the maximal total influence on any one site to satisfy $\alpha<1$~\cite{Dobrushin1968}.
Indirect influence is consequently summable, yielding the direct bound $\norm{[\rho_1-\rho_0]_B}_1\le\min\{2,\beta\varepsilon_R\abs{B}/(1-\alpha)\}$.
Sec.~\smref{sec:supp-dobrushin-paths}{S2.G} of the SM supplies the derivation and its connection to the shell envelope~\cite{Follmer1982,RebeschiniVanHandel2014}.
It should also be noted that the Dobrushin condition is sufficient rather than necessary, and failure of $\alpha<1$ does not always imply the absence of correlation decay.

\paragraph{Noncommuting quantum paths.}
For a noncommuting path, the Kubo--Mori response involves an imaginary-time-evolved residual operator.
By the Baker-Campbell-Hausdorff (BCH) expansion, its nested commutators with $H_s$ can spread it beyond its original support, thereby allowing distant observables to respond.
Quantum belief propagation (QBP) gives an equivalent expression for the same path derivative in terms of ordinary covariances with a dressed residual operator~\cite{Hastings2007QBP,CapelEtAl2025}.
Because this dressed operator can be approximated near the original residual support, spatial decay of ordinary covariances bounds its influence on distant observables.
Thus ordinary covariance decay, together with the quasi-locality supplied by QBP, controls the thermal response without requiring a separate decay assumption for the Kubo--Mori covariance.
Within SCTA, the resulting control may be used either to establish the termwise envelope in Eq.~\eqref{eq:path-averaged-response-envelope-main} or to bound directly the local Gibbs-state change caused by the entire centered residual $\bar R = \sum_X \bar r_X$.
Theorem~34 of \citet{CapelEtAl2025} gives the latter bound under its stated hypotheses, with an error proportional to $\varepsilon_R$ even when the global operator norm of $\bar R$ grows with system size.
Sec.~\smref{sec:supp-capel-dictionary}{S2.H} of the SM states the precise assumptions and derives the resulting conditional SCTA bound.

This result requires the comparison path $H_s$ to remain local with Hamiltonian strength per site bounded uniformly in $s$ and system size, while ordinary covariances in its Gibbs states $\rho_s$ decay sufficiently rapidly with spatial separation.
{On a regular lattice, the locality and strength conditions are met when the core and residual have uniformly bounded interaction range and local strength.}
At sufficiently high temperature, established results can also supply exponential spatial decay of ordinary covariances along the path, so the imported theorem yields an exponential response control~\cite{KlieschEtAl2014,CapelEtAl2025}.
For finite-range translation-invariant chains, exponential decay of ordinary covariances can be established at every fixed finite temperature in both infinite-chain and uniform finite-chain formulations, going beyond the high-temperature regime~\cite{Araki1969,BluhmCapelPerezHernandez2022}.
Finite-range interactions are also not essential: suitable one-dimensional models with sufficiently rapidly decaying long-range couplings also have bounded strength per site and decaying thermal correlations~\cite{KimuraKuwahara2025}.
More recently, \citet{MobusEtAl2026} directly bound how much local thermal observables change under Hamiltonian perturbations in long-range systems under suitable decay conditions.
\section{Exact operational anchors}
\label{sec:exact-anchors}

We now turn our attention to physical Hamiltonians whose Gibbs states admit an exact purification-based preparation from a tractable core and a finite-depth tail, such that the core--tail identity \eqref{eq:intro-core-tail-identity} is satisfied with a vanishing residual, $R=0$.
We examine three representative classes of such Hamiltonians, including graph-stabilizer models, independent Rydberg dimers, and quadratic fermions.
The list is not exhaustive, but is chosen to exhibit three reusable anchor-design patterns and to supply starting points for more general Hamiltonians in Sec.~\ref{sec:local-reduction}.

\begin{definition}[Purification-based exact operational anchor]
  \label{def:exact-operational-anchor}
  A purification-based exact operational anchor for a finite-system Hamiltonian $H_0$ consists of the following components:
  \begin{enumerate}
      \renewcommand{\labelenumi}{(\roman{enumi})}
    \item A core Hamiltonian $H_{C,0}$ with an explicit ideal procedure for preparing a purification $\ket{\Psi_{C,0}}_{AS}$ of its Gibbs state $\sigma_{C,0}=\rhoG{H_{C,0}}$ at the temperatures considered.
    \item A unitary tail $U_0$ acting on the system register, with a specified circuit decomposition in the physical geometry.
    \item The exact core-frame Hamiltonian identity
      \begin{equation}
        U_0^\dagger H_0U_0=H_{C,0}
        \label{eq:exact-anchor-identity}
      \end{equation}
      holds.
  \end{enumerate}
  {The core preparation and tail implementation must have efficient classical and quantum resource scalings with system size for the anchor family under consideration.}
  After the tail is applied to $S$, the ancilla register $A$ is discarded, leaving the physical Gibbs state on $S$.
\end{definition}

These requirements distinguish an operational SCTA anchor from a Hamiltonian that is merely algebraically solvable.
{Furthermore, the tail's circuit depth and gate supports determine how local observables propagate backward and whether the locality conditions for certification hold.}
Table~\ref{tab:anchor-classes} summarizes the core purification, tail construction, and operational locality conditions for the three anchor classes.

\begin{table*}[t]
  \caption{Three purification-based exact operational anchor patterns, including the core type, the tail construction, and the conditions under which their Hamiltonian terms and tail circuits remain local in the physical qubit geometry.}
  \label{tab:anchor-classes}
  \centering
  \vspace{1em}
  \begingroup
  \renewcommand{\arraystretch}{1.22}
  \begin{tabular}{@{}p{0.15\textwidth}@{\hspace{0.03\textwidth}}p{0.25\textwidth}@{\hspace{0.03\textwidth}}p{0.21\textwidth}@{\hspace{0.03\textwidth}}p{0.255\textwidth}@{}}
    \toprule
    Anchor & Core purification & Tail & Operational locality condition \\
    \midrule
    Graph stabilizer
    & Product or finite-memory stabilizer-syndrome core
    & Hadamard--$CZ$ Clifford circuit $U_G$
    & Bounded graph degree and physical edge length; bounded-diameter stabilizer products \\
    Rydberg dimers
    & Independent four-label dimers
    & Parallel dimer basis-changing unitaries $\bigotimes_m u_m$
    & Uniformly bounded dimers and constant-size local compilation \\
    Quadratic fermions
    & Independent mode occupations
    & Gaussian mode transformation $U_F$
    & Volume-independent gate support and depth after fermion-to-qubit mapping \\
    \bottomrule
  \end{tabular}
  \endgroup
\end{table*}

\subsection{Graph stabilizers}
\label{sec:graph-anchor}

Graph-stabilizer Hamiltonians provide a canonical example in which an interacting physical Hamiltonian has an entangled eigenbasis generated by a shallow Clifford circuit.
Let $G=(V,E)$ be a finite simple graph with one qubit at each vertex $i\in V$ and edges $(i,j)\in E$ connecting adjacent vertices.
For each vertex $i$, define the graph-stabilizer generator~\cite{Gottesman1997,HeinEisertBriegel2004}
\begin{equation}
  K_i=X_i\prod_{j\in\mathcal N(i)}Z_j,
  \label{eq:graph-stabilizer-generator}
\end{equation}
where $\mathcal N(i)$ is the set of vertices adjacent to $i$.
These generators commute and have eigenvalues $\pm1$.
For a finite set $X\subseteq V$, we write $K_X=\prod_{i\in X}K_i$.
Let $\mathcal F$ contain sets whose cardinalities and graph diameters are uniformly bounded, and define the physical Hamiltonian
\begin{equation}
  H_{0,G}=-\sum_{X\in\mathcal F}\lambda_XK_X.
  \label{eq:graph-stabilizer-hamiltonian}
\end{equation}
Because all $K_i$'s commute, the energy eigenstates can be labeled by their simultaneous eigenvalues.
Writing the eigenvalue of $K_i$ as $(-1)^{z_i}$, each basis state is equivalently labeled by a binary stabilizer-syndrome string $z\in\{0,1\}^{|V|}$, where $z_i=0$ and $z_i=1$ correspond to eigenvalues $+1$ and $-1$, respectively.

The tail that converts these syndrome labels into the physical graph-state eigenbasis is
\begin{equation}
  U_G=
  \left(\prod_{(i,j)\in E}CZ_{ij}\right)
  \left(\prod_{i\in V}\Had_i\right),
  \label{eq:graph-tail}
\end{equation}
where $\Had_i$ is the Hadamard gate on qubit $i$ and $CZ_{ij}$ is the controlled-$Z$ gate on qubits $i$ and $j$.
The Clifford relations give $U_G^\dagger K_iU_G=Z_i$ and therefore $U_G^\dagger K_XU_G=Z_X$, with $Z_X=\prod_{i\in X}Z_i$.
Consequently, the diagonal syndrome core is
\begin{equation}
  H_{C,G}=-\sum_{X\in\mathcal F}\lambda_XZ_X,
  \label{eq:graph-anchor-identity}
\end{equation}
which satisfies the exact core--tail identity $U_G^\dagger H_{0,G}U_G=H_{C,G}$.
If the graph has bounded degree $\Delta$ and bounded physical edge length, the implementation of $U_G$ first applies one parallel Hadamard layer and then schedules the $CZ$ gates in at most $2\Delta-1$ layers by a greedy edge coloring.
Hence the ideal tail depth is at most $2\Delta$, independently of $|V|$.
The bounded graph diameter assumed for $X\in\mathcal F$, together with bounded edge length, also gives bounded physical support for every $K_X$.

In the simplest tractable case consisting of independent syndrome bits, i.e., when $\mathcal F$ is a set of singletons, the core Hamiltonian is a sum of single-qubit $Z$ terms: $H_{C,G}=-\sum_i\lambda_iZ_i$.
Its Gibbs state is therefore a product of independent single-qubit Gibbs states: $\rhoG{H_{C,G}}=\bigotimes_i\rhoG{-\lambda_iZ_i}$.
For this product core, the probability of label $z_i=1$ is given by
\begin{equation}
  p_{\beta,i}=\Pr(z_i=1)=\frac{1}{1+\ee^{2\beta\lambda_i}},
  \label{eq:graph-defect-probability}
\end{equation}
and all ancilla--system pairs in the purification are independent:
\begin{equation}
  \ket{\Psi_{C,G}}
  =\bigotimes_i
  \left(
    \sqrt{1-p_{\beta,i}}\ket{0}_{A_i}\ket{0}_{S_i}
    +\sqrt{p_{\beta,i}}\ket{1}_{A_i}\ket{1}_{S_i}
  \right).
  \label{eq:graph-product-purification}
\end{equation}
In an ideal continuous-rotation model, the purification loader uses one single-qubit rotation and one ancilla-to-system CNOT per site.
Since the loader for each site can be implemented in parallel, the total ideal depth is two, independent of system size.
Furthermore, the core entropy is the sum of $|V|$ binary entropies, one for each independent syndrome label.

The graph construction also supports interacting syndrome labels.
Label the vertices by $1,\ldots,N$, fix a maximum index span $m$ independently of $N$, and we can write the corresponding core Hamiltonian explicitly as
\begin{equation}
  \begin{split}
    H_{C,G}^{(m)}
    ={}&-
    \sum_{i=1}^{N}\lambda_i Z_i
    -\sum_{\substack{1\le i<j\le N\\j-i\le m}}
    \lambda_{ij}Z_iZ_j\\
    &-\sum_{\substack{1\le i<j<k\le N\\k-i\le m}}
    \lambda_{ijk}Z_iZ_jZ_k
    -\cdots.
  \end{split}
  \label{eq:finite-range-syndrome-core}
\end{equation}
The ellipsis denotes higher-order products $Z_{i_1}\cdots Z_{i_p}$ whose index span satisfies $i_p-i_1\le m$.
Replacing every $Z_i$ by $K_i$ gives the corresponding physical interacting stabilizer Hamiltonian, which is related to $H_{C,G}^{(m)}$ by the same Clifford tail $U_G$.
On a computational-basis string $z$, the eigenequation $Z_i\ket z=(-1)^{z_i}\ket z$ turns Eq.~\eqref{eq:finite-range-syndrome-core} into a classical finite-range energy function of the syndrome labels.
In particular, a term involving the next label $z_{i+1}$ can reach no farther into the preceding configuration than $z_{i-m+1}$.
After the unassigned suffix is summed over, the conditional probability of $z_{i+1}$ therefore depends on the preceding configuration only through the finite-memory string $s_i=(z_{i-m+1},\ldots,z_i)$.
The Gibbs distribution consequently factorizes as
\begin{equation}
  p_\beta(z_1,\ldots,z_N)
  =\pi_{\beta,m}(s_m)
  \prod_{i=m}^{N-1}T_{\beta,i}(z_{i+1}\mid s_i).
  \label{eq:finite-memory-core}
\end{equation}
Here $\pi_{\beta,m}$ is the normalized distribution of the first $m$ labels, and $T_{\beta,i}(z_{i+1}\mid s_i)$ is the conditional probability of the next label.
This factorization is a classical order-$m$ Markov description of the syndrome labels, and it is the basis for a still efficient loader and entropy evaluation.

For fixed $m$, a backward dynamic program obtains the normalized conditional probabilities and the entropy in $O(N2^m)$ classical time and memory.
Sequential conditional rotations prepare the exact core purification with $O(N2^m)$ gates in an ideal continuous-rotation model.
The gate count is therefore polynomial for a fixed $m$, but the depth is $O(N)$, unlike the constant depth of the independent-syndrome case.
It should also be noted that synthesis over a discrete gate set would have additional precision-dependent gate cost and core-preparation error.
The full recursion, core-loading circuit construction, entropy formula, and finite-precision accounting are given in SM Sec.~{\smref{sec:supp-graph-anchor}{S3.A}}.

\subsection{Independent Rydberg dimers}
\label{sec:rydberg-anchor}

The second anchor class is motivated by Rydberg-atom arrays, a prominent platform for neutral-atom quantum computing with demonstrations of reconfigurable connectivity, parallel high-fidelity entangling gates, and logical-qubit processing~\cite{BluvsteinEtAl2022,EveredEtAl2023,BluvsteinEtAl2024}.
The standard driven Rydberg-array Hamiltonian combines local laser terms with interactions between excited atoms~\cite{SaffmanWalkerMolmer2010,BrowaeysLahaye2020}.
For atom $i$, let $\ket{g}_i$ and $\ket{r}_i$ denote the ground and Rydberg states, respectively.
With $n_i=\ket{r}_i\!\bra{r}$ and $X_i=\ket{g}_i\!\bra{r}+\ket{r}_i\!\bra{g}$, its physical form is
\begin{equation}
  H_{\mathrm{Ry}}
  =\sum_i\left(\frac{\Omega_i}{2}X_i-\Delta_i n_i\right)
  +\sum_{i<j}V_{ij}n_in_j,
  \label{eq:rydberg-array-hamiltonian}
\end{equation}
where $\Omega_i$ and $\Delta_i$ are the Rabi frequency and detuning, and $V_{ij}$ is the interaction energy of a pair of Rydberg excitations.
To obtain an exact anchor, we partition the atoms into disjoint pairs, retain the interaction within each pair, and omit interactions between distinct dimers.
Assuming $M$ total dimers, each with equal drive and detuning, the resulting physical anchor is
\begin{equation}
  \begin{aligned}
    H_{0,\mathrm{Ry}}&=\sum_{m=1}^M h_m,\\
    h_m&=\frac{\Omega_m}{2}(X_{m,1}+X_{m,2})
    -\Delta_m(n_{m,1}+n_{m,2})\\
    &\hspace{1.5em}+V_mn_{m,1}n_{m,2}.
  \end{aligned}
  \label{eq:rydberg-dimer-anchor}
\end{equation}
Here $n_{m,a}=\ket{r}_{m,a}\!\bra{r}$ and $X_{m,a}=\ket{g}_{m,a}\!\bra{r}+\ket{r}_{m,a}\!\bra{g}$, with $m$ being the dimer index and $a=1,2$ labeling the two atoms in the dimer.
Because the terms $h_m$ act on disjoint dimers, the anchor Gibbs state factorizes as $\rhoG{H_{0,\mathrm{Ry}}}=\bigotimes_m\rhoG{h_m}$.
Each $h_m$ acts on a four-dimensional dimer Hilbert space and can therefore be diagonalized independently.
This requires only a fixed-size classical eigensolver and produces a constant-size two-qubit basis transformation.
Let $\{\ket{\phi_{m,\alpha}}\}_{\alpha=1}^4$ be an orthonormal physical eigenbasis satisfying $h_m\ket{\phi_{m,\alpha}}=\epsilon_{m,\alpha}\ket{\phi_{m,\alpha}}$, and let $\{\ket{\alpha}\}_{\alpha=1}^4$ denote the two-qubit computational basis used as the core label basis.
They are related by the two-qubit unitary $u_m\ket{\alpha}=\ket{\phi_{m,\alpha}}$.
Therefore, the complete core Hamiltonian and tail are
\begin{equation}
  \begin{aligned}
    H_{C,\mathrm{Ry}}&=\sum_{m,\alpha}\epsilon_{m,\alpha}\ket{\alpha}\!\bra{\alpha}_m,\\
    U_{\mathrm{Ry}}&=\bigotimes_m u_m,
  \end{aligned}
  \label{eq:rydberg-anchor-identity}
\end{equation}
where $U_{\mathrm{Ry}}^\dagger H_{0,\mathrm{Ry}}U_{\mathrm{Ry}}=H_{C,\mathrm{Ry}}$ is satisfied.
The core Gibbs state of block $m$ has weights $p_{\beta,m}(\alpha)=\ee^{-\beta\epsilon_{m,\alpha}}/\sum_\gamma\ee^{-\beta\epsilon_{m,\gamma}}$ and is loaded through a four-label purification.
The complete core purification is
\begin{equation}
  \ket{\Psi_{C,\mathrm{Ry}}}_{AS}
  =\bigotimes_m\left[
    \sum_{\alpha=1}^{4}\sqrt{p_{\beta,m}(\alpha)}
    \ket{\alpha}_{A_m}\ket{\alpha}_{S_m}
  \right].
  \label{eq:rydberg-core-purification}
\end{equation}
All block loaders run in parallel and, subsequently, all block tails run in parallel, with $O(M)$ total gates and volume-independent ideal depth.
An explicit binary-tree loader needs three $R_y$ rotations and four CNOTs per purified dimer.
Moreover, the total entropy is a sum of independent dimer entropies, each of which is the Shannon entropy of a four-outcome distribution.
Details of the dimer diagonalization, basis-change unitary, core-loading circuit, loader angles, and entropy formula are given in SM Sec.~{\smref{sec:supp-rydberg-anchor}{S3.B}}.

\subsection{Quadratic fermions}
\label{sec:quadratic-fermion-anchor}

Quadratic fermionic Hamiltonians admit an exact core--tail decomposition in which the core consists of independent fermionic modes and the tail is a Gaussian unitary that transforms the mode basis into the physical basis.
Consider a system of $N$ fermionic modes with creation and annihilation operators $f_j^\dagger$ and $f_j$, satisfying the canonical anticommutation relations $\{f_j,f_k\}=\{f_j^\dagger,f_k^\dagger\}=0$ and $\{f_j,f_k^\dagger\}=\delta_{jk}\id$.
For these physical fermionic operators, define the corresponding Majorana operators $\eta_{2j-1}=f_j+f_j^\dagger$ and $\eta_{2j}=-\mathrm{i}(f_j-f_j^\dagger)$.
After omitting an additive scalar, a parity-preserving quadratic Hamiltonian has the compact physical form
\begin{equation}
  H_{0,F}=\frac{\mathrm i}{4}\bm{\eta}^{\mathsf T}\mathsf A\bm{\eta},
  \qquad \mathsf A^{\mathsf T}=-\mathsf A,
  \label{eq:fermion-physical-hamiltonian}
\end{equation}
where $\mathsf A$ is a real antisymmetric matrix of size $2N\times 2N$.
Its skew-canonical form supplies $\mathsf O_F\in\mathrm{SO}(2N)$ and real single-mode energies $\epsilon_j$ such that~\cite{TerhalDiVincenzo2002,Bravyi2005FLO}
\begin{equation}
  \mathsf O_F^{\mathsf T}\mathsf A\mathsf O_F
  =\bigoplus_j
  \begin{pmatrix}0&\epsilon_j\\-\epsilon_j&0
  \end{pmatrix},
  \label{eq:fermion-canonical-form}
\end{equation}
which reduces the many-body diagonalization to finding the canonical form of the Majorana matrix.

To expose the core Hamiltonian, we introduce the core fermionic modes $c_j^\dagger$ and $c_j$, as well as core Majoranas $\gamma_{2j-1}=c_j+c_j^\dagger$ and $\gamma_{2j}=-\mathrm{i}(c_j-c_j^\dagger)$.
Accordingly, the core occupations are $n_{C,j}=c_j^\dagger c_j$.
To exhibit the basis transformation as a circuit, we choose an ordered Givens factorization
\begin{equation}
  \mathsf O_F
  =\prod_{\nu=1}^{N_{\mathrm{rot}}}
  \mathsf R_{p_\nu q_\nu}(\theta_\nu),
  \label{eq:fermion-givens-factorization}
\end{equation}
where $\mathsf R_{pq}(\theta)$ represents a planar rotation of Majoranas $p$ and $q$, and the number of rotations is bounded by the number of independent entries in $\mathsf A$: $N_{\mathrm{rot}}\le N(2N-1)$.
The corresponding Gaussian gate is
\begin{equation}
  G_{pq}(\theta)
  :=\exp\!\left(\frac{\theta}{2}\gamma_p\gamma_q\right),
  \label{eq:fermion-givens-gate}
\end{equation}
and, with the same ordering as in Eq.~\eqref{eq:fermion-givens-factorization}, the tail can be chosen explicitly as
\begin{equation}
  U_F
  =\prod_{\nu=1}^{N_{\mathrm{rot}}}
  G_{p_\nu q_\nu}(\theta_\nu).
  \label{eq:fermion-tail}
\end{equation}
This circuit implements $U_F^\dagger\bm\eta U_F=\mathsf O_F\bm\gamma$ and therefore gives the core Hamiltonian diagonal in the occupation basis:
\begin{equation}
  H_{C,F}
  =U_F^\dagger H_{0,F}U_F
  =\sum_j\epsilon_j\left(n_{C,j}-\frac12\id\right).
  \label{eq:fermion-anchor-identity}
\end{equation}

The core Hamiltonian is therefore a sum of independent single-mode energies, and the Gibbs state factorizes as $\rhoG{H_{C,F}}=\bigotimes_j\rhoG{\epsilon_j(n_{C,j}-\frac12\id)}$.
The corresponding occupation probability for mode $j$ is
\begin{equation}
  p_{\beta,j}=\frac{1}{1+\ee^{\beta\epsilon_j}},
  \label{eq:fermi-factor-main}
\end{equation}
which is the familiar Fermi-Dirac factor.
Consequently, it has a parallel product-purification loader and additive entropy, similar to the graph-stabilizer product core discussed in Sec.~\ref{sec:graph-anchor}.
To express this fermionic core as a qubit loader, however, we must first map the fermions to qubits.
We encode the two-dimensional occupation space of mode $c_j$ in system qubit $S_j$, identifying $\ket0_{S_j}$ with $n_{C,j}=0$ and $\ket1_{S_j}$ with $n_{C,j}=1$.
This identification is the occupation-basis part of the Jordan--Wigner mapping introduced below.
Ancilla qubit $A_j$ carries a correlated copy of the same binary occupation label.
The encoded core purification is therefore
\begin{equation}
  \begin{aligned}
    \ket{\Psi_{C,F}}_{AS}
    =\bigotimes_j\Bigl[{}
      &\sqrt{1-p_{\beta,j}}\ket{0}_{A_j}\ket{0}_{S_j}\\
      &+\sqrt{p_{\beta,j}}\ket{1}_{A_j}\ket{1}_{S_j}
    \Bigr].
  \end{aligned}
  \label{eq:fermion-core-purification}
\end{equation}
Tracing out the ancilla register leaves the fermionic core Gibbs state, represented in the occupation basis of the system-qubit register.

The occupation correspondence fixes the loader, but locality of the Gaussian tail still depends on the full operator map.
For this purpose, we adopt a one-dimensional Jordan--Wigner (JW) mapping of the core modes and define $\mathcal Z_{<j}:=\prod_{\ell<j}Z_\ell$, with $\mathcal Z_{<1}:=\id$~\cite{JordanWigner1928,BravyiKitaev2002}.
With $\ket0$ denoting an empty mode and $\ket1$ an occupied mode, the encoding gives
\begin{equation}
  c_j\longmapsto
  \mathcal Z_{<j}\frac{X_j+\mathrm{i}Y_j}{2},
  \qquad
  c_j^\dagger\longmapsto
  \mathcal Z_{<j}\frac{X_j-\mathrm{i}Y_j}{2},
  \label{eq:fermion-jw-complex}
\end{equation}
and hence,
\begin{equation}
  n_{C,j}\longmapsto\frac{\id-Z_j}{2}.
  \label{eq:fermion-jw-occupation}
\end{equation}
Equivalently, the Majorana operators become
\begin{equation}
  \gamma_{2j-1}\longmapsto\mathcal Z_{<j}X_j,
  \qquad
  \gamma_{2j}\longmapsto\mathcal Z_{<j}Y_j.
  \label{eq:fermion-jw-majorana}
\end{equation}
The parity string $\mathcal Z_{<j}$ records the total occupation parity of the preceding modes before $j$ and enforces the fermionic anticommutation relations.

Let $\Gamma_p$ denote the Pauli string for $\gamma_p$ and define the Hermitian Pauli string $P_{pq}:=\mathrm{i}\Gamma_p\Gamma_q$.
The elementary Gaussian gate in Eq.~\eqref{eq:fermion-givens-gate} is then encoded as
\begin{equation}
  G_{pq}(\theta)
  \longmapsto
  \exp\!\left(-\frac{\mathrm{i}\theta}{2}P_{pq}\right).
  \label{eq:fermion-jw-gate}
\end{equation}
If the two Majoranas belong to modes $j$ and $\ell$ with $j<\ell$, the common parity string to the left of $j$ cancels in $P_{pq}$, leaving support only on the qubit interval from $j$ through $\ell$.
Thus Gaussian gates between modes at bounded separation remain bounded-support qubit gates when the JW ordering is aligned with a one-dimensional physical geometry.
The same cancellation further ensures that any parity-even fermionic operator beyond bilinears on a bounded interval has bounded encoded support, whereas a parity-odd operator generally retains an uncancelled parity string.
In higher dimensions, or for an ordering unrelated to the physical geometry, even a physically local fermionic gate can instead acquire qubit support that grows with system size.
{Moreover, uniformly bounded gate supports alone do not bound the backward radius independently of system size.}
{A sufficient additional condition is that the encoded circuit can be scheduled into a volume-independent number of layers.}
{The $O(N^2)$ gate-count bound in Eq.~\eqref{eq:fermion-tail} does not directly provide this depth guarantee.}
{The locality conditions of Sec.~\ref{sec:architecture-certification} must therefore be checked in the encoded qubit geometry before the local certification framework can be applied.}
Sec.~{\smref{sec:supp-gaussian-anchor}{S3.C}} of the SM gives the full JW support analysis and an explicit realization by independent adjacent hopping dimers, whose tail has bounded support and depth in a one-dimensional geometry.
\section{Perturbative local reduction for weakly deformed Hamiltonians}
\label{sec:local-reduction}

The exact anchors in Sec.~\ref{sec:exact-anchors} separate a tractable thermal spectrum from a structured basis change, but a physical Hamiltonian will generally contain additional fields or couplings that move it outside an exactly solvable anchor class.
This section develops a perturbative reduction for such deformations through a Schrieffer--Wolff procedure.
In the core frame, the reduction retains deformation terms compatible with the loadable core family and removes core-incompatible terms through additional basis-changing corrections.

\subsection{Deforming an exact anchor}

Let $(H_0,H_{C,0},U_0)$ be an exact operational anchor in the sense of Definition~\ref{def:exact-operational-anchor}.
We deform this anchor by considering the physical family
\begin{equation}
  H_\lambda=H_0+\lambda V.
  \label{eq:nearby-family}
\end{equation}
Conjugation by the anchor tail gives the corresponding core-frame family
\begin{equation}
  U_0^\dagger H_\lambda U_0
  =H_{C,0}+\lambda V_C,
  \label{eq:nearby-core-frame-main}
\end{equation}
where $V_C:=U_0^\dagger V U_0$ is the core-frame deformation.
Here and throughout this section, $\lambda\in\mathbb R$ is a dimensionless scaling parameter while the fixed operator $V$ contains the physical structure and energy scale of the deformation.

If the anchor core and tail were left unchanged, the entire term $\lambda V_C$ would remain as a core-frame mismatch.
Its Hamiltonian residual is therefore first order in $\lambda$.
To improve on this baseline, our general strategy is to absorb into a corrected core the part of $V_C$ that preserves the chosen tractable family, remove as much as possible of the remaining basis-changing part by correcting the tail, and collect any uncanceled terms in a residual.
To formalize this separation, choose a real vector space $\mathfrak H_C$ of admissible core-Hamiltonian corrections and a linear projection $P_C$ onto this space, and define
\begin{equation}
  Q_C(O):=O-P_C(O)
  \label{eq:core-offcore-projectors-main}
\end{equation}
for every Hermitian operator $O$ in the declared domain.
Examples of admissible corrections in $\mathfrak H_C$ include operators diagonal in a product basis, finite-range classical interactions, and operators that are block diagonal on disjoint bounded-size blocks.
The decomposition
\begin{equation}
  V_C=P_C(V_C)+Q_C(V_C)
  \label{eq:core-offcore-split-main}
\end{equation}
then separates the part to be retained in the corrected core from the part to be canceled by a basis-changing correction or left in the residual.
It is worth noting that in the context of an operational SCTA, core compatibility requires that the corrected core must remain loadable and its Gibbs states must remain operationally accessible.

As such, we seek a correction unitary $W_\lambda$ and a loadable corrected core $H_{C,\lambda}$ satisfying
\begin{equation}
  W_\lambda^\dagger(H_{C,0}+\lambda V_C)W_\lambda
  =H_{C,\lambda}+R_\lambda,
  \label{eq:reduction-target}
\end{equation}
where $R_\lambda$ represents the mismatch that remains after the transformation.
Returning to the physical frame gives the corrected tail $U_\lambda:=U_0W_\lambda$ and hence
\begin{equation}
  U_\lambda^\dagger H_\lambda U_\lambda
  =H_{C,\lambda}+R_\lambda.
  \label{eq:corrected-core-tail-identity}
\end{equation}
Fig.~\ref{fig:core-relative-reduction} illustrates how the deformation is represented in the physical and core frames, and how the correction $W_\lambda$ separates the deformed core-frame Hamiltonian into a loadable corrected core and a residual.

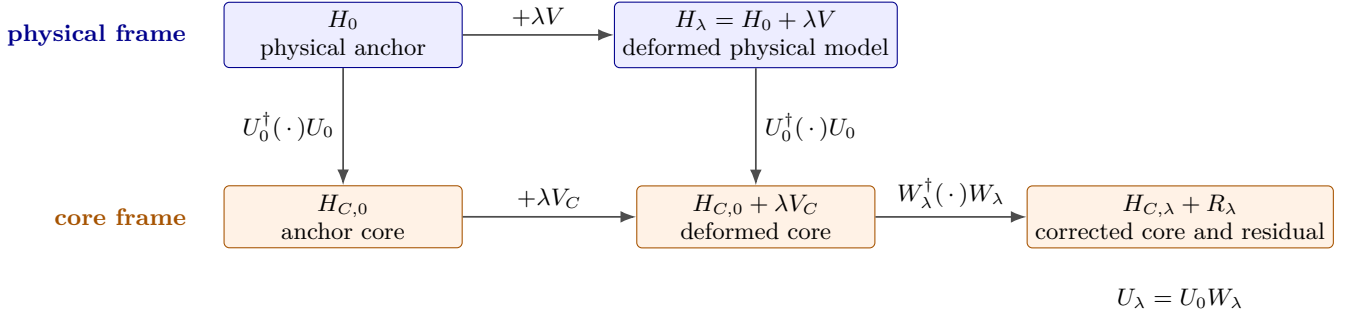
\begin{figure*}[t]
  \centering
  \begin{tikzpicture}[
      >=Latex,
      font=\small,
      node distance=1.55cm and 2.0cm,
      box/.style={draw=black!55,rounded corners=2pt,minimum width=3.15cm,minimum height=0.82cm,align=center,fill=white},
      physicalbox/.style={box,draw=blue!55!black,fill=blue!7},
      corebox/.style={box,draw=orange!65!black,fill=orange!10},
      framelabel/.style={font=\small\bfseries,anchor=east},
      arr/.style={->,line width=0.65pt,draw=black!70}
    ]
    \node[physicalbox] (h0) {$H_0$\\physical anchor};
    \node[physicalbox,right=of h0] (hl) {$H_{\lambda}=H_0+\lambda V$\\deformed physical model};
    \node[corebox,below=of h0] (hc) {$H_{C,0}$\\anchor core};
    \node[corebox,below=of hl] (hcv) {$H_{C,0}+\lambda V_C$\\deformed core};
    \node[corebox,right=of hcv] (red) {$H_{C,\lambda}+R_{\lambda}$\\corrected core and residual};
    \node[framelabel,text=blue!55!black] at ($(h0.west)+(-0.38cm,0)$) {physical frame};
    \node[framelabel,text=orange!65!black] at ($(hc.west)+(-0.38cm,0)$) {core frame};

    \draw[arr] (h0) -- node[above] {$+\lambda V$} (hl);
    \draw[arr] (h0) -- node[left,align=right] {$U_0^{\dagger}(\,\cdot\,)U_0$} (hc);
    \draw[arr] (hl) -- node[right,align=left] {$U_0^{\dagger}(\,\cdot\,)U_0$} (hcv);
    \draw[arr] (hc) -- node[above] {$+\lambda V_C$} (hcv);
    \draw[arr] (hcv) -- node[above] {$W_{\lambda}^{\dagger}(\,\cdot\,)W_{\lambda}$} (red);
    \node[below=0.42cm of red,align=center,font=\small] {$U_{\lambda}=U_0W_{\lambda}$};
  \end{tikzpicture}
  \caption{Physical- and core-frame representation of a deformation of an exact anchor.
    Adding $\lambda V$ deforms $H_0$ into $H_{\lambda}$, and conjugation by the anchor tail $U_0$ expresses the deformed Hamiltonian as $H_{C,0}+\lambda V_C$ in the core frame.
    The correction $W_{\lambda}$ is chosen so that the transformed core-frame Hamiltonian separates into the corrected loadable core $H_{C,\lambda}$ and the residual $R_{\lambda}$.
  Equivalently, the corrected physical tail is $U_{\lambda}=U_0W_{\lambda}$.}
  \label{fig:core-relative-reduction}
\end{figure*}

This reduction belongs broadly to the Schrieffer--Wolff (SW) transformation framework~\cite{SchriefferWolff1966,BravyiDiVincenzoLoss2011,Wegner1994,Kehrein2006,WurtzClaeysPolkovnikov2020}, but its finite-temperature role differs from that of the conventional low-energy construction.
Conventional SW theory begins with a selected low-energy subspace and constructs a perturbative unitary that suppresses its coupling to the high-energy sector, yielding an effective Hamiltonian within the retained subspace.
Here we assume no thermal truncation, since a finite-temperature Gibbs state can assign non-negligible weight throughout the spectrum.
The retained structure is instead a full-spectrum family of core Hamiltonians whose admissible corrections lie in $\mathfrak H_C$.
The resulting correction unitary rotates away operator directions that leave this family without discarding any spectral sector.

\subsection{Local transition channels and a scheduled correction}
\label{sec:local-transition-correction}

We first determine the correction algebraically up to first order in $\lambda$.
Let $W_\lambda^{(1)}=\exp(\lambda S_1)$ with $S_1^\dagger=-S_1$.
For any operator $K$, by the BCH formula, conjugation by this exponential gives
\begin{equation}
  \ee^{-\lambda S_1}K\ee^{\lambda S_1}
  =K+\lambda[K,S_1]+O(\lambda^2).
\end{equation}
Taking $K=H_{C,0}+\lambda V_C$, the first-order conjugation is
\begin{equation}
  \begin{aligned}
    &(W_\lambda^{(1)})^\dagger
    (H_{C,0}+\lambda V_C)W_\lambda^{(1)}\\
    &\qquad=H_{C,0}
    +\lambda\left(V_C+[H_{C,0},S_1]\right)+O(\lambda^2).
  \end{aligned}
  \label{eq:first-order-algebraic-expansion-main}
\end{equation}
Eliminating the off-core component linear in $\lambda$ requires the homological equation
\begin{equation}
  Q_C\!\left(V_C+[H_{C,0},S_1]\right)=0.
  \label{eq:first-order-homological-main}
\end{equation}
We use a sufficient, but stronger, condition
\begin{equation}
  [H_{C,0},S_1]=-Q_C(V_C),
  \label{eq:strong-homological-main}
\end{equation}
which leaves the first-order corrected core
\begin{equation}
  H_{C,\lambda}^{(1)}=H_{C,0}+\lambda P_C(V_C).
  \label{eq:first-order-corrected-core}
\end{equation}
This choice not only gives a clean separation between the corrected core and the residual, but also permits a transition-channel decomposition to its solution, as discussed below, that makes the locality and strength of the correction generator easier to control.
A more general solution of Eq.~\eqref{eq:first-order-homological-main} is possible, but any additional retained term would have to be incorporated into the corrected core and shown to preserve its loader.

One constructive solution of Eq.~\eqref{eq:strong-homological-main} decomposes the off-core deformation into local transition eigenoperators $A_\alpha$, satisfying $[H_{C,0},A_\alpha]=\omega_\alpha A_\alpha$ with $\omega_\alpha \in \mathbb{R}$:
\begin{equation}
  Q_C(V_C)=\sum_\alpha(A_\alpha+A_\alpha^\dagger)
  \label{eq:transition-channel-decomposition-main}
\end{equation}
Physically, each $A_\alpha$ represents a local process connecting eigenstates of $H_{C,0}$ and $\omega_\alpha$ is the corresponding core-energy difference between the initial and final states.
Although a transition-eigenoperator decomposition always exists at finite volume, it need not preserve locality: resolving a local deformation into exact many-body transition channels can produce operators with system-wide support.
{For a commuting finite-range core, however, a local term of $Q_C(V_C)$ can be resolved into transition channels using only the core interactions that overlap its support.}
{Each channel therefore acts only on the support of the local deformation term and the sites coupled to it by the core Hamiltonian.}
{The channel supports are uniformly bounded if all such enlarged regions have sizes and diameters bounded independently of system size.}
{This is the local Bohr-frequency construction used for commuting Hamiltonians~\cite{KastoryanoBrandao2016}.}
{This construction applies to the core structures considered in Sec.~\ref{sec:exact-anchors}, provided that the off-core deformation satisfies the locality conditions above.}
If the selected channels satisfy the active transition-gap condition
\begin{equation}
  |\omega_\alpha|\ge\gamma>0,
  \label{eq:active-transition-gap-main}
\end{equation}
then the anti-Hermitian generator
\begin{equation}
  S_1=\sum_\alpha\frac{A_\alpha^\dagger-A_\alpha}{\omega_\alpha}
  \label{eq:first-order-generator-main}
\end{equation}
solves Eq.~\eqref{eq:strong-homological-main}.
Indeed, $[H_{C,0},A_\alpha^\dagger]=-\omega_\alpha A_\alpha^\dagger$, so each summand produces $-(A_\alpha+A_\alpha^\dagger)$ under commutation with $H_{C,0}$.

A channel with $\omega_\alpha=0$ connects equal-energy core states and cannot be generated by a commutator with $H_{C,0}$, so it cannot be removed through Eq.~\eqref{eq:strong-homological-main}.
Such a resonant channel is not necessarily an obstruction and can be handled in a few ways.
One possibility is a local resonant basis reduction that changes the tail without modifying the core family.
{For example, suppose the anchor and resonant channel both decompose into disjoint bounded-size blocks, and the core family already permits block-diagonal Hamiltonians.}
{Each deformed block can then be diagonalized locally, with its eigenvalues retained in the block core and its basis change incorporated into the tail.}
A second possibility genuinely enlarges the core family while preserving its tractability.
For example, let $H_{C,0}=-\sum_i h_iZ_i$ and add the diagonal deformation $V_C=\sum_iJ_iZ_iZ_{i+1}$.
The deformation lies outside the initially admissible product core family but commutes with $H_{C,0}$, so it is a resonant component that cannot be canceled through Eq.~\eqref{eq:strong-homological-main}.
It can instead be absorbed into $H_{C,\lambda}=H_{C,0}+\lambda V_C$, enlarging the product core to the nearest-neighbor finite-memory core described in Sec.~\ref{sec:graph-anchor} while leaving the tail unchanged.
{If neither mechanism is applicable, the uncanceled resonant contribution must remain in the residual.}

{
  As a concrete example, we consider the product core $H_{C,0}=-\sum_jh_jZ_j$ and bond deformation $V_{C,i}=J_iX_iX_{i+1}$.
  Writing $\sigma_j^\pm=(X_j\pm\mathrm iY_j)/2$ splits the deformation into pair and exchange channels, which represent the simultaneous flipping of two spins and the exchange of two spins, respectively.
  The corresponding transition operators and energies are
  \begin{equation}
    \begin{aligned}
      A_{\mathrm p}&=J_i\sigma_i^+\sigma_{i+1}^+,
      &\omega_{\mathrm p}&=-2(h_i+h_{i+1}),\\
      A_{\mathrm e}&=J_i\sigma_i^+\sigma_{i+1}^-,
      &\omega_{\mathrm e}&=-2(h_i-h_{i+1}).
    \end{aligned}
    \label{eq:bond-transition-example-main}
  \end{equation}
  These channels satisfy $V_{C,i}=A_{\mathrm p}+A_{\mathrm p}^\dagger+A_{\mathrm e}+A_{\mathrm e}^\dagger$ and retain the original two-site support.
  When both transition energies are nonzero, Eq.~\eqref{eq:first-order-generator-main} gives a two-site anti-Hermitian generator $S_{1,i}$ satisfying $[H_{C,0},S_{1,i}]=-V_{C,i}$.
  The exchange channel becomes resonant at $h_i=h_{i+1}$, whereas the pair channel becomes resonant at $h_i=-h_{i+1}$.
  The full graph-stabilizer-chain construction is given in SM Sec.~\smref{sec:supp-graph-reduction}{S4.C}.
}

The constructions above determine an algebraic first-order generator $S_1$.
Implementing this generator as a bounded-depth local circuit additionally requires its terms to be organized according to their supports.
The channel generators in Eq.~\eqref{eq:first-order-generator-main} may overlap, so the corresponding gates cannot in general occupy the same parallel layer.
To schedule the application of these gates, we first combine all channels with the same support $Y$ into the gate generator
\begin{equation}
  S_Y
  :=
  \sum_{\alpha:\,\supp A_\alpha=Y}
  \frac{A_\alpha^\dagger-A_\alpha}{\omega_\alpha},
  \label{eq:support-generator-main}
\end{equation}
so that $S_1=\sum_YS_Y$.
Next, we form the support-overlap graph, whose vertices are the distinct supports $Y$ and whose edges join intersecting supports.
A proper coloring partitions these supports into classes $\mathcal C_1,\ldots,\mathcal C_m$ in which all supports are disjoint.
If every support contains at most $q_s$ sites and at most $n_s$ supports meet any site, then each support overlaps with at most $q_s(n_s-1)$ other supports.
A greedy coloring requires at most one more color than the maximum number of overlapping supports, giving $m\le q_s(n_s-1)+1$, independently of the total volume.
We then define the layer generators
\begin{equation}
  T_j=\sum_{Y\in\mathcal C_j}S_Y,
  \qquad
  j=1,\ldots,m
  \label{eq:layer-generators-main}
\end{equation}
and schedule the first-order correction as a product of parallel layers:
\begin{equation}
  W_{\lambda,\mathrm{loc}}^{(1)}
  =\ee^{\lambda T_1}\cdots\ee^{\lambda T_m}.
  \label{eq:scheduled-correction-main}
\end{equation}
Because the supports in one color class are disjoint, each $\ee^{\lambda T_j}=\prod_{Y\in\mathcal C_j}\ee^{\lambda S_Y}$ is exactly one parallel layer in the circuit.
In fact, the ordered product $W_{\lambda,\mathrm{loc}}^{(1)}$ is the first-order Lie--Trotter approximation to the algebraic correction $W_\lambda^{(1)}=\exp(\lambda S_1)$.
While they need not be equal at finite $\lambda$, expanding each layer unitary about $\lambda = 0$ gives
\begin{equation}
  \begin{aligned}
    W_{\lambda,\mathrm{loc}}^{(1)}
    &= \id + \lambda \sum_{j=1}^m T_j + O(\lambda^2) \\
    &= \id + \lambda S_1 + O(\lambda^2),
  \end{aligned}
\end{equation}
which is the same as the first-order expansion of $W_\lambda^{(1)}$.
Thus $W_{\lambda,\mathrm{loc}}^{(1)}$ implements the same first-order correction, while providing a bounded-depth local circuit structure that is more amenable to physical implementation.

\subsection{Volume-uniform first-order residual reduction}
\label{sec:first-order-residual-reduction}

We now pass from the first-order cancellation condition to the exact finite-$\lambda$ mismatch produced by the scheduled circuit.
We define the first-order residual as
\begin{equation}
  R_\lambda^{(1)}
  :=
  \left(W_{\lambda,\mathrm{loc}}^{(1)}\right)^\dagger
  (H_{C,0}+\lambda V_C)
  W_{\lambda,\mathrm{loc}}^{(1)}
  -H_{C,\lambda}^{(1)}.
  \label{eq:exact-first-order-residual}
\end{equation}
The first-order expansion of $W_{\lambda,\mathrm{loc}}^{(1)}$ established in Sec.~\ref{sec:local-transition-correction}, together with the strong homological equation~\eqref{eq:strong-homological-main}, ensures that the coefficient linear in $\lambda$ vanishes.
The remaining task is to show that the full residual admits a local decomposition whose incident strength, support size, and support diameter remain bounded independently of the total volume.

To state these bounds, give an operator $A$ a declared local decomposition $A=\sum_Xa_X$ and define
\begin{equation}
  \begin{aligned}
    \norm{A}_{\mathrm{loc}}
    &:=\sup_x\sum_{X\ni x}\norm{a_X}_\infty, \\
    q(A)&:=\sup_{a_X\ne0}|X|, \\
    a(A)&:=\sup_{a_X\ne0}\diam(X).
  \end{aligned}
  \label{eq:declared-local-data}
\end{equation}
These quantities depend on the declared decomposition.
The local decomposition norm $\norm{A}_{\mathrm{loc}}$ is an auxiliary bookkeeping quantity that, unlike the incident strength $\varepsilon_R$ in Eq.~\eqref{eq:incident-residual-strength}, does not remove the scalar part of each term.
After $R_\lambda^{(1)}$ is decomposed into Hermitian terms and centered, its incident strength therefore satisfies $\varepsilon_R(\lambda)\le\norm{R_\lambda^{(1)}}_{\mathrm{loc}}$.
The estimates in Sec.~\smref{sec:supp-local-bookkeeping}{S4.A} of the SM show how these local decomposition data change under commutators and conjugation by a disjoint circuit layer, which are the operations used to construct $R_\lambda^{(1)}$.
When the lattice varies along a sequence $\{\Lambda_N\}_{N\ge1}$, we call a family of declared decompositions uniformly local if all three quantities in Eq.~\eqref{eq:declared-local-data} remain bounded independently of $N$.
Similarly, a scheduled solution $S_1=\sum_Y S_Y=\sum_{j=1}^mT_j$ is uniformly local if there exist $N$-independent bounds $m$, $q_s$, $a_s$, and $s_0$ such that every $S_Y$ is supported on a set $Y$ with $|Y|\le q_s$ and $\diam(Y)\le a_s$, and $\norm{T_j}_{\mathrm{loc}}\le s_0$ for every layer.
{The bound on the layer generators controls the commutator estimates used to obtain a volume-independent residual coefficient.}
{For the channel construction in Eq.~\eqref{eq:support-generator-main}, a sufficient condition is}
\begin{equation*}
  {
    \norm{T_j}_{\mathrm{loc}}
    \le \sup_x\sum_{\alpha:\,x\in\supp A_\alpha}
    \frac{2\norm{A_\alpha}_\infty}{|\omega_\alpha|}
    \le s_0.
  }
\end{equation*}
{A uniform active gap in Eq.~\eqref{eq:active-transition-gap-main} together with uniformly bounded incident channel strength ensures this condition.}

\begin{theorem}[Volume-uniform quadratic residual]
  \label{thm:first-order-local-reduction}
  Fix $\lambda_0>0$ and consider a family of core-frame Hamiltonian pairs $(H_{C,0},V_C)$ on a sequence of finite metric lattices $\{\Lambda_N\}_{N\ge1}$ with $|\Lambda_N|\to\infty$.
  Suppose that
  \begin{enumerate}
      \renewcommand{\labelenumi}{(\roman{enumi})}
    \item $H_{C,0}$ and $V_C$ admit uniformly local declared decompositions;
    \item the strong homological equation~\eqref{eq:strong-homological-main} admits an anti-Hermitian uniformly local scheduled solution $S_1=\sum_YS_Y=\sum_{j=1}^mT_j$ in the sense defined above.
  \end{enumerate}
  Then $H_{C,\lambda}^{(1)}$ in Eq.~\eqref{eq:first-order-corrected-core} is uniformly local for $|\lambda|\le\lambda_0$.
  Moreover, the exact residual $R_\lambda^{(1)}$ in Eq.~\eqref{eq:exact-first-order-residual} admits a declared Hermitian decomposition in which every nonzero term has support cardinality at most $q_R$ and support diameter at most $a_R$, and after termwise centering its incident strength satisfies
  \begin{equation}
    \varepsilon_R(\lambda)
    \le C_{\mathrm{SW}}\lambda^2,
    \qquad |\lambda|\le\lambda_0.
    \label{eq:quadratic-residual-main}
  \end{equation}
  The constants $q_R$, $a_R$, and $C_{\mathrm{SW}}$ depend only on $\lambda_0$ and the uniform bounds in assumptions (i)--(ii), not on $N$.
\end{theorem}

\begin{proof}
  See SM Sec.~\smref{sec:supp-first-order-proof}{S4.B}.
\end{proof}

{Some remarks on the hierarchy of assumptions and conclusions are in order.}
{At fixed finite volume, an exact solution of the strong homological equation already gives a quadratic residual, but its bounds may grow with volume.}
{Theorem~\ref{thm:first-order-local-reduction} adds uniform local bounds on the Hamiltonians and the scheduled solution to make this reduction volume-uniform: the corrected core remains local, while the exact finite-$\lambda$ residual has bounded local support and centered incident strength of $O(\lambda^2)$.}
{To realize this reduction as an operational SCTA, the anchor must satisfy Definition~\ref{def:exact-operational-anchor}, and the corrected core and tail must retain the preparation and implementation resources required by the chosen SCTA mode.}
Under these additional conditions, $U_\lambda^{(1)}=U_0W_{\lambda,\mathrm{loc}}^{(1)}$ defines the corrected tail.
We refer to $H_{C,\lambda}^{(1)}$ and $U_\lambda^{(1)}$ together as the operational first-order construction, which satisfies the core--tail identity in Eq.~\eqref{eq:corrected-core-tail-identity}.

{Finally, for the ideal corrected preparation $\rho_\lambda^{(1)}:=U_\lambda^{(1)}\rhoG{H_{C,\lambda}^{(1)}}(U_\lambda^{(1)})^\dagger$, the local Gibbs-state error follows directly from Sec.~\ref{sec:architecture-certification}.}
{Suppose the path-averaged response envelope in Eq.~\eqref{eq:path-averaged-response-envelope-main} holds along $H_{C,\lambda}^{(1)}+sR_\lambda^{(1)}$, $s\in[0,1]$.}
{For each nonempty output region $A$, let $B=B_{U_\lambda^{(1)}}(A)$ be the corresponding backward causal region under the corrected tail $U_\lambda^{(1)}$.}
{Proposition~\ref{prop:backward-region-propagation} and Theorem~\ref{thm:shell-response-main}, together with Eq.~\eqref{eq:quadratic-residual-main}, then give}
\begin{equation}
  {
    \begin{aligned}
      &\norm{\bigl[\rho_\lambda^{(1)}-\rhoG{H_\lambda}\bigr]_A}_1\\
      &\quad\le\min\!\left\{2,\,
        \beta C_{\mathrm{SW}}\lambda^2
      \left[|B|+C_\beta(B)\Sigma_g(B;\Lambda)\right]\right\}.
    \end{aligned}
  }
  \label{eq:certified-first-order-main}
\end{equation}
{Taking the supremum over all such $A$ with $\diam(A)\le r$ gives a volume-uniform $O(\lambda^2)$ bound on $D_r$ at fixed $\beta$ and $r$, provided the factor in the square brackets remains bounded uniformly over these regions and system sizes.}

The deformed graph-stabilizer, Rydberg-dimer, and quadratic-fermion anchors are worked out as separate first-order reduction examples in Secs.~\smref{sec:supp-graph-reduction}{S4.C}--\smref{sec:supp-fermion-reduction}{S4.E} of the SM under their stated locality and active-gap assumptions.

\subsection{Higher-order recursion and its limitations}
\label{sec:higher-order-reduction}

We now examine how the core/off-core separation extends algebraically to higher powers of $\lambda$.
Choose a target order $k\ge1$.
At each step $\ell=1,\ldots,k$, the generators $S_1,\ldots,S_{\ell-1}$ determine every contribution to the transformed Hamiltonian at order $\lambda^\ell$ except the commutator $[H_{C,0},S_\ell]$.
The $P_C$ component of this order-$\ell$ term is absorbed into the corrected core, while setting its $Q_C$ component to zero gives a linear equation for $S_\ell$.
The following proposition states the fixed-volume algebraic recursion.

Define
\begin{equation}
  S_\lambda^{(k)}:=\sum_{j=1}^k\lambda^jS_j,
  \qquad
  W_\lambda^{(k)}:=\exp(S_\lambda^{(k)}),
  \label{eq:higher-order-generator-main}
\end{equation}
where every $S_j$ is anti-Hermitian.
We retain generators only through the target order $k$, but evaluate the exponential $W_\lambda^{(k)}$ exactly.
The transformed Hamiltonian therefore contains powers beyond $\lambda^k$ and, at fixed finite volume, has the Taylor expansion
\begin{equation}
  (W_\lambda^{(k)})^\dagger
  (H_{C,0}+\lambda V_C)W_\lambda^{(k)}
  =H_{C,0}+\sum_{\ell\ge1}\lambda^\ell G_\ell.
  \label{eq:higher-order-expansion-main}
\end{equation}

\begin{proposition}[Higher-order core--tail recursion]
  \label{prop:formal-higher-order-reduction}
  Fix a finite volume and an integer $k\ge1$.
  For every $1\le \ell\le k$, the coefficient in Eq.~\eqref{eq:higher-order-expansion-main} has the form
  \begin{equation}
    G_\ell=\Xi_\ell+[H_{C,0},S_\ell],
    \label{eq:higher-order-coefficient-main}
  \end{equation}
  where $\Xi_\ell$ depends only on $H_{C,0}$, $V_C$, and $S_1,\ldots,S_{\ell-1}$.
  Suppose that anti-Hermitian generators in the declared generator space can be chosen successively so that
  \begin{equation}
    Q_C\!\left(\Xi_\ell+[H_{C,0},S_\ell]\right)=0
    \label{eq:nth-homological-main}
  \end{equation}
  for every $1\le \ell\le k$, and define
  \begin{equation}
    \Delta H_C^{(\ell)}
    :=P_C\!\left(\Xi_\ell+[H_{C,0},S_\ell]\right).
    \label{eq:nth-core-correction-main}
  \end{equation}
  Then every $\Delta H_C^{(\ell)}$ belongs to $\mathfrak H_C$, and
  \begin{equation}
    \begin{aligned}
      &(W_\lambda^{(k)})^\dagger
      (H_{C,0}+\lambda V_C)W_\lambda^{(k)}\\
      &\qquad=H_{C,0}+\sum_{\ell=1}^k\lambda^\ell\Delta H_C^{(\ell)}
      +O(\lambda^{k+1})
    \end{aligned}
    \label{eq:higher-order-normal-form-main}
  \end{equation}
  in operator norm as $\lambda\to0$ at that fixed volume.
\end{proposition}

\begin{proof}
  See SM Sec.~\smref{sec:supp-higher-order-proof}{S4.F}.
\end{proof}

For example, let us consider the first two orders of the recursion, whose coefficients are
\begin{equation}
  \begin{aligned}
    G_1={}&V_C+[H_{C,0},S_1],\\
    G_2={}&[H_{C,0},S_2]+[V_C,S_1]
    +\frac12[[H_{C,0},S_1],S_1].
  \end{aligned}
  \label{eq:first-two-higher-order-coefficients-main}
\end{equation}
The first equation determines $S_1$ from the original deformation.
The second equation shows explicitly that the first rotation generates new retained and off-core terms at order $\lambda^2$.
The former are retained in $\Delta H_{C}^{(2)}$, whereas the latter must be canceled by choosing $S_2$ to satisfy the second-order homological equation.
{However, we must note that in the transition-channel construction, new resonant channels may appear at higher orders, so first-order solvability does not guarantee that the recursion remains solvable.}
{Therefore, the recursion is conditional on the existence of a solution at each order.}

For any fixed $k$, define the corrected core,
\begin{equation}
  H_{C,\lambda}^{(k)}
  :=H_{C,0}+\sum_{\ell=1}^k\lambda^\ell\Delta H_C^{(\ell)},
\end{equation}
and the exact residual,
\begin{equation}
  R_{\lambda, \text{alg}}^{(k)}
  :=(W_\lambda^{(k)})^\dagger
  (H_{C,0}+\lambda V_C)W_\lambda^{(k)}
  -H_{C,\lambda}^{(k)}.
\end{equation}
Together with the anchor tail $U_0$, these definitions give an exact core--tail identity for every finite $\lambda$, while Proposition~\ref{prop:formal-higher-order-reduction} gives $\norm{R_{\lambda, \text{alg}}^{(k)}}_\infty=O(|\lambda|^{k+1})$ at fixed volume.
Unlike the first-order construction, the proposition does not establish a volume-uniform local residual or efficiently implementable corrected core and tail.
Each of these properties must be proved separately before the formal recursion defines an operational higher-order SCTA.
\section{Numerical demonstration}
\label{sec:numerics}

Here we test the local reduction mechanism of Sec.~\ref{sec:local-reduction} using the graph-stabilizer anchor of Sec.~\ref{sec:graph-anchor} on an $N$-qubit open chain.
The physical anchor Hamiltonian is
\begin{equation}
  H_0=-\sum_{i=1}^{N}h_iK_i,
  \label{eq:numerical-physical-anchor}
\end{equation}
where $K_i=Z_{i-1}X_iZ_{i+1}$, with $K_1=X_1Z_2$ and $K_N=Z_{N-1}X_N$ at the open boundaries.
We deform this anchor by the nearest-neighbor Ising interaction $V=J\sum_{i=1}^{N-1}Z_iZ_{i+1}$.
The deformed physical target is thus $H_{\lambda}=H_0+\lambda V$.
The graph-state tail $U_G$ in Eq.~\eqref{eq:graph-tail} maps these terms according to $U_G^\dagger K_iU_G=Z_i$ and $U_G^\dagger Z_iZ_{i+1}U_G=X_iX_{i+1}$.
The core-frame target is thus
\begin{equation}
  U_G^\dagger H_{\lambda}U_G
  =-\sum_{i=1}^{N}h_iZ_i
  +\lambda J\sum_{i=1}^{N-1}X_iX_{i+1}.
  \label{eq:numerical-core-hamiltonian}
\end{equation}
The first term on the right-hand side is $H_{C,0}$, and the second term gives the core-frame deformation $V_C = J\sum_{i=1}^{N-1}X_iX_{i+1}$.
For simplicity, we choose $J=1$ and a staggered field pattern for the anchor.
The staggered fields are
\begin{equation}
  h_i=\bar h+(-1)^i\delta,
  \label{eq:numerical-staggered-fields}
\end{equation}
with $\bar h=1.25$ and, unless stated otherwise, $\delta=0.25$.
The corresponding transition operators and energies are given in Eq.~\eqref{eq:bond-transition-example-main}.
With the chosen parameters, the energies are $-4\bar h$ and $-4(-1)^i\delta$ for the pair and exchange channels, respectively.
At these parameters, the smaller exchange energy makes the exchange channel the limiting one for the first-order construction.
We also note that the model in Eq.~\eqref{eq:numerical-core-hamiltonian} is an inhomogeneous transverse-field Ising chain after a local basis rotation and is therefore free-fermionic~\cite{Pfeuty1970}.
It consequently admits an exact algebraic representation using the occupation-number core and Gaussian tail introduced in Sec.~\ref{sec:quadratic-fermion-anchor}.
Here, however, we retain the graph-stabilizer anchor and apply the transition-generated correction because our purpose is to test the local reduction mechanism.

Leaving the anchor core and tail unchanged gives the bare residual
\begin{equation}
  R^{(0)}_\lambda=\lambda V_C.
  \label{eq:numerical-bare-residual}
\end{equation}
The scheduled first-order correction instead gives
\begin{equation}
  R^{(1)}_\lambda=
  W_{\lambda,\mathrm{loc}}^{(1)\dagger}
  \bigl(H_{C,0}+\lambda V_C\bigr)W_{\lambda,\mathrm{loc}}^{(1)}-H_{C,0}.
  \label{eq:numerical-corrected-residual}
\end{equation}
Here $W_{\lambda,\mathrm{loc}}^{(1)}=\exp(\lambda S_{\rm odd})\exp(\lambda S_{\rm even})$ is the exact scheduled two-layer circuit derived in Sec.~\smref{sec:supp-graph-reduction}{S4.C} of the SM.
Thus, $R^{(1)}_\lambda$ includes all higher-order terms generated by the finite rotation $W_{\lambda,\mathrm{loc}}^{(1)}$.

For the Pauli expansion of the residual $R=\sum_{P\ne\id}c_PP$, we evaluate the Pauli incident strength:
\begin{equation}
  \varepsilon_P(R)
  :=\max_i\sum_{P:\,i\in\supp(P)}|c_P|.
  \label{eq:pauli-certificate}
\end{equation}
It sums the absolute coefficients of all Pauli strings acting on a given site and then maximizes over sites.
To compare it with the centered incident strength $\varepsilon_R$ in Eq.~\eqref{eq:incident-residual-strength}, we group all Pauli strings with the same support $X$ into one residual term $r_X=\sum_{P:\,\supp(P)=X}c_PP$.
The centered half-width of this term satisfies $J_X\le\norm{r_X}_\infty\le\sum_{P:\,\supp(P)=X}|c_P|$, so summing over the supports incident on each site gives $\varepsilon_R\le\varepsilon_P (R)$.
Thus $\varepsilon_P$ bounds $\varepsilon_R$ for the equal-support grouped decomposition and is computable directly from the Pauli coefficients.

We also compare the exact target Gibbs state with the bare and first-order corrected preparations in the physical frame.
For this open chain, the local-state metric in Eq.~\eqref{eq:local-trace-distance} is the largest trace distance between corresponding nearest-neighbor reduced states:
\begin{equation}
  D_1(\rho,\sigma)
  =\max_{1\le i\le N-1}
  \norm{\rho_{i,i+1}-\sigma_{i,i+1}}_1.
  \label{eq:numerical-local-trace-distance}
\end{equation}
We also track the bond-correlation error
\begin{equation}
  \Delta_{ZZ}(\rho,\sigma)
  :=\max_{1\le i\le N-1}\abs{\Tr\!\left[Z_iZ_{i+1}(\rho-\sigma)\right]},
  \label{eq:numerical-zz-error}
\end{equation}
which measures the worst nearest-neighbor $ZZ$ correlation error.
All state comparisons reported here are deterministic endpoint calculations and do not verify a response bound along the interpolating Gibbs path.
{The residual construction and finite-size Gibbs-state calculations are detailed in Sec.~\smref{sec:supp-pauli-propagation}{S5.A} of the SM.}

\subsection{Residual and local accuracy scaling}

We first test whether the scheduled correction removes the leading Hamiltonian mismatch and whether the resulting local residual diagnostic remains stable as the chain grows.
At $N=12$, we evaluate the bare and corrected Pauli incident strengths at ten logarithmically spaced couplings over $10^{-3}\le\lambda\lesssim 0.178$ and fit each dataset to $\varepsilon_P=C|\lambda|^p$.
The exponent $p$ distinguishes the linear bare residual from the expected quadratic corrected residual, while a separate scan at $\lambda=0.05$ tests the dependence on system size.
These coupling- and size-scaling results are shown in Figs.~\ref{fig:residual-scaling}.

\begin{figure}[t!]
  \centering
  \includegraphics[width=0.9\columnwidth]{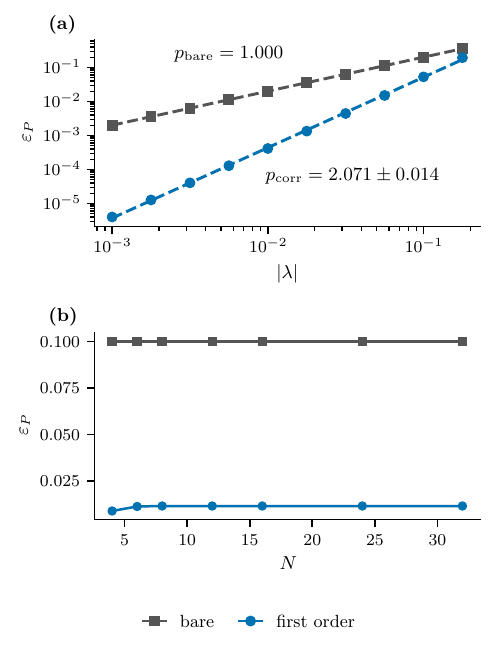}
  \caption{Local residual scaling for the graph-stabilizer benchmark.
    (a) Pauli incident strength $\varepsilon_P$ of the bare and first-order residuals versus $\abs{\lambda}$ at $N=12$.
    Markers represent the computed values, dashed lines are power-law fits over the complete sampled range, and the adjacent labels give the fitted exponents, including one-standard-error uncertainties.
    (b) Incident strength versus system size at fixed $\lambda=0.05$, showing that the first-order result saturates by $N=12$ in the tested size sequence.
  The bare and first-order constructions are indicated by gray squares and blue circles, respectively.}
  \label{fig:residual-scaling}
\end{figure}

As shown in Fig.~\ref{fig:residual-scaling}(a), the bare fit gives $p_{\rm bare}=1.000$ with a formal standard error below $10^{-15}$.
In comparison, the corrected fit gives $p_{\rm corr}=2.071\pm0.014$ with $R_{\rm fit}^2=0.9996$.
This demonstrates the expected quadratic behavior of the corrected residual.
The exact scheduled circuit therefore removes the linear term in the Pauli incident strength, in agreement with Theorem~\ref{thm:first-order-local-reduction}.
Fig.~\ref{fig:residual-scaling}(b) further shows that, at $\lambda=0.05$, the corrected Pauli incident strength reaches $0.0116$ by $N=12$ and is unchanged at the displayed precision through $N=32$.
The observed saturation is therefore a useful consistency check for volume-independent behavior of the local reduction mechanism.

Next, to test whether the Hamiltonian-level reduction is reflected in the output Gibbs state, we compare the bare and corrected states with the exact target at $N=8$, nine couplings, and three inverse temperatures $\beta\in\{0.5,1,2\}$.
Fig.~\ref{fig:thermal-accuracy}(a) reports the worst-case local-state error $D_1$, while Fig.~\ref{fig:thermal-accuracy}(b) reports the bond-correlation error $\Delta_{ZZ}$.
For each metric, the data at all three temperatures are fitted simultaneously to $y_\beta(\lambda)=A_\beta|\lambda|^p$.
The exponent $p$ is shared across temperatures, while each $\beta$ has an independent prefactor $A_\beta$ that captures the temperature dependence of the error magnitude.

\begin{figure}[t!]
  \centering
  \includegraphics[width=0.9\linewidth]{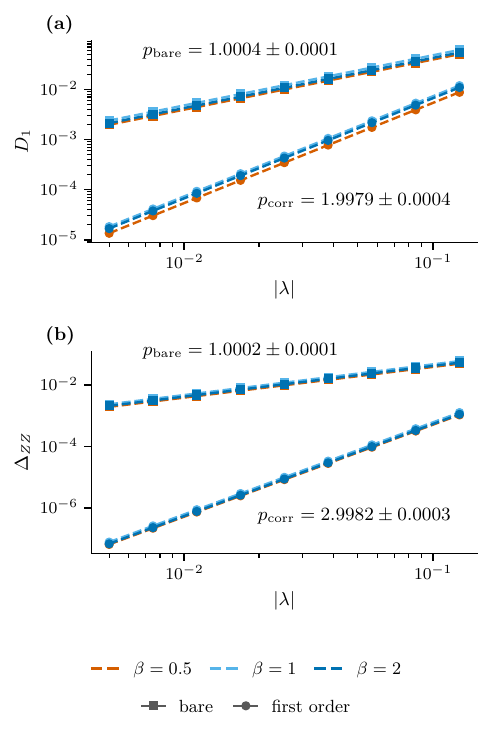}
  \captionof{figure}{Physical local Gibbs-state accuracy at $N=8$.
    (a) Maximum trace-norm distance $D_1$ between contiguous two-site marginals of the exact Gibbs state and those of the bare or first-order state, versus $\abs{\lambda}$ for three inverse temperatures $\beta\in\{0.5,1,2\}$.
    (b) Maximum error $\Delta_{ZZ}$ in the corresponding $Z_iZ_{i+1}$ expectation values.
    In both panels, markers are the computed values and dashed curves are grouped power-law fits over all nine sampled couplings.
    Each grouped fit uses a common exponent and a temperature-dependent prefactor, and the adjacent labels give the common exponent with its one-standard-error uncertainty.
  The common legends encode inverse temperature by color---orange for $\beta=0.5$, light blue for $\beta=1$, and dark blue for $\beta=2$---and distinguish the bare and first-order constructions by squares and circles, respectively.}
  \label{fig:thermal-accuracy}
\end{figure}

In Fig.~\ref{fig:thermal-accuracy}(a), the grouped fits over $0.005\le\lambda\lesssim 0.128$ give $p_{\rm bare}=1.0004\pm0.0001$ and $p_{\rm corr}=1.9979\pm0.0004$, confirming that the leading-order error in the local trace distance is removed by the first-order correction.
Fig.~\ref{fig:thermal-accuracy}(b) gives the corresponding $\Delta_{ZZ}$ exponents $1.0002\pm0.0001$ and $2.9982\pm0.0003$.
The near-cubic corrected scaling in Fig.~\ref{fig:thermal-accuracy}(b) follows from a sublattice symmetry that makes both the exact and corrected bond expectations odd functions of $\lambda$.
After first-order cancellation removes their linear difference, the next symmetry-allowed contribution is cubic.
This behavior is specific to the observable and model and does not imply a generic third-order local-state guarantee.
The fitting protocols and the symmetry argument are detailed in Secs.~{\smref{sec:supp-fit-protocol}{S5.B}} and {\smref{sec:supp-ZZ-symmetry}{S5.C}} of the SM.

\subsection{Breakdown of the first-order prescription and variational continuation}

With the pair and exchange eigenoperators $A_{i,\mathrm p}=\sigma_i^+\sigma_{i+1}^+$ and $A_{i,\mathrm e}=\sigma_i^+\sigma_{i+1}^-$, we define the gate family for each bond $i$ as
\begin{equation}
  \begin{split}
    W_i(\theta_{\mathrm p},\theta_{\mathrm e})
    =\exp\!\big[&
      \theta_{\mathrm p}(A_{i,\mathrm p}-A_{i,\mathrm p}^{\dagger})
      \\
      &+(-1)^i\theta_{\mathrm e}(A_{i,\mathrm e}-A_{i,\mathrm e}^{\dagger})
    \big].
  \end{split}
  \label{eq:numerical-variational-bond-gate}
\end{equation}
For $\delta\ne0$, applying Eq.~\eqref{eq:first-order-generator-main} to the pair and exchange transition energies stated above fixes $\theta_{\mathrm p}^{(1)}=\lambda J/(4\bar h)$ and $\theta_{\mathrm e}^{(1)}=\lambda J/(4\delta)$.
The prescribed first-order correction circuit $W_{\lambda,\mathrm{loc}}^{(1)}$ then applies $W_i(\theta_{\mathrm p}^{(1)},\theta_{\mathrm e}^{(1)})$ in successive odd- and even-bond layers.

The angle assignments in the scheduled circuit thus reveal two sources of breakdown: small transition energies and large deformation.
First, as $\delta\to0$, the exchange-transition energy vanishes and the corresponding first-order angle $\theta_{\mathrm e}^{(1)}$ diverges, even though the physical Hamiltonian remains finite.
For small but nonzero $\delta$, the circuit is still well defined, but the growing exchange angle signals a progressive loss of perturbative control before exact resonance is reached.
The second limitation occurs away from resonance: at fixed $\delta$, increasing $\abs{\lambda}$ enlarges both correction angles until the higher-order terms generated by the finite rotations are no longer parametrically suppressed.

\begin{figure}[t!]
  \centering
  \includegraphics[width=0.9\linewidth]{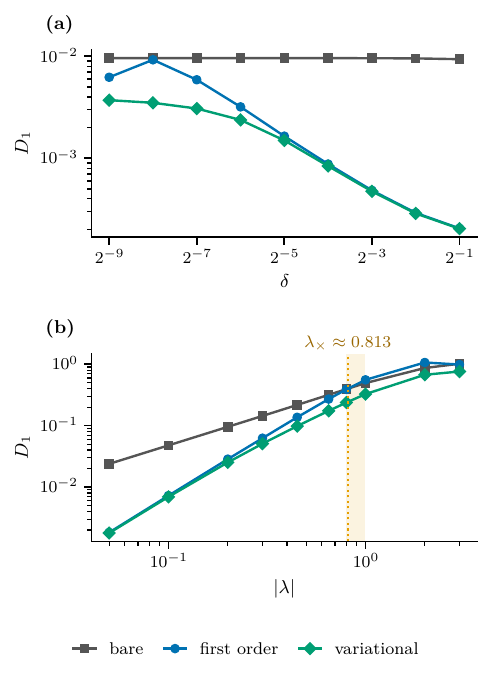}
  \captionof{figure}{Variational performance in the two regimes where the first-order construction loses perturbative control, evaluated at $N=8$ and $\beta=1$.
    (a) Maximum nearest-neighbor trace-norm error $D_1$ versus the field staggering $\delta$ at fixed $\lambda=0.02$, with resonance approached toward the left.
    (b) The same error versus $\abs{\lambda}$ at fixed $\delta=0.25$.
    The orange band marks the first sampled bracket $0.8\le\lambda\le1$ in which the first-order corrected error becomes larger than the bare error, and the dotted line marks $\lambda_{\times}\simeq0.813$, obtained by linear interpolation in $\ln\lambda$ within that bracket.
  In both panels, the bare, first-order corrected, and variational constructions are indicated by gray squares, blue circles, and green diamonds, respectively.}
  \label{fig:variational-failure}
\end{figure}

We probe these mechanisms at $N=8$ and $\beta=1$ by isolating the behavior of the prescribed first-order construction.
The details of the numerical protocol are provided in Sec.~{\smref{sec:supp-optimization-details}{S5.E}} of the SM.
Firstly, we scan the staggering $\delta$ at fixed $\lambda=0.02$ to approach resonance.
In the resonance scan of Fig.~\ref{fig:variational-failure}(a), the first-order error remains below the bare error at every displayed positive $\delta$, so the loss of perturbative control at the chosen parameters does not immediately produce a local-error crossover.
Nevertheless, the first-order error generally increases as $\delta$ decreases, despite the nonmonotonicity observed at the two smallest staggerings.
Its improvement over the bare Gibbs error generally increases as the system is moved farther from resonance.
At exact resonance, the exchange denominator vanishes and the first-order correction circuit must omit that channel.
This separate pair-only comparison is reported in Sec.~{\smref{sec:supp-optimization-details}{S5.E}} of the SM.

Secondly, we scan the deformation parameter $\lambda$ in the range $0.05\le\lambda\le 3$ at fixed $\delta=0.25$ to probe the finite-deformation limitation.
In Fig.~\ref{fig:variational-failure}(b), the first-order error crosses the bare error within $0.8\le\lambda\le1$, with the interpolated crossing $\lambda_\times\simeq0.813$.
Beyond this point, fixing the circuit angles at their first-order values can be worse than leaving the anchor state uncorrected.
The first-order error falls slightly below the bare error again at $\lambda=3$.
However, this re-entry is not a recovery of perturbative control, but rather a reflection of the nonmonotonic finite-angle behavior of the exactly exponentiated circuit.

Next, we test whether the transition-generated circuit can be variationally optimized to extend its usefulness beyond the perturbative regime.
We therefore retain the bond-gate family in Eq.~\eqref{eq:numerical-variational-bond-gate} and its odd--even schedule while treating $\theta_{\mathrm p}$ and $\theta_{\mathrm e}$ as free parameters to be optimized together with the core populations.
We denote the resulting two-layer correction by $W(\theta_{\mathrm p},\theta_{\mathrm e})$.
The core populations are varied through two loader angles, one for each field sublattice.
With $s(i)=\mathrm o$ on odd sites and $s(i)=\mathrm e$ on even sites, the reduced core state is
\begin{equation}
  \rho_C(\alpha_{\mathrm o},\alpha_{\mathrm e})
  =\bigotimes_{i=1}^{N}
  \left[
    \cos^2\!\alpha_{s(i)}\ket0\!\bra0
    +\sin^2\!\alpha_{s(i)}\ket1\!\bra1
  \right].
  \label{eq:numerical-variational-core}
\end{equation}
The two-sublattice restriction follows the same alternating pattern as the physical fields in Eq.~\eqref{eq:numerical-staggered-fields}.
The four parameters in our setup are therefore $\boldsymbol\vartheta=(\theta_{\mathrm p},\theta_{\mathrm e},\alpha_{\mathrm o},\alpha_{\mathrm e})$, and the corresponding physical trial state is
\begin{equation}
  \rho_{\boldsymbol\vartheta}
  =U_GW(\theta_{\mathrm p},\theta_{\mathrm e})
  \rho_C(\alpha_{\mathrm o},\alpha_{\mathrm e})
  W(\theta_{\mathrm p},\theta_{\mathrm e})^{\dagger}U_G^{\dagger}.
  \label{eq:numerical-variational-state}
\end{equation}
The variational ansatz structure is shown in Fig.~\smref{fig:supp-variational-circuit}{S3} in the SM.

We determine $\boldsymbol\vartheta$ by minimizing the physical free-energy density
\begin{equation}
  f_\beta(\boldsymbol\vartheta)
  =\frac{1}{N}
  \left[
    \Tr(\rho_{\boldsymbol\vartheta}H_\lambda)
    -\beta^{-1}S(\rho_{\boldsymbol\vartheta})
  \right].
  \label{eq:numerical-variational-free-energy}
\end{equation}
The energy is evaluated with the physical target Hamiltonian, while the entropy is obtained directly from the product probabilities in Eq.~\eqref{eq:numerical-variational-core} because the correction and graph tail are unitary and therefore preserve the spectrum of the core state.
At every scan point, we use ten L-BFGS-B runs initialized from the bare parameters, the nonsingular first-order parameters, and eight random points.
The entropy formula and details of the variational circuit and its optimization are given in Secs.~{\smref{sec:supp-variational-protocol}{S5.D}} and {\smref{sec:supp-optimization-details}{S5.E}} of the SM.

All $200$ optimization runs reached convergence and here we report the best result at each parameter point.
In the positive-$\delta$ resonance scan of Fig.~\ref{fig:variational-failure}(a), the variational state has the smallest $D_1$ among the three displayed constructions at all nine points.
Its advantage over the first-order state is small away from resonance and grows as the prescribed exchange angle becomes large, showing that the transition-generated directions remain useful after their perturbative coefficients lose control.
At exact resonance, however, the variational error is $3.93\times10^{-3}$, above the pair-only first-order value $2.81\times10^{-3}$ (not shown in the figure), so the improvement is not uniform at the singular point.

In the finite-deformation scan of Fig.~\ref{fig:variational-failure}(b), the variational state also has the smallest $D_1$ at all ten sampled couplings.
At $\lambda=1$, its error is the lowest among the three constructions.
At $\lambda=3$, the variational error reaches $0.763$ despite remaining below the bare and first-order values, $1.020$ and $0.990$, respectively.
Thus optimizing the transition-generated circuit extends its usefulness beyond the controlled first-order regime, but still fails to produce a high-quality approximation at large deformation.
\section{Discussion and outlook}
\label{sec:discussion}

In this work, we have focused on the SCTA for Gibbs-state preparation.
Building on this framework, we developed a rigorous local error analysis that bounds the error in local thermal observable predictions from the Hamiltonian mismatch between the target and the reference model specified by the core--tail construction.
We identified exact anchor Hamiltonians that admit tractable core--tail pairs that yield zero residuals.
For weakly deformed Hamiltonians around these anchors, we developed a systematic Schrieffer--Wolff reduction that absorbs core-compatible terms into the tractable core and cancels the leading basis-changing terms through a correction to the tail.
{Under the locality and solvability assumptions of Theorem~\ref{thm:first-order-local-reduction}, the first-order correction leaves an exact residual with volume-independent local support bounds and incident strength $O(\lambda^2)$.}
{Together with the response and geometry conditions of Sec.~\ref{sec:architecture-certification}, this yields a volume-uniform local Gibbs-state error bound for the ideal corrected preparation.}

{The construction and its certificate have separate limitations.}
{The reduction requires a uniformly local scheduled solution of the homological equation.}
{In the transition-channel construction, small transition energies or resonances may obstruct such a solution within the chosen core family.}
{The higher-order recursion remains a formal, fixed-volume result conditional on solvability at every order.}
{On the certification side, a small incident residual alone does not certify local thermal accuracy: the Kubo--Mori-based derivation also requires a controlled path-averaged response, including response decay and shell summability conditions.}
{Near criticality or in other regimes where the response decays too slowly or the response factor grows with system size, these conditions may fail to remain uniform.}
{Failure of a sufficient response condition therefore makes the present certificate inconclusive.}

We additionally performed a finite-size numerical benchmark on a graph-stabilizer Hamiltonian, which supports the predicted quadratic suppression of residual strength and local-state error over the tested perturbative window.
Beyond perturbative control, the same correction circuit structure also provides a useful variational ansatz, showing improvement over the direct first-order construction in broad parameter ranges.
However, the improvement is not uniform and local errors remain large at strong deformations.
{We also recognize that within a restricted ansatz, lowering the free energy reduces the quantum relative entropy to the target Gibbs state but need not improve every local observable prediction.}
{Residual and local-error diagnostics therefore remain useful alongside the variational objective.}

These limitations point to several concrete extensions.
{One direction for further analysis is to obtain local certificates without separately estimating the response along a Hamiltonian interpolation.
  A possible route starts from local dissipative dynamics that rapidly converge to the core Gibbs state.
  One would then construct a sampler for the perturbed Gibbs state and show that adding the residual changes its generator only through locally weak terms with sufficiently fast spatial decay.
  Under the required uniform mixing and locality assumptions, dissipative stability results could then bound the change in expectation values of local observables independently of system size~\cite{CubittEtAl2015}.
  Recent rapid-mixing results for selected weakly interacting systems suggest concrete settings in which this approach could be developed~\cite{SmidEtAl2025}.
  Directly certifying local thermal observables offers a complementary way to assess selected observables of an SCTA preparation~\cite{FawziFawziScalet2024}.
  Recent work developed certified algorithms for computing expectation values of observables in thermal states of local quantum Hamiltonians~\cite{FawziFawziScalet2024}.
Obtaining sufficiently tight bounds at feasible classical cost, however, remains a practical challenge.}

A further direction is to develop local core--tail corrections when explicit transition-channel inversion is impractical.
One possibility is to determine the correction variationally, choosing local basis rotations that reduce the core-incompatible part of the Hamiltonian while keeping the correction simple and local~\mbox{\cite{SelsPolkovnikov2017,WurtzClaeysPolkovnikov2020}}.
Continuous unitary-flow methods offer a related alternative by gradually transforming the Hamiltonian toward a more tractable form~\mbox{\cite{Wegner1994,Kehrein2006,ThomsonEisert2024}}.
These approaches avoid explicit channel-by-channel energy denominators, although genuine resonances can still obstruct a perturbative basis rotation.
For variational applications, a promising direction is to use the constructive picture to guide how the ansatz is enlarged, rather than simply adding more layers or free parameters.
When a chosen core--tail pair leaves a significant local mismatch, one could selectively enrich the core or add physically motivated local rotations to the tail, while using residual and local-error diagnostics to determine which changes are actually useful.
This may be particularly valuable near resonance or outside the perturbative regime, where the correction generators remain physically relevant even after their perturbative coefficients lose validity, as observed in the numerical benchmark.

\begingroup
\makeatletter
\@secpenalty=10000\relax
\makeatother
\section*{{Software and Data availability}}
{The numerical data reported in this work are available from the author upon reasonable request.}
{The source code, calculation parameters, and instructions for reproducing the numerical results and figures are publicly available in Ref.~\cite{LiSCTACode2026}.}
\par
\endgroup

\section*{Acknowledgments}
We thank Khadijeh Najafi and Semeon Valgushev for helpful discussions in the early stages of this work.

{GPT-5.6 Sol (OpenAI) was used under the author's direction to assist with manuscript polishing, literature discovery, code development, and proof development.}
The scientific ideas and conceptual development of this work originated solely with the author.
{All AI-assisted materials were reviewed and verified by the author, and the author takes full responsibility for the content of this work.}

\nocite{Georgii2011,Baxter1982,MottonenEtAl2005,VatanWilliams2004,VirtanenEtAl2020}
\bibliography{references,supplement-reference}

\onecolumngrid
\clearpage
\arxivRestoreTOCWrites
\pagestyle{plain}

\renewcommand{\thesection}{S\arabic{section}}
\renewcommand{\thefigure}{S\arabic{figure}}
\renewcommand{\theHsection}{supp.\arabic{section}}
\renewcommand{\theHsubsection}{supp.\arabic{section}.\arabic{subsection}}
\renewcommand{\theHsubsubsection}{supp.\arabic{section}.\arabic{subsection}.\arabic{subsubsection}}
\renewcommand{\theHfigure}{supp.\arabic{figure}}
\renewcommand{\thetable}{S\arabic{table}}
\renewcommand{\theHtable}{supp.\arabic{table}}
\setcounter{section}{0}
\setcounter{subsection}{0}
\setcounter{subsubsection}{0}
\setcounter{equation}{0}
\setcounter{theorem}{0}
\setcounter{figure}{0}
\setcounter{table}{0}
\phantomsection
\pdfbookmark[0]{Supplemental Material}{arxiv-supplement}

\begin{CJK*}{UTF8}{gbsn}
  \begin{center}
    {\large\bfseries Supplemental Material for ``Spectral Core--Tail Architecture for Locally Certified Gibbs-State Preparation''\par}
    \vspace{0.8em}
    Rui-Hao Li (李瑞浩)\par
    {\href{mailto:lir9@ccf.org}{lir9@ccf.org}}\par
  \end{center}
\end{CJK*}

This Supplemental Material supplies the proofs and construction details omitted from the condensed main text.
It is organized to preserve the logical separation in the main text.
Each proposition, theorem, and corollary stated in the main text has a dedicated proof subsection below.

\begingroup
\renewcommand{\baselinestretch}{1.0}\selectfont
\makeatletter
\let\@refstar\arxiv@refstar
\let\l@paragraph\@gobbletwo
\let\l@subsubsection\@gobbletwo
\tableofcontents
\makeatother
\endgroup
\clearpage

\section{Backward geometry and the operational error budget}
\label{sec:supp-geometry}

\subsection{Proof of Proposition~\ref*{main-prop:backward-region-propagation}}
\label{sec:supp-propagation-proof}

\begin{proof}
  Let
  \begin{equation*}
    \Delta=\rho-\sigma,
    \qquad
    B=B_U(A),
    \qquad
    \Delta'_A=\left[U\Delta U^\dagger\right]_A.
  \end{equation*}
  Because $\rho$ and $\sigma$ are Hermitian, both $\Delta$ and $\Delta'_A$ are Hermitian.
  {Applying the trace-norm duality relation in Eq.~\maineqref{eq:trace-norm-duality-main} of the main text to $\Delta'_A$ gives}
  \begin{equation}
    \norm{\Delta'_A}_1
    =
    \sup_{\substack{O_A=O_A^\dagger\\\norm{O_A}_\infty\le1}}
    \abs{\Tr_A(O_A\Delta'_A)}.
    \label{eq:supp-region-cone-duality-A}
  \end{equation}
  Fix an arbitrary Hermitian contraction $O_A$ appearing in this supremum, and extend it to the full system as
  \begin{equation*}
    \widehat O_A=O_A\otimes\id_{A^c}.
  \end{equation*}
  By the defining property of the partial trace,
  \begin{equation*}
    \Tr_A(O_A\Delta'_A)
    =
    \Tr\!\left(\widehat O_AU\Delta U^\dagger\right).
  \end{equation*}
  Cyclicity of the full-system trace then yields
  \begin{equation*}
    \Tr\!\left(\widehat O_AU\Delta U^\dagger\right)
    =
    \Tr\!\left(U^\dagger\widehat O_AU\Delta\right).
  \end{equation*}
  By the definition of $B=B_U(A)$, the pulled-back observable $U^\dagger\widehat O_AU$ is supported on $B$.
  Consequently, there exists an operator $O_B$ on $\cH_B$ such that
  \begin{equation}
    U^\dagger\widehat O_AU=O_B\otimes\id_{B^c}.
    \label{eq:supp-region-cone-pulled-observable}
  \end{equation}
  Unitary conjugation preserves Hermiticity and the operator norm, while tensoring with an identity also does not change the operator norm.
  It follows that
  \begin{equation*}
    O_B=O_B^\dagger,
    \qquad
    \norm{O_B}_\infty
    =
    \norm{\widehat O_A}_\infty
    =
    \norm{O_A}_\infty
    \le1.
  \end{equation*}
  Thus $O_B$ is a Hermitian contraction on $B$.
  Using Eq.~\eqref{eq:supp-region-cone-pulled-observable} and the defining property of the partial trace once more, we obtain
  \begin{equation*}
    \Tr\!\left(U^\dagger\widehat O_AU\Delta\right)
    =
    \Tr\!\left[\left(O_B\otimes\id_{B^c}\right)\Delta\right]
    =
    \Tr_B\!\left(O_B\Delta_B\right),
  \end{equation*}
  where $\Delta_B=\left[\rho-\sigma\right]_B$.
  {Using the fact that $O_B$ is a Hermitian contraction and applying the trace-norm duality again to $\Delta_B$ gives}
  \begin{equation*}
    \abs{\Tr_B\!\left(O_B\Delta_B\right)}
    \le
    \norm{\Delta_B}_1
    =
    \norm{\left[\rho-\sigma\right]_{B_U(A)}}_1.
  \end{equation*}
  This bound holds for every Hermitian contraction $O_A$ in Eq.~\eqref{eq:supp-region-cone-duality-A}.
  Taking the supremum over all such $O_A$ therefore yields
  \begin{equation}
    \norm{\left[U\rho U^\dagger-U\sigma U^\dagger\right]_A}_1
    \le
    \norm{\left[\rho-\sigma\right]_{B_U(A)}}_1.
    \label{eq:supp-backward-region}
  \end{equation}

  We next prove the diameter and local-metric claims.
  Fix a nonempty output region $A$ with $\diam(A)\le r$.
  For any $L\ge0$ and any $x,x'\in A^{(L)}$, choose $y,y'\in A$ such that $\dist(x,y)\le L$ and $\dist(x',y')\le L$.
  The triangle inequality gives
  \begin{equation*}
    \dist(x,x')
    \le
    \dist(x,y)+\dist(y,y')+\dist(y',x')
    \le
    \diam(A)+2L,
  \end{equation*}
  where
  \begin{equation*}
    \diam(A)=\sup_{y,y'\in A}\dist(y,y').
  \end{equation*}
  Taking the supremum over $x,x'\in A^{(L)}$ gives
  \begin{equation*}
    \diam(A^{(L)})\le\diam(A)+2L.
  \end{equation*}
  Applying this estimate with $L=\ell_U(A)$, and using $B_U(A)\subseteq A^{(\ell_U(A))}$ together with $\ell_U(A)\le\ell_U(r)$, gives
  \begin{equation}
    \begin{aligned}
      \diam(B_U(A))
      &\le
      \diam\!\left(A^{(\ell_U(A))}\right)
      \\
      &\le
      \diam(A)+2\ell_U(A)
      \le
      r+2\ell_U(r).
    \end{aligned}
    \label{eq:supp-neighborhood-diameter}
  \end{equation}

  Hence $B_U(A)$ is an admissible region in the supremum defining $D_{r+2\ell_U(r)}(\rho,\sigma)$.
  Applying Eq.~\eqref{eq:supp-backward-region} yields
  \begin{equation*}
    \norm{\left[U\rho U^\dagger-U\sigma U^\dagger\right]_A}_1
    \le
    \norm{\left[\rho-\sigma\right]_{B_U(A)}}_1
    \le
    D_{r+2\ell_U(r)}(\rho,\sigma).
  \end{equation*}
  The right-hand side is independent of $A$.
  Taking the supremum of the left-hand side over every nonempty $A$ with $\diam(A)\le r$ gives
  \begin{equation}
    D_r(U\rho U^\dagger,U\sigma U^\dagger)
    \le
    D_{r+2\ell_U(r)}(\rho,\sigma).
    \label{eq:supp-finite-radius}
  \end{equation}

  Finally, suppose that $U$ is a depth-$d$ circuit whose gate supports have diameter at most $a$.
  Write the circuit as $U=U_d\cdots U_1$, where each $U_j$ is a layer of gates with pairwise disjoint supports of diameter at most $a$.
  Set $A_0=A$.
  For $k=1,\ldots,d$, let $A_k$ be the union of $A_{k-1}$ and the supports of all gates in layer $U_{d-k+1}$ whose supports intersect $A_{k-1}$:
  \begin{equation*}
    A_k
    =
    A_{k-1}\cup\bigcup_{\substack{S\in\supp(U_{d-k+1})\\S\cap A_{k-1}\neq\varnothing}}S.
  \end{equation*}
  A gate whose support does not intersect $A_{k-1}$ commutes with an observable supported on $A_{k-1}$ and therefore cannot enlarge its support.
  A gate that does intersect $A_{k-1}$ may spread the observable over its entire support $S$.
  Then after pulling an observable supported on $A$ backward through the last $k$ layers, its support is contained in $A_k$.
  Moreover, if a newly added site $x\in A_k\setminus A_{k-1}$, then $x$ belongs to a gate support $S$ that contains some $y\in A_{k-1}$.
  The gate-diameter assumption gives $\dist(x,A_{k-1})\le \dist(x,y)\le a$, and therefore
  \begin{equation*}
    A_k\subseteq A_{k-1}^{(a)}.
  \end{equation*}
  That is, every site in $A_k$ lies in the $a$-neighborhood of $A_{k-1}$.
  Iterating this inclusion gives
  \begin{equation*}
    A_k\subseteq A^{(ka)},
    \qquad
    k=0,\ldots,d.
  \end{equation*}
  Consequently, every observable initially supported on $A$ is supported within $A^{(ad)}$ after being pulled backward through the entire circuit.
  By the definition of the backward causal region, $B_U(A)\subseteq A^{(ad)}$, so $\ell_U(A)\le ad$.
  Taking the supremum over all nonempty $A$ with $\diam(A)\le r$ proves $\ell_U(r)\le ad$.
\end{proof}

\subsection{{Three-source operational error budget}}
\label{sec:supp-error-budget-proof}

{
  We use the local metric $D_r$ and backward radius $\ell_U(r)$ defined in Sec.~\ref*{main-sec:architecture-certification} of the main text.
  Let $\widetilde\sigma_C$ be the prepared approximation to an intended core state $\sigma_C$, which is diagonal or block-diagonal with block dimensions bounded independently of system size.
  For an ideal tail unitary $U$ with specified physical geometry, let $\widetilde\rho_{\mathrm{out}}$ be the output of the implemented channel acting on $\widetilde\sigma_C$.
  Preparation and implementation errors are quantified relative to $\sigma_C$ and $U\widetilde\sigma_CU^\dagger$, respectively.

  For the constructive setting, take $\sigma_C=\rhoG{H_C}$ and $U^\dagger HU=H_C+R$, as in Eq.~\maineqref{eq:intro-core-tail-identity} of the main text.
  Set
  \begin{equation}
    \rho_C=\rhoG{H_C},
    \qquad
    \rho_R=\rhoG{H_C+R}.
    \label{eq:supp-core-and-residual-gibbs}
  \end{equation}
  Here $\rho_C$ is the ideal core Gibbs state, whereas $\rho_R$ is the exact target Gibbs state expressed in the core frame.
  Indeed, unitary covariance of the Gibbs state and Eq.~\maineqref{eq:intro-core-tail-identity} of the main text imply $\rhoG{H}=U\rho_RU^\dagger$.
  At output scale $r$, define each error contribution by the same local metric in the physical frame:
  \begin{subequations}
    \begin{align}
      \varepsilon_{\mathrm{tail}}(r)
      &:=D_r(\widetilde\rho_{\mathrm{out}},U\widetilde\sigma_CU^\dagger),
      \\
      \varepsilon_{\mathrm{core}}(r)
      &:=D_r(U\widetilde\sigma_CU^\dagger,U\rho_CU^\dagger),
      \\
      \varepsilon_{\mathrm{Gibbs}}(r)
      &:=D_r(U\rho_CU^\dagger,U\rho_RU^\dagger).
    \end{align}
    \label{eq:supp-three-errors}
  \end{subequations}

  \begin{theorem}[Operational SCTA error budget]
    \label{thm:supp-operational-error-budget}
    For every $r\ge0$,
    \begin{equation}
      \begin{split}
        D_r(\widetilde\rho_{\mathrm{out}},\rhoG{H})
        \le\min\!\big\{2,{}&\varepsilon_{\mathrm{tail}}(r)
          +\varepsilon_{\mathrm{core}}(r)\\
        &+\varepsilon_{\mathrm{Gibbs}}(r)\big\}.
      \end{split}
      \label{eq:supp-operational-error-budget}
    \end{equation}
    If $\ell_U(r)\le\ell$, then also
    \begin{equation}
      \begin{split}
        D_r(\widetilde\rho_{\mathrm{out}},\rhoG{H})
        \le\min\!\big\{2,{}&\varepsilon_{\mathrm{tail}}(r)
          +D_{r+2\ell}(\widetilde\sigma_C,\rho_C)\\
        &+D_{r+2\ell}(\rho_R,\rho_C)\big\}.
      \end{split}
      \label{eq:supp-finite-radius-operational-error-budget}
    \end{equation}
  \end{theorem}

}

\begin{proof}
  {Adopt the notation and hypotheses of Theorem~\ref{thm:supp-operational-error-budget}, including the three error quantities in Eq.~\eqref{eq:supp-three-errors}.}
  Unitary covariance gives $\rhoG{H}=U\rho_RU^\dagger$.
  To prove the first bound, fix a nonempty output region $A\subseteq\Lambda$ with $\diam(A)\le r$.
  Consider the sequence
  \begin{equation*}
    \tau_0=\widetilde{\rho}_{\mathrm{out}},
    \qquad
    \tau_1=U\widetilde{\sigma}_CU^\dagger,
    \qquad
    \tau_2=U\rho_CU^\dagger,
    \qquad
    \tau_3=U\rho_RU^\dagger=\rhoG{H}.
  \end{equation*}
  The last equality follows from unitary covariance of the Gibbs state and the exact identity $U^\dagger HU=H_C+R$.
  This sequence changes one component of the preparation at a time: tail implementation, core preparation, and residual-induced Gibbs mismatch.

  Linearity of the partial trace and the triangle inequality give
  \begin{equation*}
    \begin{aligned}
      \norm{[\tau_0-\tau_3]_A}_1
      &\le
      \norm{[\tau_0-\tau_1]_A}_1
      \\
      &\quad+
      \norm{[\tau_1-\tau_2]_A}_1
      +
      \norm{[\tau_2-\tau_3]_A}_1.
    \end{aligned}
  \end{equation*}
  Since $A$ is admissible in the three suprema defining $D_r$, the three terms on the right are bounded by $\varepsilon_{\mathrm{tail}}(r)$, $\varepsilon_{\mathrm{core}}(r)$, and $\varepsilon_{\mathrm{Gibbs}}(r)$, respectively.
  Substituting the three estimates into the region-wise triangle inequality gives
  \begin{equation*}
    \norm{
      [\widetilde{\rho}_{\mathrm{out}}-\rhoG{H}]_A
    }_1
    \le
    \varepsilon_{\mathrm{tail}}(r)
    +
    \varepsilon_{\mathrm{core}}(r)
    +
    \varepsilon_{\mathrm{Gibbs}}(r).
  \end{equation*}
  The right-hand side is independent of $A$.
  {Taking the supremum over all nonempty $A$ with $\diam(A)\le r$ establishes the sum of the three error terms in Eq.~\eqref{eq:supp-operational-error-budget}.}
  Moreover, the trace norm of the difference between any two reduced density operators is at most $2$.
  Taking the minimum with this universal bound proves the first claim.

  To prove the finite-radius bound, suppose that $\ell_U(r)\le\ell$.
  Then $B_U(A)\subseteq A^{(\ell)}$ for every admissible output region $A$, and Eq.~\eqref{eq:supp-neighborhood-diameter} gives
  \begin{equation*}
    \diam\!\left(A^{(\ell)}\right)
    \le
    \diam(A)+2\ell
    \le
    r+2\ell.
  \end{equation*}
  Applying Eq.~\eqref{eq:supp-backward-region} to the pair $\widetilde\sigma_C,\rho_C$, using $B_U(A)\subseteq A^{(\ell)}$, and then applying contractivity under the partial trace gives
  \begin{equation*}
    \begin{aligned}
      \norm{
        [U\widetilde\sigma_CU^\dagger-U\rho_CU^\dagger]_A
      }_1
      &\le
      \norm{
        [\widetilde{\sigma}_C-\rho_C]_{B_U(A)}
      }_1
      \\
      &\le
      \norm{
        [\widetilde{\sigma}_C-\rho_C]_{A^{(\ell)}}
      }_1
      \\
      &\le
      D_{r+2\ell}(\widetilde{\sigma}_C,\rho_C),
    \end{aligned}
  \end{equation*}
  and, by the same argument,
  \begin{equation*}
    \norm{
      [U\rho_RU^\dagger-U\rho_CU^\dagger]_A
    }_1
    \le
    D_{r+2\ell}(\rho_R,\rho_C).
  \end{equation*}
  Taking the supremum over $A$ therefore gives
  \begin{equation*}
    \varepsilon_{\mathrm{core}}(r)
    \le
    D_{r+2\ell}(\widetilde\sigma_C,\rho_C),
    \qquad
    \varepsilon_{\mathrm{Gibbs}}(r)
    \le
    D_{r+2\ell}(\rho_R,\rho_C).
  \end{equation*}
  {Substituting these two estimates into the first error budget proves the second inequality in Theorem~\ref{thm:supp-operational-error-budget}.}
\end{proof}
\section{Thermal response and the imported stability theorem}
\label{sec:supp-response}

\subsection{Proof of Proposition~\ref*{main-prop:exact-gibbs-response-main}}
\label{sec:supp-response-identity-proof}

\begin{proof}
  Adopt the notation of Proposition~\ref*{main-prop:exact-gibbs-response-main} and Eqs.~\maineqref{eq:comparison-path}--\maineqref{eq:kubo-mori-kernel} of the main text.
  Let $G_s=\ee^{-\beta H_s}$ and $Z_s=\Tr(G_s)$, so that $\rho_s=G_s/Z_s$.
  Duhamel's formula gives
  \begin{equation}
    \frac{dG_s}{ds}
    =-\beta\int_0^1
    \ee^{-\beta(1-\tau)H_s}R\ee^{-\beta\tau H_s}\,d\tau.
    \label{eq:supp-duhamel}
  \end{equation}
  Taking the trace and using cyclicity yields
  \begin{equation}
    \frac{dZ_s}{ds}
    =
    -\beta\Tr(RG_s)
    =
    -\beta Z_s\langle R\rangle_s.
    \label{eq:supp-partition-derivative}
  \end{equation}
  For the unnormalized numerator $N_O(s)=\Tr(OG_s)$, Eq.~\eqref{eq:supp-duhamel}, the substitution $\tau\mapsto1-\tau$, and cyclicity give
  \begin{equation}
    \frac{dN_O(s)}{ds}
    =-\beta Z_s\int_0^1
    \Tr(\rho_s^\tau R\rho_s^{1-\tau}O)\,d\tau.
    \label{eq:supp-numerator-derivative}
  \end{equation}
  Differentiating $\langle O\rangle_s=N_O(s)/Z_s$ and inserting Eqs.~\eqref{eq:supp-partition-derivative} and \eqref{eq:supp-numerator-derivative} therefore gives
  \begin{equation*}
    \begin{aligned}
      \frac{d}{ds}\langle O\rangle_s
      &=-\beta\int_0^1
      \Tr(\rho_s^\tau R\rho_s^{1-\tau}O)\,d\tau
      +\beta\langle R\rangle_s\langle O\rangle_s
      \\
      &=-\beta\cK_s(R,O),
    \end{aligned}
  \end{equation*}
  which proves the response identity.

  For the remaining properties, let $V$ and $O$ be arbitrary Hermitian operators.
  Bilinearity of $\cK_s$ follows directly from its definition in Eq.~\maineqref{eq:kubo-mori-kernel} of the main text.
  To prove symmetry, substitute $u=1-\tau$ in the integral and use cyclicity of the trace to obtain $\cK_s(V,O)=\cK_s(O,V)$.
  For Hermitian $V$ and $O$, complex conjugation of the integral interchanges the two arguments, so symmetry also implies that $\cK_s(V,O)$ is real.
  Moreover, the same definition gives $\cK_s(V+c\id,O+d\id)=\cK_s(V,O)$ for all real $c$ and $d$.

  It remains to prove positivity and the spectral-half-width bound.
  Diagonalize $\rho_s=\sum_ip_i\ket{i}\!\bra{i}$ and set $\widetilde V=V-\langle V\rangle_s\id$.
  Direct evaluation gives
  \begin{equation}
    \cK_s(V,V)
    =
    \sum_{i,j}L(p_i,p_j)\abs{\widetilde V_{ij}}^2,
    \label{eq:supp-log-mean-form}
  \end{equation}
  where the logarithmic mean is
  \begin{equation*}
    L(x,y)
    =
    \begin{cases}
      \dfrac{x-y}{\log x-\log y},&x\ne y,\\[0.6em]
      x,&x=y.
    \end{cases}
  \end{equation*}
  Every Gibbs eigenvalue $p_i$ is positive and thus $L(p_i,p_j)>0$, so Eq.~\eqref{eq:supp-log-mean-form} proves positive semidefiniteness.
  The Cauchy--Schwarz inequality for a positive semidefinite real bilinear form consequently gives
  \begin{equation}
    \abs{\cK_s(V,O)}^2
    \le
    \cK_s(V,V)\cK_s(O,O).
    \label{eq:supp-km-cauchy-schwarz}
  \end{equation}
  For every Hermitian $X$, set $\widetilde X=X-\langle X\rangle_s\id$.
  The scalar inequality $L(x,y)\le(x+y)/2$ and Eq.~\eqref{eq:supp-log-mean-form} give
  \begin{equation}
    \begin{aligned}
      \cK_s(X,X)
      &\le
      \frac12\sum_{i,j}(p_i+p_j)\abs{\widetilde X_{ij}}^2
      \\
      &=\Tr(\rho_s\widetilde X^2)
      \\
      &=\Tr(\rho_sX^2)-\langle X\rangle_s^2
      \le\Tr(\rho_sX^2)
      \le\norm{X}_\infty^2.
    \end{aligned}
    \label{eq:supp-km-variance}
  \end{equation}
  In the first equality, the contributions weighted by $p_i$ and $p_j$ coincide after interchanging the dummy indices because $\abs{\widetilde X_{ij}}^2=\abs{\widetilde X_{ji}}^2$.
  The final inequality follows from $X^2\le\norm{X}_\infty^2\id$.
  Applying Eqs.~\eqref{eq:supp-km-cauchy-schwarz} and \eqref{eq:supp-km-variance} to $V-c\id$ and $O-d\id$ gives
  \begin{equation*}
    \abs{\cK_s(V,O)}
    \le
    \norm{V-c\id}_\infty\norm{O-d\id}_\infty.
  \end{equation*}
  Taking the infimum independently over $c,d\in\mathbb R$ proves the spectral-half-width bound and completes the proof of Proposition~\ref*{main-prop:exact-gibbs-response-main}.
\end{proof}

\subsection{Proof of Corollary~\ref*{main-cor:global-width-main}}
\label{sec:supp-width-bound-proof}

\begin{proof}
  Adopt the notation of Corollary~\ref*{main-cor:global-width-main} and let $\Delta=\rho_1-\rho_0$ and $\Delta_B=[\rho_1-\rho_0]_B$.
  Trace-norm contractivity of the partial trace on Hermitian operators gives $\norm{\Delta_B}_1\le\norm{\Delta}_1$, while normalization and the triangle inequality give $\norm{\Delta}_1\le2$.
  For any Hermitian contraction $O$, integrating Eq.~\maineqref{eq:kubo-mori-derivative} of the main text and applying the spectral-half-width bound of Proposition~\ref*{main-prop:exact-gibbs-response-main} give
  \begin{equation*}
    \begin{aligned}
      \abs{\Tr(O\Delta)}
      &\le
      \beta\int_0^1\abs{\cK_s(R,O)}\,ds
      \\
      &\le
      \beta w(R)w(O)
      \le
      \beta w(R).
    \end{aligned}
  \end{equation*}
  Trace-norm duality for the Hermitian operator $\Delta$ therefore yields $\norm{\Delta}_1\le\beta w(R)$.
  Combining the three estimates proves Corollary~\ref*{main-cor:global-width-main}.
\end{proof}

\subsection{Proof of Theorem~\ref*{main-thm:response-profile-main}}
\label{sec:supp-localized-residual-proof}

\begin{proof}
  Adopt the notation of Theorem~\ref*{main-thm:response-profile-main} and set $c_R=\sum_Xc_X^\star$, so that $R=c_R\id+\sum_X\bar r_X$.
  Let $\Delta_B=[\rho_1-\rho_0]_B$, and fix a Hermitian contraction $O_B$.
  Eq.~\maineqref{eq:integrated-expectation-main} of the main text and the response identity in Eq.~\maineqref{eq:kubo-mori-derivative} give
  \begin{equation*}
    \Tr(O_B\Delta_B)
    =
    -\beta\int_0^1\cK_s(R,O_B)\,ds.
  \end{equation*}
  Bilinearity and scalar-shift invariance of $\cK_s$, together with $R=c_R\id+\sum_X\bar r_X$, imply
  \begin{equation*}
    \cK_s(R,O_B)
    =
    \sum_X\cK_s(\bar r_X,O_B).
  \end{equation*}
  Since the lattice is finite, the residual sum is finite, and hence the sum and path integral may be interchanged.
  Taking the absolute value and using the definition of $\chi_s$ yields
  \begin{equation*}
    \begin{aligned}
      \abs{\Tr(O_B\Delta_B)}
      &\le
      \beta\sum_X\int_0^1
      \abs{\cK_s(\bar r_X,O_B)}\,ds
      \\
      &\le
      \beta\sum_XJ_X\int_0^1\chi_s(X\to B)\,ds
      \\
      &=
      \beta\sum_XJ_X\bar\chi(X\to B).
    \end{aligned}
  \end{equation*}
  Trace-norm duality for the Hermitian operator $\Delta_B$ gives
  \begin{equation*}
    \norm{\Delta_B}_1
    =
    \sup_{\substack{O_B=O_B^\dagger\\\norm{O_B}_\infty\le1}}
    \abs{\Tr(O_B\Delta_B)},
  \end{equation*}
  so taking the supremum proves the residual-dependent estimate in Theorem~\ref*{main-thm:response-profile-main}.
  Finally, $\rho_{0,B}$ and $\rho_{1,B}$ are normalized density operators, and therefore $\norm{\Delta_B}_1\le2$.
  Combining the two estimates proves Theorem~\ref*{main-thm:response-profile-main}.
\end{proof}

\subsection{Proof of Theorem~\ref*{main-thm:shell-response-main}}
\label{sec:supp-shell-bound-proof}

\begin{proof}
  Adopt the notation and response-envelope hypothesis of Theorem~\ref*{main-thm:shell-response-main}.
  {Terms with $J_X=0$ are scalar and contribute neither to the Gibbs response nor to the incident strength, so it suffices to consider nonempty supports with $J_X>0$.}
  It is enough to bound the response sum in Theorem~\ref*{main-thm:response-profile-main}.
  We first split that sum into supports that overlap $B$ and supports that are disjoint from $B$.

  For the overlapping supports, the universal bound $\bar\chi(X\to B)\le1$ gives
  \begin{equation*}
    \begin{aligned}
      \sum_{X:X\cap B\ne\varnothing}
      J_X\bar\chi(X\to B)
      &\le
      \sum_{X:X\cap B\ne\varnothing}J_X
      \\
      &\le
      \sum_{x\in B}\sum_{X\ni x}J_X
      \\
      &\le
      \varepsilon_R\abs{B}.
    \end{aligned}
  \end{equation*}
  The second inequality holds because every support $X$ with $X\cap B\ne\varnothing$ contains at least one site $x\in B$, so its weight $J_X$ appears at least once in the double sum.
  Then if $X$ intersects $B$ at several sites, it is counted several times, which only enlarges the upper bound.

  For each support $X$ disjoint from $B$, choose a nearest anchor $x_X\in X$ satisfying $\dist(x_X,B)=\dist(X,B)$.
  {Such an anchor exists for every nonempty finite $X$.}
  {The response envelope, assumed for every support in the decomposition disjoint from $B$, then gives}
  \begin{equation*}
    \bar\chi(X\to B)
    \le
    C_\beta(B)g(\dist(x_X,B)).
  \end{equation*}
  Grouping the residual terms by their nearest anchors therefore yields
  \begin{equation*}
    \begin{aligned}
      \sum_{X:X\cap B=\varnothing}
      J_X\bar\chi(X\to B)
      &\le
      C_\beta(B)
      \sum_{x\in\Lambda\setminus B}
      g(\dist(x,B))
      \sum_{\substack{X:X\cap B=\varnothing\\x_X=x}}J_X
      \\
      &\le
      C_\beta(B)\varepsilon_R
      \sum_{x\in\Lambda\setminus B}g(\dist(x,B)).
    \end{aligned}
  \end{equation*}
  The last inequality follows because $x_X=x$ implies $x\in X$, so the inner anchored sum is a sub-sum of $\sum_{X\ni x}J_X$, which itself is bounded by $\varepsilon_R$.

  Finally, partition $\Lambda\setminus B$ into the unit-width shells defined before Theorem~\ref*{main-thm:shell-response-main} in the main text.
  If $x\in\mathcal S_B(n)$, then $\dist(x,B)\ge n$, and monotonicity of $g$ gives $g(\dist(x,B))\le g(n)$.
  Consequently,
  \begin{equation*}
    \sum_{x\in\Lambda\setminus B}g(\dist(x,B))
    \le
    \sum_{n=0}^\infty N_B(n)g(n)
    =
    \Sigma_g(B;\Lambda).
  \end{equation*}
  Adding the overlapping and disjoint estimates and inserting the result into Theorem~\ref*{main-thm:response-profile-main} proves Theorem~\ref*{main-thm:shell-response-main}.
\end{proof}

\subsection{Proof of Corollary~\ref*{main-cor:volume-uniform-shell-main}}
\label{sec:supp-volume-uniform-proof}

\begin{proof}
  Adopt the lattice sequence and uniform quantities defined in Corollary~\ref*{main-cor:volume-uniform-shell-main} and the preceding main-text discussion.
  Fix $N\ge1$ and a nonempty region $B\subseteq\Lambda_N$ with $\diam(B)\le r$.
  Choose any $x\in B$.
  Every $y\in B$ obeys $\dist(x,y)\le r$, so the definition of $v(r)$ gives $\abs{B}\le v(r)$.
  Applying Theorem~\ref*{main-thm:shell-response-main} on $\Lambda_N$ and using the definitions of $C_{\beta,r}$ and $\Sigma_r$ yields
  \begin{equation*}
    \norm{[\rho_1^{(N)}-\rho_0^{(N)}]_B}_1
    \le
    \min\left\{2,
      \beta\varepsilon_R
      \left[v(r)+C_{\beta,r}\Sigma_r\right]
    \right\}.
  \end{equation*}
  The right-hand side is independent of both $B$ and $N$.
  Taking the supremum first over all nonempty $B\subseteq\Lambda_N$ with $\diam(B)\le r$ and then over $N$ proves Corollary~\ref*{main-cor:volume-uniform-shell-main}.
\end{proof}

\subsection{Factorized paths and collective block response}
\label{sec:supp-factorized-paths}

This subsection supplies the factorized-path result summarized in Sec.~\ref*{main-sec:physical-response-regimes} of the main text.
Let $\{\Lambda_\alpha\}_{\alpha\in\mathcal I}$ be a fixed partition of $\Lambda$ into pairwise disjoint nonempty blocks, and assume that their size and diameter are bounded independently of the total volume:
\begin{equation}
  q_{\rm prod}
  :=
  \sup_{\alpha\in\mathcal I}\abs{\Lambda_\alpha}
  <
  \infty,
  \qquad
  a_{\rm prod}
  :=
  \sup_{\alpha\in\mathcal I}\diam(\Lambda_\alpha)
  <
  \infty.
  \label{eq:supp-product-block-bounds}
\end{equation}
Suppose that the complete comparison path has the block form
\begin{equation}
  \begin{aligned}
    H_C&=\sum_{\alpha\in\mathcal I}h_{C,\alpha},
    &
    R&=\sum_{\alpha\in\mathcal I}r_\alpha,
    \\
    H_s&=\sum_{\alpha\in\mathcal I}\bigl(h_{C,\alpha}+sr_\alpha\bigr),
    &
    \rho_s&=\bigotimes_{\alpha\in\mathcal I}\rho_{s,\alpha},
  \end{aligned}
  \label{eq:supp-product-path}
\end{equation}
where $\supp(h_{C,\alpha}),\supp(r_\alpha)\subseteq\Lambda_\alpha$ and $\rho_{s,\alpha}:=\rhoG{h_{C,\alpha}+sr_\alpha}$.
The operators within one block may be arbitrary and need not commute.
The essential assumption is that the same partition factorizes the complete path.
If several residual contributions occur within one block, combine them into the single block residual $r_\alpha$ before applying the following argument.
Center each block residual as in Eq.~\maineqref{eq:centered-residual-terms} of the main text and write
\begin{equation}
  \bar r_\alpha
  =
  r_\alpha-c_\alpha^\star\id_{\Lambda_\alpha},
  \qquad
  J_\alpha
  =
  \norm{\bar r_\alpha}_\infty.
  \label{eq:supp-product-centered-residual}
\end{equation}
For an observed region $B$, let
\begin{equation}
  \mathcal I_B
  :=
  \{\alpha\in\mathcal I:\Lambda_\alpha\cap B\ne\varnothing\}
  \label{eq:supp-product-intersecting-blocks}
\end{equation}
denote the blocks intersecting $B$.

\begin{proposition}[Collective product-path bound]
  \label{prop:supp-product-collective}
  Under Eqs.~\eqref{eq:supp-product-block-bounds}--\eqref{eq:supp-product-intersecting-blocks}, every region $B$ satisfies
  \begin{equation}
    \begin{aligned}
      \norm{[\rho_1-\rho_0]_B}_1
      &\le
      \min\left\{2,
        \beta\int_0^1
        \left[
          \sum_{\alpha\in\mathcal I_B}
          \cK_s(\bar r_\alpha,\bar r_\alpha)
        \right]^{1/2}ds
      \right\}
      \\
      &\le
      \min\left\{2,
        \beta
        \left(
          \sum_{\alpha\in\mathcal I_B}J_\alpha^2
        \right)^{1/2}
      \right\}.
    \end{aligned}
    \label{eq:supp-product-collective}
  \end{equation}
\end{proposition}

\begin{proof}
  Fix a Hermitian contraction $O_B$ supported on $B$.
  We first identify the block residuals that can contribute to its response.
  Since $R$ differs from $\sum_{\alpha\in\mathcal I}\bar r_\alpha$ only by a scalar, scalar-shift invariance gives
  \begin{equation*}
    \cK_s(R,O_B)
    =
    \sum_{\alpha\in\mathcal I}\cK_s(\bar r_\alpha,O_B).
  \end{equation*}
  If $\alpha\notin\mathcal I_B$, then $\Lambda_\alpha\cap B=\varnothing$, so $O_B$ acts as the identity on the factor block $\Lambda_\alpha$.
  Write $\rho_s=\rho_{s,\alpha}\otimes\rho_{s,\Lambda_\alpha^c}$, where $\rho_{s,\Lambda_\alpha^c}$ is the product state on all remaining blocks.
  Using Eq.~\eqref{eq:supp-product-path}, the trace in the Kubo--Mori covariance factorizes for every $\tau\in[0,1]$ as
  \begin{equation*}
    \begin{aligned}
      \Tr\!\left(\rho_s^\tau\bar r_\alpha\rho_s^{1-\tau}O_B\right)
      &={}
      \Tr_{\Lambda_\alpha}\!\left(
        \rho_{s,\alpha}^\tau\bar r_\alpha\rho_{s,\alpha}^{1-\tau}
      \right)
      \Tr_{\Lambda_\alpha^c}\!\left(\rho_{s,\Lambda_\alpha^c}O_B\right)
      \\
      &={}
      \Tr_{\Lambda_\alpha}(\rho_{s,\alpha}\bar r_\alpha)
      \Tr(\rho_sO_B)
      \\
      &={}
      \langle\bar r_\alpha\rangle_s\langle O_B\rangle_s.
    \end{aligned}
  \end{equation*}
  In the second line, cyclicity combines $\rho_{s,\alpha}^\tau$ and $\rho_{s,\alpha}^{1-\tau}$ into $\rho_{s,\alpha}$.
  The final expression is independent of $\tau$ and equals the disconnected term subtracted in Eq.~\maineqref{eq:kubo-mori-kernel} of the main text.
  Hence
  \begin{equation}
    \cK_s(\bar r_\alpha,O_B)
    =
    0,
    \qquad
    \alpha\notin\mathcal I_B.
    \label{eq:supp-product-outside-zero}
  \end{equation}
  Define the centered residual on the blocks touched by $B$ by
  \begin{equation*}
    R_{\mathcal I_B}
    :=
    \sum_{\alpha\in\mathcal I_B}\bar r_\alpha.
  \end{equation*}
  Eq.~\eqref{eq:supp-product-outside-zero} gives $\cK_s(R,O_B)=\cK_s(R_{\mathcal I_B},O_B)$, so Eq.~\maineqref{eq:kubo-mori-derivative} of the main text becomes
  \begin{equation*}
    \abs{\frac{d}{ds}\langle O_B\rangle_s}
    =
    \beta\abs{\cK_s(R_{\mathcal I_B},O_B)}.
  \end{equation*}
  The Cauchy--Schwarz inequality for the Kubo--Mori covariance and Eq.~\maineqref{eq:kubo-mori-width-main} of the main text imply
  \begin{equation*}
    \abs{\cK_s(R_{\mathcal I_B},O_B)}
    \le
    \cK_s(R_{\mathcal I_B},R_{\mathcal I_B})^{1/2}
    \cK_s(O_B,O_B)^{1/2}
    \le
    \cK_s(R_{\mathcal I_B},R_{\mathcal I_B})^{1/2}.
  \end{equation*}
  Expanding the remaining quadratic form gives
  \begin{equation*}
    \cK_s(R_{\mathcal I_B},R_{\mathcal I_B})
    =
    \sum_{\alpha,\gamma\in\mathcal I_B}
    \cK_s(\bar r_\alpha,\bar r_\gamma).
  \end{equation*}
  If $\alpha\ne\gamma$, product factorization gives
  \begin{equation*}
    \Tr\!\left(\rho_s^\tau\bar r_\alpha\rho_s^{1-\tau}\bar r_\gamma\right)
    =
    \langle\bar r_\alpha\rangle_s\langle\bar r_\gamma\rangle_s
  \end{equation*}
  for every $\tau\in[0,1]$, and therefore $\cK_s(\bar r_\alpha,\bar r_\gamma)=0$.
  Thus
  \begin{equation*}
    \cK_s(R_{\mathcal I_B},R_{\mathcal I_B})
    =
    \sum_{\alpha\in\mathcal I_B}
    \cK_s(\bar r_\alpha,\bar r_\alpha).
  \end{equation*}
  Combining these identities gives
  \begin{equation*}
    \abs{\frac{d}{ds}\langle O_B\rangle_s}
    \le
    \beta
    \left[
      \sum_{\alpha\in\mathcal I_B}
      \cK_s(\bar r_\alpha,\bar r_\alpha)
    \right]^{1/2}.
  \end{equation*}
  Integrating over $s$, taking the supremum over Hermitian contractions $O_B$, and using trace-norm duality proves the first residual-dependent bound in Eq.~\eqref{eq:supp-product-collective}.
  Eq.~\maineqref{eq:kubo-mori-width-main} and Eq.~\eqref{eq:supp-product-centered-residual} give
  \begin{equation*}
    \cK_s(\bar r_\alpha,\bar r_\alpha)
    \le
    w(\bar r_\alpha)^2
    =
    J_\alpha^2.
  \end{equation*}
  The integration interval has unit length, so substitution proves the second residual-dependent bound.
  Finally, the trace norm of the difference between two reduced states is at most $2$, which completes the proof.
\end{proof}

The same factorization realizes Eq.~\maineqref{eq:path-averaged-response-envelope-main} of the main text with an exact finite-range envelope.
Writing $X_\alpha:=\supp(r_\alpha)$, if $\dist(X_\alpha,B)>a_{\rm prod}$, then $\Lambda_\alpha$ cannot intersect $B$, and Eq.~\eqref{eq:supp-product-outside-zero} gives $\bar\chi(X_\alpha\!\to\!B)=0$.
Together with the universal bound $\bar\chi\le1$ stated after Eq.~\maineqref{eq:path-averaged-response-main}, this yields $C_\beta(B)=1$ and
\begin{equation}
  g_{\rm prod}(r)
  =
  \begin{cases}
    1,&r\le a_{\rm prod},
    \\
    0,&r>a_{\rm prod}.
  \end{cases}
  \label{eq:supp-product-envelope}
\end{equation}
The associated shell sum contains only sites within the fixed distance $a_{\rm prod}$ of $B$ and is therefore volume uniform whenever fixed-radius neighborhoods have uniformly bounded cardinality across the lattice sequence.
Proposition~\ref{prop:supp-product-collective} is stronger than applying the main-text shell theorem (Theorem~\ref*{main-thm:shell-response-main}) term by term because it retains cross-block orthogonality and combines the contributing block strengths through a root-sum-square bound.

Single-site factors are recovered by taking $\Lambda_\alpha$ to contain one site, while a disjoint dimer covering provides the simplest nonsingleton example.
Each block Hamiltonian and residual may contain strong classical or quantum correlations internally, but an observable can respond only to residuals in the blocks it intersects.

\subsection{Classical diagonal paths and Dobrushin contraction}
\label{sec:supp-dobrushin-paths}

This subsection derives the classical response bound quoted in Sec.~\ref*{main-sec:physical-response-regimes} of the main text.
Assume that every $H_s=H_C+sR$ is diagonal in one fixed local product basis and that the centered residual terms $\bar r_X$ are diagonal in that basis.
Write a classical configuration as $\sigma=(\sigma_x)_{x\in\Lambda}$ and let $H_s(\sigma)$ denote the corresponding diagonal matrix element.
The associated Gibbs law is
\begin{equation}
  \mu_s(\sigma)
  =
  \frac{\ee^{-\beta H_s(\sigma)}}{Z_s},
  \qquad
  Z_s
  =
  \sum_\sigma\ee^{-\beta H_s(\sigma)}.
  \label{eq:supp-classical-law}
\end{equation}
Replacing a test observable $O_B$ by its diagonal part leaves its Kubo--Mori pairing with $\bar r_X$ unchanged and cannot increase its operator norm.
The supremum defining the response profile may therefore be restricted to diagonal $O_B$, which can be identified with ordinary classical functions.
Because $\bar r_X$, $O_B$, and $\rho_s$ then commute, the Kubo--Mori covariance reduces to
\begin{equation}
  \cK_s(\bar r_X,O_B)
  =
  \operatorname{Cov}_{\mu_s}(\bar r_X,O_B).
  \label{eq:supp-classical-km-covariance}
\end{equation}
Local Gibbs certification is thus reduced to controlling the propagation of classical statistical influence.

For $i\ne j$, define the direct-influence coefficient
\begin{equation}
  C^{(s)}_{ij}
  =
  \sup_{\substack{\sigma,\sigma'\\
  \sigma_{\Lambda\setminus\{j\}}=\sigma'_{\Lambda\setminus\{j\}}}}
  \left\|
  \mu_{s,i}(\,\cdot\mid\sigma_{\Lambda\setminus\{i\}})
  -
  \mu_{s,i}(\,\cdot\mid\sigma'_{\Lambda\setminus\{i\}})
  \right\|_{\rm TV},
  \label{eq:supp-Dobrushin-direct-influence}
\end{equation}
and set $C^{(s)}_{ii}=0$.
The two configurations agree everywhere except possibly at site $j$, so $C^{(s)}_{ij}$ measures the largest change in the conditional distribution at $i$ caused by changing the value at $j$.
Here $\norm{p-q}_{\rm TV}:=\frac12\sum_u\abs{p(u)-q(u)}$.
For a path-uniform estimate, assume that one nonnegative matrix $C^{\rm Dob}$ bounds all direct-influence matrices,
\begin{equation}
  C^{(s)}_{ij}
  \le
  C^{\rm Dob}_{ij}
  \qquad
  \text{for every }s\in[0,1],
  \label{eq:supp-Dobrushin-path-majorant}
\end{equation}
and impose the Dobrushin contraction condition
\begin{equation}
  \alpha
  :=
  \sup_i\sum_jC^{\rm Dob}_{ij}
  <
  1.
  \label{eq:supp-Dobrushin-contraction}
\end{equation}
The quantity $\alpha$ bounds the total direct influence entering any one site.
Because this is a worst-case sufficient condition, $\alpha\ge1$ makes the Dobrushin argument inconclusive but does not imply the absence of uniqueness or correlation decay.

One multiplication by $C^{\rm Dob}$ represents one step of influence propagation, so $(C^{\rm Dob})^m_{ij}$ sums the weights of $m$-step influence chains from $j$ to $i$.
The maximal row sum of $(C^{\rm Dob})^m$ is at most $\alpha^m$, and Eq.~\eqref{eq:supp-Dobrushin-contraction} therefore makes the total influence summable.
Define
\begin{equation}
  D^{\rm Dob}
  :=
  (I-C^{\rm Dob})^{-1}
  =
  \sum_{m=0}^{\infty}(C^{\rm Dob})^m.
  \label{eq:supp-Dobrushin-resolvent}
\end{equation}
The identity term describes zero-step influence, $C^{\rm Dob}$ describes direct influence, and higher powers describe influence transmitted through intermediate sites.

For a classical function $F$, define its coordinate oscillation by
\begin{equation}
  \delta_i(F)
  =
  \sup_{\substack{\sigma,\sigma'\\
  \sigma_{\Lambda\setminus\{i\}}=\sigma'_{\Lambda\setminus\{i\}}}}
  \abs{F(\sigma)-F(\sigma')}.
  \label{eq:supp-coordinate-oscillation}
\end{equation}
The Dobrushin--F\"ollmer covariance estimate gives
\begin{equation}
  \abs{\operatorname{Cov}_{\mu_s}(F,G)}
  \le
  \frac14
  \sum_{i,j}
  \delta_i(G)D^{\rm Dob}_{ij}\delta_j(F)
  \label{eq:supp-Follmer-covariance}
\end{equation}
for classical functions $F$ and $G$~\cite{main-Follmer1982,main-RebeschiniVanHandel2014}.
For $F=\bar r_X$ and a diagonal Hermitian contraction $G=O_B$, respectively,
\begin{equation}
  \delta_j(\bar r_X)
  \le
  2J_X\mathbf1_{j\in X},
  \qquad
  \delta_i(O_B)
  \le
  2\mathbf1_{i\in B}.
  \label{eq:supp-Dobrushin-oscillation-bounds}
\end{equation}
Changing a variable outside the relevant support has no effect, while changing a variable inside it can change a function by at most its full range.

Fix a diagonal Hermitian contraction $O_B$ and suppose first that $J_X>0$.
Eqs.~\eqref{eq:supp-classical-km-covariance}, \eqref{eq:supp-Follmer-covariance}, and \eqref{eq:supp-Dobrushin-oscillation-bounds} give
\begin{equation*}
  \abs{\cK_s(\bar r_X,O_B)}
  \le
  J_X
  \sum_{i\in B}\sum_{j\in X}D^{\rm Dob}_{ij}.
\end{equation*}
After division by $J_X$, the right-hand side is independent of both $O_B$ and $s$.
Taking the supremum over $O_B$ and integrating over the path therefore yields
\begin{equation}
  \bar\chi(X\!\to\!B)
  \le
  \sum_{i\in B}\sum_{j\in X}D^{\rm Dob}_{ij}.
  \label{eq:supp-Dobrushin-response-profile}
\end{equation}
The same conclusion holds when $J_X=0$ by the zero-profile convention following Eq.~\maineqref{eq:path-averaged-response-main} of the main text.

Multiplying Eq.~\eqref{eq:supp-Dobrushin-response-profile} by $J_X$, summing over the residual supports, and interchanging the finite sums gives
\begin{equation*}
  \sum_XJ_X\bar\chi(X\!\to\!B)
  \le
  \sum_{i\in B}\sum_jD^{\rm Dob}_{ij}
  \sum_{X\ni j}J_X.
\end{equation*}
The incident-strength assumption in Eq.~\maineqref{eq:incident-residual-strength} of the main text implies
\begin{equation*}
  \sum_XJ_X\bar\chi(X\!\to\!B)
  \le
  \varepsilon_R
  \sum_{i\in B}\sum_jD^{\rm Dob}_{ij}.
\end{equation*}
The Neumann expansion in Eq.~\eqref{eq:supp-Dobrushin-resolvent} and the row-sum bound give
\begin{equation}
  \sum_jD^{\rm Dob}_{ij}
  \le
  \sum_{m=0}^{\infty}\alpha^m
  =
  \frac{1}{1-\alpha}.
  \label{eq:supp-Dobrushin-row-sum}
\end{equation}
Consequently,
\begin{equation}
  \sum_XJ_X\bar\chi(X\!\to\!B)
  \le
  \frac{\varepsilon_R\abs{B}}{1-\alpha}.
  \label{eq:supp-Dobrushin-weighted-response}
\end{equation}
Substitution into Theorem~\ref*{main-thm:response-profile-main} of the main text proves
\begin{equation}
  \norm{[\rho_1-\rho_0]_B}_1
  \le
  \min\left\{2,
    \frac{\beta\varepsilon_R\abs{B}}{1-\alpha}
  \right\}.
  \label{eq:supp-Dobrushin-local-bound}
\end{equation}
At fixed $\beta$, this bound is independent of the total volume whenever $\abs{B}$, $\varepsilon_R$, and a common upper bound $\alpha<1$ are controlled uniformly along the lattice sequence.

The same argument can be converted into the termwise decay envelope of Sec.~\ref*{main-sec:shell-summability} when $C^{\rm Dob}$ has finite influence range $a_{\rm Dob} > 0$, meaning $C^{\rm Dob}_{ij}=0$ whenever $\dist(i,j)>a_{\rm Dob}$.
An $m$-step influence chain can then connect $j$ to $i$ only if $m\ge\lceil\dist(i,j)/a_{\rm Dob}\rceil$, so
\begin{equation}
  D^{\rm Dob}_{ij}
  \le
  \frac{\alpha^{\lceil\dist(i,j)/a_{\rm Dob}\rceil}}{1-\alpha}, \qquad 0<\alpha<1.
  \label{eq:supp-Dobrushin-spatial-decay}
\end{equation}
{If the residual supports satisfy $\abs{X}\le q$, then for disjoint $X$ and $B$, Eq.~\eqref{eq:supp-Dobrushin-response-profile} gives Eq.~\maineqref{eq:path-averaged-response-envelope-main} of the main text with}
\begin{equation}
  C_\beta(B)
  =
  \frac{q\abs{B}}{1-\alpha},
  \qquad
  g(r)
  =
  \alpha^{\lceil r/a_{\rm Dob}\rceil}.
  \label{eq:supp-Dobrushin-envelope}
\end{equation}

As an explicit illustration of both the criterion and its limitation, consider the zero-field nearest-neighbor ferromagnetic Ising chain with coupling $J>0$.
Changing one neighboring spin gives the standard influence bound $C^{(s)}_{ij}\le\tanh(\beta J)$~\cite{main-Georgii2011}, so an interior site obeys the elementary Dobrushin test
\begin{equation}
  2\tanh(\beta J)
  <
  1,
  \qquad\text{or}\qquad
  \beta J
  <
  \frac12\log3.
  \label{eq:supp-Ising-Dobrushin-test}
\end{equation}
Thus the Dobrushin route certifies the required response bound for this example when $\beta J<\frac12\log3$.
Outside this interval, the criterion is inconclusive and another response estimate would be needed.

\subsection{{Application of the Gibbs-state stability theorem}}
\label{sec:supp-capel-dictionary}

{We apply the local Gibbs-state stability bound of Ref.~\cite[Theorem~34]{main-CapelEtAl2025} to the centered SCTA residual.}
Let $\nu\in\mathbb N$, and let $\Lambda\Subset\mathbb Z^\nu$ carry the $\ell^1$ metric and a fixed-dimensional onsite spin space.
{For this finite-range specialization, assume that every residual support satisfies $\abs{X}\le q$ and $\diam(X)\le a_R$ for common finite bounds $q$ and $a_R$.}
Physically, an interaction $\Psi$ is the collection of local energy terms that defines the lattice model, including onsite fields and couplings among groups of sites.
For each finite $X\subseteq\Lambda$, the term $\Psi(X)$ is Hermitian and supported on $X$, and the finite-volume Hamiltonian is their sum,
\begin{equation}
  H_\Lambda(\Psi)=\sum_{X\subseteq\Lambda}\Psi(X).
\end{equation}
For a positive nonincreasing weight $F:[0,\infty)\to(0,\infty)$ with $F(0)\le1$, define the $F$-norm of the interaction $\Psi$ following Ref.~\cite{main-CapelEtAl2025} by
\begin{equation}
  \norm{\Psi}_F
  :=
  \sup_{x\in\Lambda}
  \sum_{X\ni x}
  \frac{\abs{X}\norm{\Psi(X)}_\infty}{F(\diam(X))}.
  \label{eq:supp-F-norm}
\end{equation}

\begin{lemma}[Finite-range residual conversion]
  \label{lem:supp-F-conversion}
  Let $\Psi_V$ be the perturbation interaction defined by $\Psi_V(X)=\bar r_X$, using the decomposition, centering, and incident strength in Eqs.~\maineqref{eq:local-residual-decomposition}--\maineqref{eq:incident-residual-strength} of the main text.
  Then
  \begin{equation}
    \norm{\Psi_V}_F
    \le\frac{q}{F(a_R)}\varepsilon_R.
    \label{eq:supp-F-conversion}
  \end{equation}
\end{lemma}

\begin{proof}
  Since $\abs{X}\le q$, $\diam(X)\le a_R$, and $F$ is nonincreasing, every site $x$ satisfies
  \begin{equation*}
    \sum_{X\ni x}
    \frac{\abs{X}\norm{\Psi_V(X)}_\infty}{F\!\left(\diam(X)\right)}
    \le
    \frac{q}{F(a_R)}\sum_{X\ni x}J_X
    \le\frac{q}{F(a_R)}\varepsilon_R.
  \end{equation*}
  Taking the supremum over $x$ proves the claim.
\end{proof}

{The stability theorem controls local Gibbs-state changes through $\norm{\Psi_V}_F$, while our residual estimates are expressed in terms of $\varepsilon_R$.}
{The lemma provides a conversion that lets us express the stability bound of Ref.~\cite{main-CapelEtAl2025} in terms of $\varepsilon_R$, without introducing a factor that grows with system size.}

{Using the convention of Ref.~\cite{main-CapelEtAl2025},
  we represent the core and centered residual by $H_C=H_\Lambda(\Psi_H)$ and $\bar R=\sum_X\bar r_X=H_\Lambda(\Psi_V)$, with $\Psi_V(X)=\bar r_X$.
The complete comparison path $H_C+s\bar R$ corresponds to $\Psi_H+s\Psi_V$ for every $0\le s\le1$.}
{Fix $C_{\rm int}>0$, $n\in\mathbb N_0$, and $\beta_0>0$, take $0<\beta<\beta_0$, and assume $\norm{\Psi_H}_F,\norm{\Psi_V}_F<C_{\rm int}$.}
{Along this path, the Gibbs states}
\begin{equation}
  \rho_s=\rhoG{H_C+s\bar R},
  \qquad 0\le s\le1,
\end{equation}
{must satisfy the following covariance bound for every finite $X,Y\subseteq\Lambda$ and all operators $A_X\in\mathcal A_X$ and $O_Y\in\mathcal A_Y$:}
\begin{equation}
  \abs{\operatorname{Cov}_{\rho_s}(A_X,O_Y)}
  \le{}
  \abs{X}^{n}f_{\rm Cov}(\abs{Y})
  (1+\dist(X,Y))^{-\alpha_{\rm Cov}}\norm{A_X}_\infty\norm{O_Y}_\infty,
  \label{eq:supp-Capel-covariance}
\end{equation}
where $f_{\rm Cov}:[0,\infty)\to[0,\infty)$ is continuous and
\begin{equation}
  \operatorname{Cov}_\rho(A,O)
  :=\Tr(\rho AO)-\Tr(\rho A)\Tr(\rho O),
\end{equation}
is the ordinary covariance.
{The interaction-norm and covariance bounds are assumed uniform in the chosen volume and boundary family, and the same function $f_{\rm Cov}$, support exponent $n$, decay exponent $\alpha_{\rm Cov}$, and constants apply at every $s\in[0,1]$.}

\begin{proposition}[SCTA specialization of Capel \emph{et al.}, Theorem 34]
  \label{prop:supp-Capel-specialization}
  {Assume the bounds $\norm{\Psi_H}_F,\norm{\Psi_V}_F<C_{\rm int}$ and the uniform covariance condition in Eq.~\eqref{eq:supp-Capel-covariance}, with $F$ and $\alpha_{\rm Cov}$ satisfying the decay conditions of \citet[Theorem~34]{main-CapelEtAl2025}.}
  {Then there is a constant $C_{\rm SLT}$, determined by the fixed parameters above and independent of system size, such that every nonempty $B\subseteq\Lambda$ satisfies}
  \begin{equation}
    \norm{[\rho_1-\rho_0]_B}_1
    \le
    \min\left\{2,
      C_{\rm SLT}\beta\frac{q}{F(a_R)}\varepsilon_R
      \abs{B}\bigl[1+f_{\rm Cov}(\abs{B})\bigr]
    \right\}.
    \label{eq:supp-Capel-trace-bound}
  \end{equation}
\end{proposition}

\begin{proof}
  {Write $R=\bar R+c_R\id$, where $c_R=\sum_Xc_X^\star$ is the scalar removed by centering the residual terms.
    Since $s c_R\id$ commutes with $H_C+s\bar R$, for every $s\in[0,1]$ we have $\ee^{-\beta(H_C+sR)}=\ee^{-\beta s c_R}\ee^{-\beta(H_C+s\bar R)}$.
    Taking the trace shows that the partition function acquires the same scalar factor, which cancels upon normalization:
    \begin{equation*}
      \rhoG{H_C+sR}
      =\frac{\ee^{-\beta s c_R}\ee^{-\beta(H_C+s\bar R)}}
      {\ee^{-\beta s c_R}\Tr\!\left(\ee^{-\beta(H_C+s\bar R)}\right)}
      =\rhoG{H_C+s\bar R}.
    \end{equation*}
    The stated lattice, temperature, interaction-norm, and decay conditions permit application of \citet[Theorem~34]{main-CapelEtAl2025} to the path generated by $\Psi_H+s\Psi_V$.
    For a nonempty region $B\subseteq\Lambda$, let $O_B$ be a Hermitian contraction supported on $B$, extended by the identity outside $B$.
    The defining property of the partial trace, the imported observable bound, and Lemma~\ref{lem:supp-F-conversion} give
    \begin{equation*}
      \begin{aligned}
        \abs{\Tr_B\!\left([\rho_1-\rho_0]_B O_B\right)}
        &=\abs{\Tr\!\left[(\rho_1-\rho_0)O_B\right]}\\
        &\le C_{\rm SLT}\beta\norm{\Psi_V}_F
        \abs{B}\bigl[1+f_{\rm Cov}(\abs{B})\bigr]\norm{O_B}_\infty\\
        &\le C_{\rm SLT}\beta\frac{q}{F(a_R)}\varepsilon_R
        \abs{B}\bigl[1+f_{\rm Cov}(\abs{B})\bigr],
      \end{aligned}
    \end{equation*}
    where the last step also uses $\norm{O_B}_\infty\le1$.
    Taking the supremum over these contractions and using trace-norm duality gives the residual-dependent term in Eq.~\eqref{eq:supp-Capel-trace-bound}.
  The independent bound $\norm{[\rho_1-\rho_0]_B}_1\le2$ for normalized reduced states supplies the other term in the minimum.}
\end{proof}

Sec.~\ref*{main-sec:physical-response-regimes} of the main text discusses several noncommuting quantum regimes in which established results can supply this covariance-decay condition.
\section{Exact-anchor constructions and resources}
\label{sec:supp-anchors}

This section supplies the constructions underlying the three purification-based exact anchors in Sec.~\ref*{main-sec:exact-anchors} of the main text.
{For a diagonal core $H_{C,0}=\sum_{x\in\mathcal X_C}E_C(x)\ket{x}\!\bra{x}$ with finite label set $\mathcal X_C$, their common purification is}
\begin{equation}
  {
    \begin{aligned}
      p_{\beta,C}(x)
      &=\frac{\ee^{-\beta E_C(x)}}{\sum_{y\in\mathcal X_C}\ee^{-\beta E_C(y)}},\\
      \ket{\Psi_{C,0}}_{AS}
      &=\sum_{x\in\mathcal X_C}
      \sqrt{p_{\beta,C}(x)}\ket{x}_A\ket{x}_S.
    \end{aligned}
  }
  \label{eq:supp-canonical-purification}
\end{equation}
{Tracing out the ancilla gives $\rhoG{H_{C,0}}$.
If $U_0^\dagger H_0U_0=H_{C,0}$, applying $U_0$ to the system gives $U_0\rhoG{H_{C,0}}U_0^\dagger=\rhoG{H_0}$ by unitary covariance.}
{Operational efficiency additionally requires an efficient amplitude loader for $p_{\beta,C}$, supplied by the constructions below.}

\subsection{Graph-stabilizer anchors}
\label{sec:supp-graph-anchor}

{For a finite simple graph $G=(V,E)$, use the tail $U_G$, stabilizers $K_i$, and diagonal core defined in Sec.~\ref*{main-sec:graph-anchor} of the main text.}
{The Clifford identity $K_i=U_GZ_iU_G^\dagger$ extends multiplicatively to $K_X=U_GZ_XU_G^\dagger$, proving the core identity in Eq.~\maineqref{eq:graph-anchor-identity} term by term.}
In particular, $U_G\ket{z}$ is a simultaneous graph-basis eigenstate with $K_iU_G\ket{z}=(-1)^{z_i}U_G\ket{z}$.

Suppose that $G$ has maximum degree $\Delta\geq1$.
An edge $(i,j)$ is adjacent in the line graph of $G$ to at most $(\deg i-1)+(\deg j-1)\leq2\Delta-2$ other edges.
Greedily assigning a color not used by the previously colored adjacent edges therefore uses at most $2\Delta-1$ colors.
Each color class is a matching and hence one layer of disjoint $CZ$ gates.
Adding the parallel Hadamard layer yields
\begin{equation}
  d_G\leq 2\Delta.
  \label{eq:supp-graph-depth}
\end{equation}
For $E=\varnothing$, only the Hadamard layer is required.
Consequently, if $\Delta$ and every physical edge length are bounded uniformly in volume, $U_G$ has uniformly bounded depth and gate-support diameter.
If the sets $X\in\mathcal F$ also have uniformly bounded graph diameter, the physical operators $K_X$ have uniformly bounded support diameter as well.

{We use the convention $R_y(\varphi)=\exp(-\mathrm{i}\varphi Y/2)$ throughout the loaders below.}
{For each independent probability $p$, the corresponding ancilla--system pair is loaded from $\ket{00}$ by $R_y(2\arcsin\sqrt p)$ on the ancilla followed by an ancilla-to-system CNOT.}
{For the product core, the independent probabilities $p_{\beta,i}$ and purification are given in Eqs.~\maineqref{eq:graph-defect-probability} and~\maineqref{eq:graph-product-purification} of the main text.
Its ideal loader has depth two and $2|V|$ gates, and the core Gibbs entropy is}
\begin{equation}
  S\!\left(\rhoG{H_{C,G}}\right)
  =-\sum_i\left[(1-p_{\beta,i})\ln(1-p_{\beta,i})
  +p_{\beta,i}\ln p_{\beta,i}\right].
  \label{eq:supp-graph-product-entropy}
\end{equation}

\subsubsection{Finite-memory syndrome distributions}
\label{sec:supp-finite-memory}

We now use the finite-memory transfer-matrix construction~\cite{main-Baxter1982}.
Let the syndrome labels be ordered along an open chain and fix a memory length $1\le m\le N$ independent of the system size $N$.
Every interaction is assumed to have index span at most $m$.
Grouping each interaction according to its rightmost site writes the classical core energy as
\begin{equation}
  E_C(z_1,\ldots,z_N)
  =E_{\mathrm{init}}(s_m)
  +\sum_{i=m}^{N-1}\epsilon_i(s_i,z_{i+1}),
  \qquad
  s_i=(z_{i-m+1},\ldots,z_i).
  \label{eq:supp-memory-energy}
\end{equation}
For example, for $H_{C,G}=-\sum_{X\in\mathcal F_m}\lambda_XZ_X$ one may take
\begin{align}
  E_{\mathrm{init}}(s_m)
  &=-\sum_{\substack{X\in\mathcal F_m\\\max X\leq m}}
  \lambda_X(-1)^{\sum_{j\in X}z_j},
  \label{eq:supp-memory-initial-energy}\\
  \epsilon_i(s_i,z_{i+1})
  &=-\sum_{\substack{X\in\mathcal F_m\\\max X=i+1}}
  \lambda_X(-1)^{\sum_{j\in X}z_j}.
  \label{eq:supp-memory-increment}
\end{align}
The span assumption ensures that Eq.~\eqref{eq:supp-memory-increment} depends only on $s_i$ and the appended bit.

Define the shift $\operatorname{sh}(s_i,b)=(z_{i-m+2},\ldots,z_i,b)$ and the suffix partition function
\begin{equation}
  F_i(s_i)=
  \sum_{z_{i+1},\ldots,z_N}
  \ee^{-\beta\sum_{k=i}^{N-1}\epsilon_k(s_k,z_{k+1})}.
  \label{eq:supp-suffix-partition}
\end{equation}
Here the later memories in each summand are generated recursively by $s_{k+1}=\operatorname{sh}(s_k,z_{k+1})$.
The empty suffix gives $F_N(s_N)=1$, and conditioning on the first suffix bit gives the backward recursion
\begin{equation}
  F_i(s_i)=\sum_{b=0}^{1}
  \ee^{-\beta\epsilon_i(s_i,b)}
  F_{i+1}\!\left(\operatorname{sh}(s_i,b)\right),
  \qquad i=N-1,\ldots,m.
  \label{eq:supp-backward-recursion}
\end{equation}
The partition function, initial-memory distribution, and normalized transition probabilities are consequently
\begin{align}
  Z_C
  &=\sum_{s_m}\ee^{-\beta E_{\mathrm{init}}(s_m)}F_m(s_m),
  \label{eq:supp-memory-partition}\\
  \pi_{\beta,m}(s_m)
  &=\frac{\ee^{-\beta E_{\mathrm{init}}(s_m)}F_m(s_m)}{Z_C},
  \label{eq:supp-memory-initial-distribution}\\
  T_{\beta,i}(b\mid s_i)
  &=\frac{\ee^{-\beta\epsilon_i(s_i,b)}
  F_{i+1}(\operatorname{sh}(s_i,b))}{F_i(s_i)}.
  \label{eq:supp-memory-transition}
\end{align}
All denominators are positive for finite $\beta$ and finite energies.
Eq.~\eqref{eq:supp-backward-recursion} normalizes $T_{\beta,i}(\cdot\mid s_i)$, and cancellation of the intermediate $F_i$ factors proves
\begin{equation}
  p_{\beta,C}(z_1,\ldots,z_N)
  =\pi_{\beta,m}(s_m)
  \prod_{i=m}^{N-1}T_{\beta,i}(z_{i+1}\mid s_i)
  =\frac{\ee^{-\beta E_C(z)}}{Z_C}.
  \label{eq:supp-memory-factorization}
\end{equation}

For coherent loading, first apply an $m$-qubit circuit $U_{\pi}$ satisfying
\begin{equation}
  U_{\pi}\ket{0}^{\otimes m}
  =\sum_{s_m\in\{0,1\}^m}\sqrt{\pi_{\beta,m}(s_m)}\ket{s_m}.
  \label{eq:supp-memory-initial-loader}
\end{equation}
This circuit can be implemented as a binary tree of prefix-controlled rotations.
For the $k$th qubit, let $x=(z_1,\ldots,z_{k-1})$ denote a fixed prefix and define
\begin{equation}
  \Pr_{\pi}(z_k=1\mid x)
  =
  \frac{\displaystyle\sum_{y\in\{0,1\}^{m-k}}\pi_{\beta,m}(x,1,y)}
  {\displaystyle\sum_{b\in\{0,1\}}\sum_{y\in\{0,1\}^{m-k}}\pi_{\beta,m}(x,b,y)},
  \qquad
  \vartheta_{k,x}=\arcsin\sqrt{\Pr_{\pi}(z_k=1\mid x)}.
  \label{eq:supp-memory-prefix-angle}
\end{equation}
Applying $R_y(2\vartheta_{k,x})$ to the $k$th qubit, conditioned on the preceding qubits being in $\ket{x}$, loads the required conditional probability.
The complete binary tree uses $2^m-1$ rotation angles and $O(2^m)$ ideal elementary gates~\cite{main-MottonenEtAl2005}.
After the first $i$ labels have been prepared, apply
\begin{equation}
  V_i=\sum_{s\in\{0,1\}^m}\ket{s}\!\bra{s}
  \otimes R_y(2\theta_{i,s}),
  \qquad
  \theta_{i,s}=\arcsin\sqrt{T_{\beta,i}(1\mid s)},
  \label{eq:supp-memory-multiplexer}
\end{equation}
to the memory qubits and a fresh ancilla.
For every memory string $s$, this gate acts as
\begin{equation}
  V_i\ket{s}\ket{0}
  =\ket{s}\left[
    \sqrt{T_{\beta,i}(0\mid s)}\ket{0}
    +\sqrt{T_{\beta,i}(1\mid s)}\ket{1}
  \right].
  \label{eq:supp-memory-multiplexer-action}
\end{equation}
Iterating from $i=m$ to $N-1$ prepares $\sum_z\sqrt{p_{\beta,C}(z)}\ket{z}_A$.
Parallel ancilla-to-system CNOTs then give the purification
\begin{equation}
  \ket{\Psi_{C,\mathrm M}}_{AS}
  =\sum_z\sqrt{p_{\beta,C}(z)}\ket{z}_A\ket{z}_S.
  \label{eq:supp-memory-purification}
\end{equation}
Fig.~\ref{fig:supp-memory-core-loader} illustrates the complete amplitude-loading stage and its component rotations for $m=2$.

\begin{figure}[t]
  \centering
  \resizebox{0.94\linewidth}{!}{%
    \begin{tikzpicture}[
        x=1cm,
        y=0.66cm,
        wire/.style={line width=0.45pt},
        initgate/.style={draw=NavyBlue,fill=NavyBlue!9,rounded corners=1.5pt,minimum width=1.05cm,align=center,font=\scriptsize},
        stepgate/.style={draw=Orange!80!black,fill=Orange!14,rounded corners=1.5pt,minimum width=1.18cm,align=center,font=\scriptsize},
        ryinit/.style={draw=NavyBlue,fill=NavyBlue!9,rounded corners=1.5pt,minimum width=1.55cm,minimum height=0.52cm,align=center,font=\scriptsize},
        rystep/.style={draw=Orange!80!black,fill=Orange!14,rounded corners=1.5pt,minimum width=1.48cm,minimum height=0.52cm,align=center,font=\scriptsize},
        ctrlzero/.style={circle,draw=black,fill=white,inner sep=0pt,minimum size=5.5pt},
        ctrlone/.style={circle,draw=black,fill=black,inner sep=0pt,minimum size=5.0pt}
      ]
      \node[anchor=west,font=\footnotesize] at (-0.55,5.20) {\textbf{(a)} Sequential loading of the correlated Gibbs amplitudes};

      \foreach \j/\yy in {1/4,2/3,3/2,4/1,5/0}{
        \node[anchor=east,font=\scriptsize] at (0.9,\yy) {$A_{\j}:\ket{0}$};
        \draw[wire] (1.0,\yy) -- (8.25,\yy);
      }

      \node[initgate,minimum height=1.03cm] at (1.75,3.5) {$U_{\pi}$};
      \node[stepgate,minimum height=1.60cm] at (3.35,3.0) {$V_2$\\[-2pt]\scriptsize $s_2\!\to z_3$};
      \node[stepgate,minimum height=1.60cm] at (4.95,2.0) {$V_3$\\[-2pt]\scriptsize $s_3\!\to z_4$};
      \node[stepgate,minimum height=1.60cm] at (6.55,1.0) {$V_4$\\[-2pt]\scriptsize $s_4\!\to z_5$};

      \node[anchor=west,align=left,font=\scriptsize] at (8.45,2.0) {$\sum_z\sqrt{p_{\beta,C}(z)}\ket{z}_A$\\[2pt]Gibbs-amplitude state};

      \node[anchor=west,font=\footnotesize] at (-0.55,-1.35) {\textbf{(b)} Binary-tree components of $U_{\pi}$ for $m=2$};
      \node[anchor=east,font=\scriptsize] at (0.9,-2.55) {$A_1:\ket{0}$};
      \node[anchor=east,font=\scriptsize] at (0.9,-3.55) {$A_2:\ket{0}$};
      \draw[wire] (1.0,-2.55) -- (9.65,-2.55);
      \draw[wire] (1.0,-3.55) -- (9.65,-3.55);

      \node[ryinit] at (2.05,-2.55) {$R_y(2\vartheta_{1,\varnothing})$};

      \draw[wire] (4.55,-2.55) -- (4.55,-3.55);
      \node[ctrlzero] at (4.55,-2.55) {};
      \node[ryinit] at (4.55,-3.55) {$R_y(2\vartheta_{2,0})$};

      \draw[wire] (7.05,-2.55) -- (7.05,-3.55);
      \node[ctrlone] at (7.05,-2.55) {};
      \node[ryinit] at (7.05,-3.55) {$R_y(2\vartheta_{2,1})$};

      \node[anchor=west,font=\scriptsize] at (9.85,-3.05) {$
        \begin{aligned}
          U_{\pi}:\;&\ket{00}\longmapsto\\[-1pt]
          &\textstyle\sum_{s_2\in\{0,1\}^{2}}\sqrt{\pi_{\beta,2}(s_2)}\ket{s_2}
      \end{aligned}$};

      \node[anchor=west,font=\footnotesize] at (-0.55,-5.05) {\textbf{(c)} Memory-branch components of $V_i$ for $m=2$};
      \node[anchor=east,font=\scriptsize] at (0.9,-6.20) {$A_{i-1}$};
      \node[anchor=east,font=\scriptsize] at (0.9,-7.20) {$A_i$};
      \node[anchor=east,font=\scriptsize] at (0.9,-8.20) {$A_{i+1}:\ket{0}$};
      \draw[wire] (1.0,-6.20) -- (9.75,-6.20);
      \draw[wire] (1.0,-7.20) -- (9.75,-7.20);
      \draw[wire] (1.0,-8.20) -- (9.75,-8.20);

      \foreach \xx in {1.85,4.05,6.25,8.45}{
        \draw[wire] (\xx,-6.20) -- (\xx,-8.20);
      }

      \node[ctrlzero] at (1.85,-6.20) {};
      \node[ctrlzero] at (1.85,-7.20) {};
      \node[rystep] at (1.85,-8.20) {$R_y(2\theta_{i,00})$};

      \node[ctrlzero] at (4.05,-6.20) {};
      \node[ctrlone] at (4.05,-7.20) {};
      \node[rystep] at (4.05,-8.20) {$R_y(2\theta_{i,01})$};

      \node[ctrlone] at (6.25,-6.20) {};
      \node[ctrlzero] at (6.25,-7.20) {};
      \node[rystep] at (6.25,-8.20) {$R_y(2\theta_{i,10})$};

      \node[ctrlone] at (8.45,-6.20) {};
      \node[ctrlone] at (8.45,-7.20) {};
      \node[rystep] at (8.45,-8.20) {$R_y(2\theta_{i,11})$};

      \node[anchor=west,font=\scriptsize] at (9.95,-7.20) {$
        \begin{aligned}
          V_i:\;&\ket{s_i}\ket{0}\longmapsto\\[-1pt]
          &\textstyle\ket{s_i}\sum_{b=0}^{1}\sqrt{T_{\beta,i}(b\!\mid\!s_i)}\ket{b}
      \end{aligned}$};
    \end{tikzpicture}%
  }
  \caption{Finite-memory amplitude loader illustrated for $m=2$.
    (a) For $N=5$, $U_{\pi}$ prepares the first two syndrome labels and the support of $V_i$ shifts by one site at each subsequent step.
    (b) The binary-tree decomposition of $U_{\pi}$ first loads the marginal distribution of $z_1$ and then selects the conditional rotation for $z_2$.
    (c) The four branches of $V_i$ correspond to $s_i\in\{00,01,10,11\}$; open and filled controls select memory bits equal to $0$ and $1$, respectively.
    For general $m$, each $V_i$ contains $2^m$ memory-conditioned rotations.
  The final parallel CNOT layer producing Eq.~\eqref{eq:supp-memory-purification} is not shown.}
  \label{fig:supp-memory-core-loader}
\end{figure}

The same transition tables determine the entropy.
Let $\mathsf M_i$ be the random variable for the $i$th memory string and $q_i(s)=\Pr(\mathsf M_i=s)$, with $q_m=\pi_{\beta,m}$.
The forward recursion
\begin{equation}
  q_{i+1}(s')=
  \sum_{\substack{s,b:\\ \operatorname{sh}(s,b)=s'}}
  q_i(s)T_{\beta,i}(b\mid s)
  \label{eq:supp-memory-forward}
\end{equation}
gives all memory marginals.
The chain rule then yields
\begin{equation}
  \begin{split}
    S\!\left(\rhoG{H_{C,\mathrm M}}\right)
    ={}&-\sum_s\pi_{\beta,m}(s)\ln\pi_{\beta,m}(s)\\
    &-\sum_{i=m}^{N-1}\sum_s q_i(s)
    \sum_{b=0}^{1}T_{\beta,i}(b\mid s)\ln T_{\beta,i}(b\mid s),
  \end{split}
  \label{eq:supp-memory-entropy}
\end{equation}
with the convention $0\ln0=0$.
Equivalently, $S=\ln Z_C+\beta\langle H_{C,\mathrm M}\rangle_\beta$.

If the local energy tables are supplied or evaluable in constant time per entry, Eqs.~\eqref{eq:supp-backward-recursion} and~\eqref{eq:supp-memory-forward} require $O(N2^m)$ classical time.
Storing all transition tables requires $O(N2^m)$ memory.
The coherent loader uses $O(N2^m)$ ideal elementary gates and sequential depth $O(N2^m)$.
It uses $N$ system qubits and $N$ ancilla qubits in the purification realization.
Compilation of the continuously parameterized rotations adds a precision-dependent cost and contributes to core-preparation error, independently of the bounded-depth Clifford tail.

\subsection{Independent Rydberg dimers}
\label{sec:supp-rydberg-anchor}

For one symmetric dimer, suppressing its block label temporarily, write~\cite{main-SaffmanWalkerMolmer2010,main-BrowaeysLahaye2020}
\begin{equation}
  h=\frac{\Omega}{2}(X_1+X_2)
  -\Delta(n_1+n_2)+Vn_1n_2.
  \label{eq:supp-rydberg-dimer}
\end{equation}
Here $n_j:=\ket{r}\!\bra{r}_j$ and $X_j:=\ket{g}\!\bra{r}_j+\ket{r}\!\bra{g}_j$.
Introduce
\begin{equation}
  \ket{G}=\ket{gg},\quad
  \ket{W}=\frac{\ket{gr}+\ket{rg}}{\sqrt2},\quad
  \ket{RR}=\ket{rr},\quad
  \ket{D}=\frac{\ket{gr}-\ket{rg}}{\sqrt2}.
  \label{eq:supp-rydberg-exchange-basis}
\end{equation}
The symmetric drive annihilates $\ket D$ and couples $\ket G\leftrightarrow\ket W\leftrightarrow\ket{RR}$ with matrix element $\Omega/\sqrt2$.
Therefore, in the ordered basis $(\ket G,\ket W,\ket{RR},\ket D)$,
\begin{equation}
  h=
  \begin{pmatrix}
    0&\Omega/\sqrt2&0&0\\
    \Omega/\sqrt2&-\Delta&\Omega/\sqrt2&0\\
    0&\Omega/\sqrt2&V-2\Delta&0\\
    0&0&0&-\Delta
  \end{pmatrix}.
  \label{eq:supp-rydberg-block}
\end{equation}
Thus $\ket D$ is an exact eigenstate of energy $-\Delta$: its antisymmetry makes the two drive-induced transition amplitudes cancel, leaving it decoupled from $\ket G$ and $\ket{RR}$.

Restoring the dimer label $m$, let $h_{m,+}$ be the upper-left $3\times3$ block in Eq.~\eqref{eq:supp-rydberg-block}.
Choose a real orthogonal matrix $v_m$ such that
\begin{equation}
  v_m^{\mathsf T}h_{m,+}v_m
  =\operatorname{diag}(\epsilon_{m,1},\epsilon_{m,2},\epsilon_{m,3}),
  \qquad \epsilon_{m,4}=-\Delta_m.
  \label{eq:supp-rydberg-symmetric-diagonalization}
\end{equation}
Using binary labels $\ket{1}=\ket{00}$, $\ket{2}=\ket{01}$, $\ket{3}=\ket{10}$, and $\ket{4}=\ket{11}$, define
\begin{equation}
  u_m^{\mathrm{ex}}
  =\ket{G}_m\!\bra{1}+\ket{W}_m\!\bra{2}
  +\ket{RR}_m\!\bra{3}+\ket{D}_m\!\bra{4},
  \label{eq:supp-rydberg-exchange-map}
\end{equation}
\begin{equation}
  u_m^{\mathrm{sym}}=v_m\oplus1.
  \label{eq:supp-rydberg-symmetric-map}
\end{equation}
\begin{equation}
  u_m=u_m^{\mathrm{ex}}u_m^{\mathrm{sym}}.
  \label{eq:supp-rydberg-tail-factorization}
\end{equation}
The fixed exchange map may be implemented as a Hadamard on the first qubit controlled by the second qubit, followed by $\operatorname{CNOT}_{1\rightarrow2}$.
The second factor is a real two-qubit unitary and admits a constant-size synthesis.
In particular, we choose the signs of the eigenvectors of the real symmetric $3\times3$ block so that $\det v_m=1$.
Consequently, $u_m^{\mathrm{sym}} \in \mathrm{SO}(4)$ can be synthesized using at most two CNOT gates and ideal one-qubit gates~\cite{main-VatanWilliams2004}.
Eq.~\eqref{eq:supp-rydberg-symmetric-diagonalization} implies
\begin{equation}
  u_m^\dagger h_m u_m
  =h_{C,m}:=\sum_{\alpha=1}^{4}
  \epsilon_{m,\alpha}\ket\alpha\!\bra\alpha.
  \label{eq:supp-rydberg-local-core}
\end{equation}
For disjoint dimers, $H_{\mathrm{Ry},0}=\sum_m h_m$, $U_{\mathrm{Ry}}=\bigotimes_m u_m$, and $H_{C,\mathrm{Ry}}=\sum_m h_{C,m}$, so $U_{\mathrm{Ry}}^\dagger H_{\mathrm{Ry},0}U_{\mathrm{Ry}} =H_{C,\mathrm{Ry}}$ with a volume-independent-depth tail.

The four Gibbs probabilities in dimer $m$ are
\begin{equation}
  p_{\beta,m}(\alpha)=\frac{\ee^{-\beta\epsilon_{m,\alpha}}}{Z_m},
  \qquad
  Z_m=\sum_{\alpha=1}^{4}\ee^{-\beta\epsilon_{m,\alpha}}.
  \label{eq:supp-rydberg-probabilities}
\end{equation}
Writing these probabilities as $p_{\beta,m}(b_1b_2)$, define
\begin{equation}
  \theta_{m,1}
  =\arcsin\sqrt{p_{\beta,m}(10)+p_{\beta,m}(11)}.
  \label{eq:supp-rydberg-loader-angle-one}
\end{equation}
\begin{equation}
  \theta_{m,2}^{(0)}
  =\arcsin\sqrt{\frac{p_{\beta,m}(01)}{p_{\beta,m}(00)+p_{\beta,m}(01)}}.
  \label{eq:supp-rydberg-loader-angle-two-zero}
\end{equation}
\begin{equation}
  \theta_{m,2}^{(1)}
  =\arcsin\sqrt{\frac{p_{\beta,m}(11)}{p_{\beta,m}(10)+p_{\beta,m}(11)}}.
  \label{eq:supp-rydberg-loader-angle-two-one}
\end{equation}
All denominators are positive at finite temperature.
Apply $R_y(2\theta_{m,1})$ to the first ancilla qubit and then
\begin{equation}
  V_m^{\mathrm{amp}}
  =\ket0\!\bra0\otimes R_y(2\theta_{m,2}^{(0)})
  +\ket1\!\bra1\otimes R_y(2\theta_{m,2}^{(1)})
  \label{eq:supp-rydberg-multiplexer}
\end{equation}
to the two ancilla qubits.
The latter multiplexer is implemented by $R_y(\theta_{m,2}^{(0)}+\theta_{m,2}^{(1)})$, a CNOT, then $R_y(\theta_{m,2}^{(0)}-\theta_{m,2}^{(1)})$, and a second CNOT.
Fig.~\ref{fig:supp-rydberg-core-loader} shows both the conditional-rotation form of the amplitude loader and this CNOT decomposition.

\begin{figure}[t]
  \centering
  \resizebox{0.90\linewidth}{!}{%
    \begin{tikzpicture}[
        x=1cm,
        y=0.66cm,
        wire/.style={line width=0.45pt},
        marginal/.style={draw=NavyBlue,fill=NavyBlue!9,rounded corners=1.5pt,minimum width=1.48cm,minimum height=0.52cm,align=center,font=\scriptsize},
        conditional/.style={draw=Orange!80!black,fill=Orange!14,rounded corners=1.5pt,minimum width=1.58cm,minimum height=0.52cm,align=center,font=\scriptsize},
        ctrlzero/.style={circle,draw=black,fill=white,inner sep=0pt,minimum size=5.5pt},
        ctrlone/.style={circle,draw=black,fill=black,inner sep=0pt,minimum size=5.0pt},
        target/.style={
          circle,
          draw=black,
          fill=white,
          inner sep=0pt,
          minimum size=9.0pt,
          path picture={
            \draw[wire] (path picture bounding box.west) -- (path picture bounding box.east);
            \draw[wire] (path picture bounding box.south) -- (path picture bounding box.north);
          }
        }
      ]
      \node[anchor=west,font=\footnotesize] at (-0.55,3.35) {\textbf{(a)} Loading of Gibbs amplitudes for dimer $m$};
      \node[anchor=east,font=\scriptsize] at (0.9,2.00) {$A_{m,1}:\ket{0}$};
      \node[anchor=east,font=\scriptsize] at (0.9,1.00) {$A_{m,2}:\ket{0}$};
      \draw[wire] (1.0,2.00) -- (9.75,2.00);
      \draw[wire] (1.0,1.00) -- (9.75,1.00);

      \node[marginal] at (2.15,2.00) {$R_y(2\theta_{m,1})$};

      \draw[wire] (4.75,2.00) -- (4.75,1.00);
      \node[ctrlzero] (muxzeroctrl) at (4.75,2.00) {};
      \node[conditional] (muxzerogate) at (4.75,1.00) {$R_y(2\theta_{m,2}^{(0)})$};

      \draw[wire] (7.40,2.00) -- (7.40,1.00);
      \node[ctrlone] (muxonectrl) at (7.40,2.00) {};
      \node[conditional] (muxonegate) at (7.40,1.00) {$R_y(2\theta_{m,2}^{(1)})$};

      \node[
        draw=Orange!80!black,
        densely dashed,
        rounded corners=2pt,
        inner xsep=0.20cm,
        inner ysep=0.16cm,
        fit=(muxzeroctrl)(muxzerogate)(muxonectrl)(muxonegate)
      ] (muxbox) {};
      \node[anchor=south east,font=\scriptsize,text=Orange!80!black] at ([xshift=-2pt,yshift=2pt]muxbox.north east) {$V_m^{\mathrm{amp}}$};
      \draw[->,densely dashed,draw=Orange!80!black,line width=0.45pt]
      (muxbox.south) -- ++(0,-1.15)
      node[midway,right,font=\scriptsize,text=Orange!80!black] {expanded in (b)};

      \node[anchor=west,align=left,font=\scriptsize] at (10.00,1.50) {$\displaystyle
      \sum_{\alpha=1}^{4}\sqrt{p_{\beta,m}(\alpha)}\ket{\alpha}_{A_m}$};

      \node[anchor=west,font=\footnotesize] at (-0.55,-1.30) {\textbf{(b)} CNOT decomposition of $V_m^{\mathrm{amp}}$};
      \node[anchor=east,font=\scriptsize] at (0.9,-2.65) {$A_{m,1}$};
      \node[anchor=east,font=\scriptsize] at (0.9,-3.65) {$A_{m,2}$};
      \draw[wire] (1.0,-2.65) -- (10.40,-2.65);
      \draw[wire] (1.0,-3.65) -- (10.40,-3.65);

      \node[conditional,minimum width=2.13cm] at (2.45,-3.65) {$R_y(\theta_{m,2}^{(0)}+\theta_{m,2}^{(1)})$};

      \draw[wire] (4.80,-2.65) -- (4.80,-3.65);
      \node[ctrlone] at (4.80,-2.65) {};
      \node[target] at (4.80,-3.65) {};

      \node[conditional,minimum width=2.13cm] at (7.15,-3.65) {$R_y(\theta_{m,2}^{(0)}-\theta_{m,2}^{(1)})$};

      \draw[wire] (9.45,-2.65) -- (9.45,-3.65);
      \node[ctrlone] at (9.45,-2.65) {};
      \node[target] at (9.45,-3.65) {};

      \node[anchor=west,font=\scriptsize] at (10.65,-3.15) {$=V_m^{\mathrm{amp}}$};
    \end{tikzpicture}%
  }
  \caption{Core-loading circuit for one Rydberg dimer.
    (a) The first rotation loads the marginal distribution of the first binary energy-label bit, while the open and filled controls select the two conditional rotations for the second bit.
    The dashed enclosure denotes their product $V_m^{\mathrm{amp}}$.
    (b) The same conditional multiplexer is implemented using two unconditional $R_y$ rotations and two CNOT gates.
  The subsequent ancilla-to-system basis-correlation layer is not shown.}
  \label{fig:supp-rydberg-core-loader}
\end{figure}

Two parallel ancilla-to-system CNOTs produce
\begin{equation}
  \ket{\Psi_{C,\mathrm{Ry}}}_{AS}
  =\bigotimes_m\left[
    \sum_{\alpha=1}^{4}\sqrt{p_{\beta,m}(\alpha)}
  \ket{\alpha}_{A_m}\ket{\alpha}_{S_m}\right].
  \label{eq:supp-rydberg-purification}
\end{equation}
Thus each dimer loader uses three $R_y$ rotations and four CNOTs.
With the ancilla block placed near its system dimer, all loaders and, subsequently, all tails can be scheduled blockwise in parallel.
The construction uses $2M$ system and $2M$ ancilla qubits, $O(M)$ gates, and $O(M)$ classical preprocessing and storage at fixed precision.
The entropy is
\begin{equation}
  S\!\left(\rhoG{H_{C,\mathrm{Ry}}}\right)
  =-\sum_m\sum_{\alpha=1}^{4}
  p_{\beta,m}(\alpha)\ln p_{\beta,m}(\alpha).
  \label{eq:supp-rydberg-entropy}
\end{equation}
Approximate loader rotations and approximate synthesis of $u_m$ enter the core-preparation and tail-implementation errors, respectively.

\subsection{Local Gaussian anchors}
\label{sec:supp-gaussian-anchor}

\subsubsection{{Majorana canonical form and occupation loader}}

{For fermionic modes $f_j$, define}
\begin{equation}
  {
    \eta_{2j-1}=f_j+f_j^\dagger,
    \qquad
    \eta_{2j}=-\mathrm{i}(f_j-f_j^\dagger),
    \qquad
    \{\eta_a,\eta_b\}=2\delta_{ab}\id.
  }
  \label{eq:supp-majoranas}
\end{equation}
{After omitting an additive scalar, every parity-preserving quadratic Hamiltonian can be written~\cite{main-TerhalDiVincenzo2002,main-Bravyi2005FLO}}
\begin{equation}
  {
    H_{0,F}=\frac{\mathrm{i}}{4}
    \bm\eta^{\mathsf T}\mathsf A\bm\eta,
    \qquad \mathsf A^{\mathsf T}=-\mathsf A\in\mathbb R^{2N\times2N}.
  }
  \label{eq:supp-majorana-hamiltonian}
\end{equation}
{The real skew-canonical form in Eq.~\maineqref{eq:fermion-canonical-form} of the main text supplies $\mathsf O_F\in\mathrm{SO}(2N)$ and real $\epsilon_j$.}
{Introduce core reference modes $c_j$ with Majoranas $\gamma_{2j-1}=c_j+c_j^\dagger$ and $\gamma_{2j}=-\mathrm{i}(c_j-c_j^\dagger)$.}
{A Gaussian unitary satisfying}
\begin{equation}
  {U_F^\dagger\bm\eta U_F=\mathsf O_F\bm\gamma}
  \label{eq:supp-gaussian-frame-map}
\end{equation}
{gives}
\begin{equation}
  {
    U_F^\dagger H_{0,F}U_F
    =\frac{\mathrm{i}}{2}\sum_j\epsilon_j
    \gamma_{2j-1}\gamma_{2j}
    =\sum_j\epsilon_j\left(n_{C,j}-\frac12\id\right),
  }
  \label{eq:supp-gaussian-core}
\end{equation}
{where $n_{C,j}=c_j^\dagger c_j$ and $\mathrm{i}\gamma_{2j-1}\gamma_{2j}=2n_{C,j}-\id$.}

{In the unrestricted Fock-space Gibbs ensemble, the occupations are independent, with probabilities $p_{\beta,j}=\Pr(n_{C,j}=1)$ given by Eq.~\maineqref{eq:fermi-factor-main} of the main text.}
{The product loader applies $R_y(2\arcsin\sqrt{p_{\beta,j}})$ to ancilla $A_j$, followed by a CNOT from $A_j$ to $S_j$, giving}
\begin{equation}
  {
    \ket{\Psi_{C,F}}_{AS}
    =\bigotimes_j\Bigl[
      \sqrt{1-p_{\beta,j}}\ket0_{A_j}\ket0_{S_j}
      +\sqrt{p_{\beta,j}}\ket1_{A_j}\ket1_{S_j}
    \Bigr].
  }
  \label{eq:supp-gaussian-purification}
\end{equation}
{The loader has depth two and $O(N)$ size in the ideal qubit gate model, and the core Gibbs entropy is}
\begin{equation}
  {
    S\!\left(\rhoG{H_{C,F}}\right)
    =-\sum_j\left[(1-p_{\beta,j})\ln(1-p_{\beta,j})
    +p_{\beta,j}\ln p_{\beta,j}\right].
  }
  \label{eq:supp-gaussian-entropy}
\end{equation}
{A fixed total-particle-number or global parity sector instead imposes correlations and requires a different loader.}

\subsubsection{Encoded locality: sufficient assumptions}
\label{sec:supp-fermion-locality}

Algebraic Gaussian diagonalization does not by itself provide an operational tail.
This is particularly important under a Jordan--Wigner encoding.
Order the modes along a one-dimensional qubit chain and set~\cite{main-JordanWigner1928,main-BravyiKitaev2002} $\mathcal Z_{<j}=\prod_{k<j}Z_k$.
With $\ket0$ empty and $\ket1$ occupied,
\begin{align}
  c_j&\longmapsto \mathcal Z_{<j}\frac{X_j+\mathrm{i}Y_j}{2},
  &c_j^\dagger&\longmapsto
  \mathcal Z_{<j}\frac{X_j-\mathrm{i}Y_j}{2},
  \label{eq:supp-jw-complex}\\
  \gamma_{2j-1}&\longmapsto\mathcal Z_{<j}X_j,
  &\gamma_{2j}&\longmapsto\mathcal Z_{<j}Y_j.
  \label{eq:supp-jw-majorana}
\end{align}
An individual odd fermionic operator therefore carries a parity string whose length can grow with $j$.
{For a parity-even fermionic monomial supported within $[j,\ell]$, the Jordan--Wigner factors on every qubit to the left of $j$ occur an even number of times and cancel.}
{Its encoded support therefore lies within $[j,\ell]$, and the same conclusion holds for linear combinations of such monomials.}

For distinct Majoranas, let $\Gamma_p$ denote their encoded Pauli strings and define $P_{pq}=\mathrm{i}\Gamma_p\Gamma_q$.
Then
\begin{equation}
  \exp\!\left(\frac{\theta}{2}\gamma_p\gamma_q\right)
  \longmapsto
  \exp\!\left(-\frac{\mathrm{i}\theta}{2}P_{pq}\right).
  \label{eq:supp-encoded-givens}
\end{equation}
{If the fermionic modes associated with $\gamma_p$ and $\gamma_q$ lie in $[j,\ell]$, this cancellation argument places the encoded gate on that interval.}
This yields the following sufficient operational condition: the Gaussian tail must be supplied as a volume-independent number $d_F$ of layers of parity-even Gaussian gates, each gate acting on modes within an interval of uniformly bounded diameter $a_F$, with the corresponding intervals pairwise disjoint in each layer.
The encoded qubit circuit then has depth $d_F$ and bounded gate-support diameter $a_F$.
Equivalently, one may supply the bounded-depth circuit directly in terms of parity-even encoded local gates.
Under either formulation, the certification theorem applies to parity-even observables that are local in the encoded qubit metric.

No inference about the support of the odd operator $U_F^\dagger f_j U_F$ is made from bounded qubit-circuit depth: both $f_j$ and its conjugate may retain a long Jordan--Wigner string.
In a native fermionic description, bounded-depth gates on bounded fermionic regions do propagate native fermionic support only a bounded distance, but applying the qubit certification framework still requires the parity-even encoded-locality condition above.
A generic Givens factorization of $\mathsf O_F\in\mathrm{SO}(2N)$ can contain distant rotations and have $O(N^2)$ total gate count with generally volume-growing depth, and therefore does not establish this condition.

\subsubsection{An explicit bounded-block hopping model}

Partition $2M$ physical modes into disjoint, bounded-diameter dimers with labels $L,R$, and take
\begin{equation}
  H_{\mathrm{hop}}
  =\sum_{m=1}^{M}\left[
    -t_m(f_{m,L}^\dagger f_{m,R}+f_{m,R}^\dagger f_{m,L})
    -\mu_m(n_{m,L}^{(f)}+n_{m,R}^{(f)})
  \right],
  \qquad t_m\geq0.
  \label{eq:supp-hopping-dimers}
\end{equation}
where $n_{m,\alpha}^{(f)}:=f_{m,\alpha}^\dagger f_{m,\alpha}$ for $\alpha\in\{L,R\}$.
The bonding and antibonding modes $b_{m,\pm}=(f_{m,L}\pm f_{m,R})/\sqrt2$ have one-particle energies $-\mu_m-t_m$ and $-\mu_m+t_m$, respectively.
Let $u_m^{\mathrm{hop}}$ be the two-mode beam splitter satisfying
\begin{equation}
  u_m^{\mathrm{hop}} c_{m,L}(u_m^{\mathrm{hop}})^\dagger=b_{m,+},
  \qquad
  u_m^{\mathrm{hop}} c_{m,R}(u_m^{\mathrm{hop}})^\dagger=b_{m,-},
  \label{eq:supp-hopping-beamsplitter}
\end{equation}
and set $U_{\mathrm{hop}}=\bigotimes_m u_m^{\mathrm{hop}}$.
Direct substitution, with $n_{C,m,\alpha}:=c_{m,\alpha}^\dagger c_{m,\alpha}$, gives
\begin{equation}
  U_{\mathrm{hop}}^\dagger H_{\mathrm{hop}}U_{\mathrm{hop}}
  =\sum_m\left[
    (-\mu_m-t_m)n_{C,m,L}
    +(-\mu_m+t_m)n_{C,m,R}
  \right].
  \label{eq:supp-hopping-core}
\end{equation}
The tail is depth one in a native two-mode gate model.
It is also a bounded qubit gate under an aligned Jordan--Wigner encoding when the two modes in each dimer are adjacent and different dimers occupy disjoint adjacent intervals.
{Eq.~\maineqref{eq:fermi-factor-main} of the main text}, with the two signed one-particle energies in Eq.~\eqref{eq:supp-hopping-core}, supplies the exact product core loader.
The same reasoning extends to mutually decoupled blocks whose number of modes and physical diameter are bounded uniformly in volume: diagonalize each block, compile its parity-even Gaussian unitary locally, and apply all block unitaries in parallel.
\section{Local reduction}
\label{sec:supp-reduction}

{This section supplies the local bookkeeping estimates, the proof of the \maintextlink{thm:first-order-local-reduction}{volume-uniform first-order residual theorem}, three representative first-order anchor reductions, and the proof of the \maintextlink{prop:formal-higher-order-reduction}{fixed-volume higher-order recursion}.}
{We use the decomposition-dependent quantities $\norm{A}_{\mathrm{loc}}$, $q(A)$, and $a(A)$ defined in Eq.~\maineqref{eq:declared-local-data} of the main text for a declared decomposition $A=\sum_Xa_X$.}
After a Hermitian residual has been decomposed, terms with equal support are combined and centered.
Since centering minimizes the distance to a scalar on that support, its intrinsic incident strength obeys
\begin{equation}
  \varepsilon_R\le\norm{R}_{\mathrm{loc}}.
  \label{eq:supp-centered-local}
\end{equation}

\subsection{Local bookkeeping estimates}
\label{sec:supp-local-bookkeeping}

\begin{lemma}[Local commutator estimate]
  \label{lem:supp-local-commutator}
  Let $A=\sum_Xa_X$ and $B=\sum_Yb_Y$ have declared local decompositions.
  Using these decompositions, expand
  \begin{equation*}
    [A,B]
    =
    \sum_{X,Y:\,X\cap Y\neq\varnothing}[a_X,b_Y],
  \end{equation*}
  and treat each pairwise commutator as one term, supported on $X\cup Y$, in the induced local decomposition of $[A,B]$.
  Then
  \begin{equation}
    \norm{[A,B]}_{\mathrm{loc}}
    \le
    2\left[q(A)+q(B)\right]
    \norm{A}_{\mathrm{loc}}
    \norm{B}_{\mathrm{loc}}.
    \label{eq:supp-commutator-bound}
  \end{equation}
  Every resulting term has support cardinality at most $q(A)+q(B)-1$ and diameter at most $a(A)+a(B)$.
\end{lemma}

\begin{proof}
  Bilinearity gives $[A,B]=\sum_{X,Y}[a_X,b_Y]$.
  If $X\cap Y=\varnothing$, then $a_X$ and $b_Y$ act on disjoint supports and commute, so the corresponding term vanishes.
  This proves the displayed expansion in the lemma.

  We next estimate the local decomposition norm of this expansion.
  Fix a site $x$.
  A commutator term supported on $X\cup Y$ contributes to the local sum at $x$ only if $x\in X$ or $x\in Y$.
  Hence
  \begin{equation*}
    \begin{aligned}
      &\sum_{\substack{X,Y:\,X\cap Y\neq\varnothing\\x\in X\cup Y}}
      \norm{[a_X,b_Y]}_{\infty}
      \\
      &\quad\le
      \sum_{X\ni x}\sum_{Y:\,Y\cap X\neq\varnothing}
      \norm{[a_X,b_Y]}_{\infty}
      +
      \sum_{Y\ni x}\sum_{X:\,X\cap Y\neq\varnothing}
      \norm{[a_X,b_Y]}_{\infty}.
    \end{aligned}
  \end{equation*}
  Terms for which $x\in X\cap Y$ appear in both sums on the right, but this double counting is harmless because we seek an upper bound.

  Consider the first sum.
  The operator-norm inequality $\norm{[a_X,b_Y]}_{\infty}\le2\norm{a_X}_{\infty}\norm{b_Y}_{\infty}$ gives
  \begin{equation*}
    \begin{aligned}
      \sum_{X\ni x}\sum_{Y:\,Y\cap X\neq\varnothing}
      \norm{[a_X,b_Y]}_{\infty}
      &\le
      2\sum_{X\ni x}\norm{a_X}_{\infty}
      \sum_{Y:\,Y\cap X\neq\varnothing}\norm{b_Y}_{\infty}
      \\
      &\le
      2\sum_{X\ni x}\norm{a_X}_{\infty}
      \sum_{y\in X}\sum_{Y\ni y}\norm{b_Y}_{\infty}
      \\
      &\le
      2\sum_{X\ni x}\norm{a_X}_{\infty}
      \abs{X}\norm{B}_{\mathrm{loc}}
      \\
      &\le
      2q(A)\norm{A}_{\mathrm{loc}}\norm{B}_{\mathrm{loc}}.
    \end{aligned}
  \end{equation*}
  The second inequality uses the fact that every $Y$ intersecting $X$ contains at least one site $y\in X$.
  Summing over all $y\in X$ may count such a $Y$ more than once and therefore gives an upper bound.
  The third inequality follows from $\sum_{Y\ni y}\norm{b_Y}_{\infty}\le\norm{B}_{\mathrm{loc}}$ for every $y$: summing this bound over the $\abs{X}$ sites in $X$ produces the factor $\abs{X}$.
  Finally, $\abs{X}\le q(A)$ and $\sum_{X\ni x}\norm{a_X}_{\infty}\le\norm{A}_{\mathrm{loc}}$.

  Interchanging $A$ and $B$ bounds the second sum by $2q(B)\norm{A}_{\mathrm{loc}}\norm{B}_{\mathrm{loc}}$.
  Adding the two estimates and taking the supremum over $x$ proves Eq.~\eqref{eq:supp-commutator-bound}.

  It remains to verify the support bounds.
  Each nonzero commutator $[a_X,b_Y]$ is supported on $X\cup Y$.
  Because $X$ and $Y$ intersect, $\abs{X\cup Y}\le\abs{X}+\abs{Y}-1$.
  Thus its support cardinality is at most $q(A)+q(B)-1$.
  To bound its diameter, choose $z\in X\cap Y$.
  If two points both lie in $X$ or both lie in $Y$, their distance is at most $a(A)$ or $a(B)$, respectively.
  For a point $u\in X$ and a point $v\in Y$, the triangle inequality gives
  \begin{equation*}
    \dist(u,v)
    \le
    \dist(u,z)+\dist(z,v)
    \le
    a(A)+a(B).
  \end{equation*}
  Therefore $\diam(X\cup Y)\le a(A)+a(B)$, completing the proof.
\end{proof}

\begin{lemma}[One disjoint conjugation layer]
  \label{lem:supp-disjoint-layer}
  Let $T=\sum_{Y\in\mathcal C}S_Y$, where the anti-Hermitian $S_Y$ have pairwise disjoint supports, $|Y|\le q_s$, and $\diam(Y)\le a_s$.
  For the induced termwise decomposition of $\ee^{-uT}A\ee^{uT}$ with $u\in\mathbb R$, the local decomposition norm obeys
  \begin{equation}
    \norm{\ee^{-uT}A\ee^{uT}}_{\mathrm{loc}}
    \le q_s\norm{A}_{\mathrm{loc}}.
    \label{eq:supp-layer-bound}
  \end{equation}
  Each output term has cardinality at most $q_sq(A)$ and diameter at most $a(A)+2a_s$.
\end{lemma}

\begin{proof}
  Because the supports in $\mathcal C$ are disjoint, the operators $S_Y$ commute with one another and the layer unitary factorizes as
  \begin{equation*}
    \ee^{uT}
    =
    \prod_{Y\in\mathcal C}\ee^{uS_Y}.
  \end{equation*}
  Fix one input term $a_X$.
  If $Y\cap X=\varnothing$, then $S_Y$ commutes with $a_X$, so the corresponding unitary factors cancel from its conjugation.
  Therefore
  \begin{equation*}
    a_X^{(T)}(u)
    =
    \left(
      \prod_{\substack{Y\in\mathcal C\\Y\cap X\neq\varnothing}}
      \ee^{-uS_Y}
    \right)
    a_X
    \left(
      \prod_{\substack{Y\in\mathcal C\\Y\cap X\neq\varnothing}}
      \ee^{uS_Y}
    \right).
  \end{equation*}
  Thus $a_X^{(T)}(u)$ is supported on
  \begin{equation*}
    \widetilde X^{(T)}
    :=
    X
    \cup
    \bigcup_{\substack{Y\in\mathcal C\\Y\cap X\neq\varnothing}}Y.
  \end{equation*}

  Let $r_X$ be the number of layer supports $Y$ that intersect $X$.
  Since these supports are pairwise disjoint, each one contains a distinct site of $X$, and hence $r_X\le\abs{X}$.
  Moreover, an intersecting support $Y$ adds at most $q_s-1$ new sites to $X$.
  It follows that
  \begin{equation*}
    \abs{\widetilde X^{(T)}}
    \le
    \abs{X}+r_X(q_s-1)
    \le
    q_s\abs{X}
    \le
    q_sq(A).
  \end{equation*}
  For the diameter bound, every point of $\widetilde X^{(T)}$ lies within distance at most $a_s$ of some point of $X$.
  The triangle inequality therefore gives
  \begin{equation*}
    \diam(\widetilde X^{(T)})
    \le
    a_s+a(A)+a_s
    =
    a(A)+2a_s.
  \end{equation*}

  We now bound the local decomposition norm.
  Since $T$ is anti-Hermitian, $\ee^{uT}$ is unitary, and hence
  \begin{equation*}
    \norm{a_X^{(T)}(u)}_{\infty}
    =
    \norm{a_X}_{\infty}.
  \end{equation*}
  Fix an output site $x$.
  If $x$ lies in no layer support, i.e., $x\notin\bigcup_{Y\in\mathcal C}Y$, then $x\in\widetilde X^{(T)}$ only if $x\in X$.
  The incident output sum at $x$ is consequently at most
  \begin{equation*}
    \sum_{X\ni x}\norm{a_X}_{\infty}
    \le
    \norm{A}_{\mathrm{loc}}.
  \end{equation*}
  If instead $x$ lies in a layer support, disjointness makes this support a unique set $Y_x$.
  In this case $x\in\widetilde X^{(T)}$ precisely when $X\cap Y_x\neq\varnothing$.
  The corresponding incident sum therefore obeys
  \begin{equation*}
    \begin{aligned}
      \sum_{X:\,x\in\widetilde X^{(T)}}
      \norm{a_X^{(T)}(u)}_{\infty}
      &=
      \sum_{X:\,X\cap Y_x\neq\varnothing}
      \norm{a_X}_{\infty}
      \\
      &\le
      \sum_{y\in Y_x}\sum_{X\ni y}\norm{a_X}_{\infty}
      \\
      &\le
      \abs{Y_x}\norm{A}_{\mathrm{loc}}
      \\
      &\le
      q_s\norm{A}_{\mathrm{loc}},
    \end{aligned}
  \end{equation*}
  where the first inequality follows from the fact that every $X$ intersecting $Y_x$ contains at least one site $y\in Y_x$, and summing over all $y\in Y_x$ may count such an $X$ more than once.
  The first case satisfies the same bound because $q_s\ge1$.
  Taking the supremum over $x$ proves Eq.~\eqref{eq:supp-layer-bound} and completes the proof.
\end{proof}

\subsection{Proof of Theorem~\ref*{main-thm:first-order-local-reduction}}
\label{sec:supp-first-order-proof}

\begin{proof}
  Adopt the notation and hypotheses of Theorem~\ref*{main-thm:first-order-local-reduction} in the main text.
  By assumption (i), choose $N$-independent bounds $j_0$, $j_V$, $q_h$, and $a_h$ such that
  \begin{equation*}
    \begin{aligned}
      \norm{H_{C,0}}_{\mathrm{loc}}&\le j_0,
      &\norm{V_C}_{\mathrm{loc}}&\le j_V,\\
      q(H_{C,0}),q(V_C)&\le q_h,
      &a(H_{C,0}),a(V_C)&\le a_h.
    \end{aligned}
  \end{equation*}
  By assumption (ii), choose $N$-independent schedule bounds $m$, $q_s$, $a_s$, and $s_0$ as defined in {Sec.~\ref*{main-sec:first-order-residual-reduction}} of the main text.
  First observe that no separate locality assumption on $P_C(V_C)$ is required.
  The scheduled-solution hypothesis gives
  \begin{equation*}
    \norm{S_1}_{\mathrm{loc}}
    \le
    \sum_{j=1}^{m}\norm{T_j}_{\mathrm{loc}}
    \le
    ms_0,
    \qquad
    q(S_1)\le q_s,
    \qquad
    a(S_1)\le a_s.
  \end{equation*}
  The strong homological equation~\maineqref{eq:strong-homological-main} implies
  \begin{equation*}
    P_C(V_C)
    =
    V_C+[H_{C,0},S_1].
  \end{equation*}
  Lemma~\ref{lem:supp-local-commutator} therefore supplies a volume-uniform local decomposition of $P_C(V_C)$:
  \begin{equation*}
    \norm{P_C(V_C)}_{\mathrm{loc}}
    \le
    j_V
    +
    2(q_h+q_s)j_0ms_0.
  \end{equation*}
  The same lemma gives $q(P_C(V_C))\le q_h+q_s-1$ and $a(P_C(V_C))\le a_h+a_s$ for the induced decomposition.
  Hence $H_{C,\lambda}^{(1)}=H_{C,0}+\lambda P_C(V_C)$ has uniformly bounded local norm, term cardinality, and term diameter throughout $|\lambda|\le\lambda_0$.

  Define
  \begin{equation*}
    H_{\mathrm{red}}(u)
    :=
    \left(W_{u,\mathrm{loc}}^{(1)}\right)^{\dagger}
    \left(H_{C,0}+uV_C\right)
    W_{u,\mathrm{loc}}^{(1)}.
  \end{equation*}
  At $u=0$, every circuit layer is the identity and
  \begin{equation*}
    \left.\frac{d}{du}W_{u,\mathrm{loc}}^{(1)}\right|_{u=0}
    =
    \sum_{j=1}^{m}T_j
    =
    S_1.
  \end{equation*}
  The product rule and the strong homological equation~\maineqref{eq:strong-homological-main} therefore give
  \begin{equation*}
    \begin{aligned}
      H_{\mathrm{red}}(0)
      &=
      H_{C,0},
      \\
      H_{\mathrm{red}}'(0)
      &=
      V_C+[H_{C,0},S_1]
      =
      P_C(V_C).
    \end{aligned}
  \end{equation*}
  By the definition of the exact residual, $R_{\lambda}^{(1)}=H_{\mathrm{red}}(\lambda)-H_{C,\lambda}^{(1)}$.
  Taylor's theorem with integral remainder and the two identities above therefore give
  \begin{equation}
    R_{\lambda}^{(1)}
    =
    \lambda^2
    \int_0^1(1-t)H_{\mathrm{red}}''(t\lambda)\,dt.
    \label{eq:supp-exact-remainder}
  \end{equation}
  It remains to bound the second derivative in the local decomposition norm.

  For each layer $j$, define
  \begin{equation*}
    \Phi_j(u)(O)
    :=
    \ee^{-uT_j}O\ee^{uT_j},
    \qquad
    \mathcal T_j(O)
    :=
    [O,T_j].
  \end{equation*}
  Define the superoperator
  \begin{equation*}
    F(u)
    :=
    \Phi_m(u)\cdots\Phi_1(u)
  \end{equation*}
  so that
  \begin{equation*}
    H_{\mathrm{red}}(u)
    =
    F(u)\left(H_{C,0}+uV_C\right).
  \end{equation*}
  Direct differentiation gives $\Phi_j'(u)=\Phi_j(u)\mathcal T_j=\mathcal T_j\Phi_j(u)$.
  Since the bare Hamiltonian is affine in $u$, the product rule gives
  \begin{equation*}
    H_{\mathrm{red}}''(u)
    =
    F''(u)\left(H_{C,0}+uV_C\right)
    +
    2F'(u)(V_C).
  \end{equation*}
  Differentiating the $m$ factors in $F(u)$ shows that this expression is a sum of two types of terms:
  \begin{enumerate}
      \renewcommand{\labelenumi}{(\alph{enumi})}
    \item $m^2$ terms containing two commutator superoperators $\mathcal T_j$ acting on $H_{C,0}+uV_C$; and
    \item $2m$ terms containing one $\mathcal T_j$ acting together with the derivative $V_C$ of the bare Hamiltonian.
  \end{enumerate}
  Every term also contains at most one exact conjugation $\Phi_j(u)$ from each color layer.

  We now give a uniform bound for all these terms.
  Set
  \begin{equation*}
    q_*
    :=
    \left[q_h+2(q_s-1)\right]q_s^m,
    \qquad
    d_*
    :=
    2(q_*+q_s)s_0.
  \end{equation*}
  Lemma~\ref{lem:supp-local-commutator} implies that each occurrence of $\mathcal T_j$ multiplies the local norm by at most $d_*$ once all intermediate supports are bounded by $q_*$.
  Lemma~\ref{lem:supp-disjoint-layer}, applied successively to the $m$ colors, multiplies the local norm by at most $q_s^m$.
  More explicitly, the two commutators enlarge an input support cardinality by at most $2(q_s-1)$, while each of the at most $m$ exact layers multiplies the current cardinality by at most $q_s$.
  Hence $q_*=[q_h+2(q_s-1)]q_s^m$ bounds every term produced by two commutators and all exact layers.
  Therefore, for $\abs{u}\le\lambda_0$,
  \begin{equation}
    \norm{H_{\mathrm{red}}''(u)}_{\mathrm{loc}}
    \le
    K_2,
    \label{eq:supp-second-derivative-bound}
  \end{equation}
  with the explicit conservative constant
  \begin{equation*}
    K_2
    :=
    q_s^m
    \left[
      m^2d_*^2(j_0+\lambda_0j_V)
      +
      2md_*j_V
    \right].
  \end{equation*}
  The corresponding term diameters are bounded by
  \begin{equation*}
    a_*
    :=
    a_h+2(m+1)a_s,
  \end{equation*}
  since a commutator joins overlapping supports and each disjoint conjugation layer enlarges a support only through gates that meet it.

  Insert Eq.~\eqref{eq:supp-second-derivative-bound} into Eq.~\eqref{eq:supp-exact-remainder}.
  Because $\int_0^1(1-t)\,dt=1/2$, the resulting declared decomposition satisfies
  \begin{equation*}
    \norm{R_{\lambda}^{(1)}}_{\mathrm{loc}}
    \le
    \frac{K_2}{2}\lambda^2.
  \end{equation*}
  At each finite volume, the terms in the differentiated decomposition are indexed by finitely many original Hamiltonian terms, generator supports, and color choices.
  Integrating each such term over $t$ therefore produces a declared local decomposition of the exact remainder.
  Since $R_{\lambda}^{(1)}$ is Hermitian, replacing every integrated term $r_X$ by $(r_X+r_X^{\dagger})/2$ preserves their total sum and does not increase any operator norm or support.
  Combining contributions with the same support also preserves the preceding incident bound by the triangle inequality.
  Set $C_{\mathrm{SW}}:=K_2/2$, $q_R:=q_*$, and $a_R:=a_*$.
  Eq.~\eqref{eq:supp-centered-local} then gives
  \begin{equation}
    \varepsilon_R(\lambda)
    \le
    \frac{K_2}{2}\lambda^2
    =
    C_{\mathrm{SW}}\lambda^2.
    \label{eq:supp-quadratic-residual}
  \end{equation}
  The bounds $q_R$ and $a_R$ are independent of $N$.
  Together with the uniform locality of $H_{C,\lambda}^{(1)}$ proved above, this establishes every conclusion of the theorem.
\end{proof}

\subsection{Graph-stabilizer anchor with an Ising deformation}
\label{sec:supp-graph-reduction}

Consider the open graph-stabilizer chain of Sec.~\ref*{main-sec:graph-anchor} in the main text.
Its anchor Hamiltonian and Ising deformation are
\begin{equation}
  H_0=-\sum_{i=1}^Nh_iK_i,
  \qquad
  V=\sum_{i=1}^{N-1}J_iZ_iZ_{i+1}.
  \label{eq:supp-graph-physical}
\end{equation}
The graph tail satisfies $U_G^\dagger K_iU_G=Z_i$ and $U_G^\dagger Z_iU_G=X_i$, so the core-frame pair is
\begin{equation}
  H_{C,0}=-\sum_i h_iZ_i,
  \qquad
  V_C=\sum_iJ_iX_iX_{i+1}.
  \label{eq:supp-graph-core}
\end{equation}
Set $\sigma_i^\pm=(X_i\pm\mathrm iY_i)/2$ and define $A_i^{\mathrm p}=\sigma_i^+\sigma_{i+1}^+$ and $A_i^{\mathrm e}=\sigma_i^+\sigma_{i+1}^-$.
On bond $i$,
\begin{equation*}
  X_iX_{i+1}
  =A_i^{\mathrm p}+A_i^{{\mathrm p}\dagger}
  +A_i^{\mathrm e}+A_i^{{\mathrm e}\dagger}.
\end{equation*}
Using $[Z,\sigma^\pm]=\pm2\sigma^\pm$ gives the pair and exchange transition energies
\begin{equation}
  \omega_i^{\mathrm p}=-2(h_i+h_{i+1}),
  \qquad
  \omega_i^{\mathrm e}=-2(h_i-h_{i+1}).
  \label{eq:supp-graph-frequencies}
\end{equation}
When both frequencies are nonzero, the local generator
\begin{equation}
  S_i=-\frac{J_i}{\omega_i^{\mathrm p}}
  (A_i^{\mathrm p}-A_i^{{\mathrm p}\dagger})
  -\frac{J_i}{\omega_i^{\mathrm e}}
  (A_i^{\mathrm e}-A_i^{{\mathrm e}\dagger})
  \label{eq:supp-graph-generator}
\end{equation}
satisfies $[H_{C,0},S_i]=-J_iX_iX_{i+1}$.
The corresponding dimensionless channel parameters are
\begin{equation}
  \eta_i^{\mathrm p}
  =\frac{|\lambda J_i|}{2|h_i+h_{i+1}|},
  \qquad
  \eta_i^{\mathrm e}
  =\frac{|\lambda J_i|}{2|h_i-h_{i+1}|}.
  \label{eq:supp-graph-channel-parameters}
\end{equation}
Thus the prescribed correction is locally perturbative when $\eta_{\mathrm{graph}}:=\sup_i\max\{\eta_i^{\mathrm p},\eta_i^{\mathrm e}\}\ll1$.
Odd and even bonds are disjoint support classes, so the scheduled correction and physical tail are
\begin{equation}
  W_{\lambda,\mathrm{loc}}^{(1)}
  =\ee^{\lambda S_{\mathrm{odd}}}\ee^{\lambda S_{\mathrm{even}}},
  \qquad
  U_\lambda^{(1)}=U_GW_{\lambda,\mathrm{loc}}^{(1)}.
  \label{eq:supp-graph-correction}
\end{equation}
For the product-syndrome core, $P_C(V_C)=0$, so $H_{C,\lambda}^{(1)}=H_{C,0}$ and the original product loader remains valid.
Theorem~\ref*{main-thm:first-order-local-reduction} then gives a bounded-range residual with $\varepsilon_R(\lambda)\le C_{\mathrm{SW}}\lambda^2$ whenever its uniform assumptions hold.
In state-preparation order, one loads the product core, applies the even- and odd-bond correction layers in the order specified by Eq.~\eqref{eq:supp-graph-correction}, and then applies $U_G$.

The exchange channel is resonant at $h_i=h_{i+1}$, whereas the pair channel is resonant at $h_i=-h_{i+1}$.
At either equality, the corresponding division in Eq.~\eqref{eq:supp-graph-generator} is undefined, and near either equality one of the parameters in Eq.~\eqref{eq:supp-graph-channel-parameters} becomes large.
The affected channel must instead be incorporated through an enlarged loadable core or a model-specific resonant reduction, or left as an explicit first-order residual.
As a contrasting exact deformation, $\sum_iJ_iK_iK_{i+1}$ becomes $\sum_iJ_iZ_iZ_{i+1}$ in the core frame and can be absorbed into the finite-memory syndrome core of Sec.~\ref*{main-sec:graph-anchor} without changing $U_G$.

\subsection{Rydberg-dimer anchor with inter-dimer interactions}
\label{sec:supp-rydberg-reduction}

The independent-dimer anchor of Sec.~\ref*{main-sec:rydberg-anchor} retains the drive, detuning, and intra-dimer interaction within each four-level block.
For dimer $m$, let $\ket{\phi_{m,\alpha}}$ and $\epsilon_{m,\alpha}$, with $\alpha=1,\ldots,4$, denote its physical eigenstates and energies, and let $u_m\ket{\alpha}=\ket{\phi_{m,\alpha}}$.
The exact anchor data are
\begin{equation}
  H_{C,0}
  =\sum_{m=1}^{M}\sum_{\alpha=1}^{4}
  \epsilon_{m,\alpha}\ket{\alpha}\!\bra{\alpha}_m,
  \qquad
  U_0=\bigotimes_{m=1}^{M}u_m.
  \label{eq:supp-rydberg-reduction-anchor}
\end{equation}
Consider an open chain of dimers with bounded nearest-neighbor interactions
\begin{equation}
  V
  =\sum_{m=1}^{M-1}\sum_{a,b=1}^{2}
  J_m^{ab}n_{m,a}n_{m+1,b},
  \label{eq:supp-rydberg-deformation}
\end{equation}
where $a$ and $b$ label the two atoms within each dimer and the real couplings $J_m^{ab}$ are uniformly bounded.
Pulling this interaction through the block tail gives
\begin{equation}
  V_C
  =\sum_{m=1}^{M-1}\sum_{a,b=1}^{2}
  J_m^{ab}\widetilde n_{m,a}\widetilde n_{m+1,b},
  \qquad
  \widetilde n_{m,a}:=u_m^\dagger n_{m,a}u_m.
  \label{eq:supp-rydberg-core-deformation}
\end{equation}
This conjugation preserves the two-dimer support, although $\widetilde n_{m,a}$ is generally not diagonal in the dimer energy-label basis.

Introduce the matrix units $E_m^{\alpha\alpha'}:=\ket{\alpha}\!\bra{\alpha'}_m$ and expand
\begin{equation*}
  \widetilde n_{m,a}
  =\sum_{\alpha,\alpha'=1}^{4}
  \nu_{\alpha\alpha'}^{(m,a)}E_m^{\alpha\alpha'},
  \qquad
  \nu_{\alpha\alpha'}^{(m,a)}
  :=\bra{\phi_{m,\alpha}}n_{m,a}\ket{\phi_{m,\alpha'}}.
\end{equation*}
Then
\begin{equation}
  V_C
  =\sum_{m=1}^{M-1}
  \sum_{\alpha,\alpha',\beta,\beta'=1}^{4}
  \mathcal J_m^{\alpha\alpha';\beta\beta'}
  E_m^{\alpha\alpha'}E_{m+1}^{\beta\beta'},
  \label{eq:supp-rydberg-matrix-units}
\end{equation}
where
\begin{equation*}
  \mathcal J_m^{\alpha\alpha';\beta\beta'}
  :=\sum_{a,b=1}^{2}J_m^{ab}
  \nu_{\alpha\alpha'}^{(m,a)}
  \nu_{\beta\beta'}^{(m+1,b)}.
\end{equation*}
Hermiticity pairs each transition with its reverse.
The transition from initial labels $(\alpha',\beta')$ to final labels $(\alpha,\beta)$ has core-energy change
\begin{equation}
  \omega_{m;\alpha\alpha',\beta\beta'}
  =\epsilon_{m,\alpha}-\epsilon_{m,\alpha'}
  +\epsilon_{m+1,\beta}-\epsilon_{m+1,\beta'}.
  \label{eq:supp-rydberg-transition-energy}
\end{equation}

Let $P_C$ retain the diagonal part in the product energy-label basis.
Its action on Eq.~\eqref{eq:supp-rydberg-matrix-units} gives
\begin{equation}
  P_C(V_C)
  =\sum_{m=1}^{M-1}\sum_{\alpha,\beta=1}^{4}
  \mathcal J_m^{\alpha\alpha;\beta\beta}
  E_m^{\alpha\alpha}E_{m+1}^{\beta\beta}.
  \label{eq:supp-rydberg-retained}
\end{equation}
The retained term correlates neighboring four-valued energy labels without changing their basis states.
Consequently, $H_{C,\lambda}^{(1)}=H_{C,0}+\lambda P_C(V_C)$ is a one-dimensional nearest-neighbor classical model of four-valued labels.
The finite-memory construction in Sec.~\ref{sec:supp-finite-memory} applies after replacing each binary syndrome label by the four-valued dimer label.
The backward recursion has four possible next labels, and the controlled loader prepares the corresponding four-outcome conditional distribution on the next two-qubit dimer register.
All such controls have fixed width.

Every remaining term changes at least one energy label and belongs to $Q_C(V_C)$.
For each bond $m$, let $\mathcal I_m$ contain one orientation $\xi=(\alpha,\alpha';\beta,\beta')$ of every Hermitian-conjugate transition pair with $(\alpha,\beta)\ne(\alpha',\beta')$ and $\mathcal J_m^{\alpha\alpha';\beta\beta'}\ne0$.
This index set includes all nonzero off-diagonal channels, whether resonant or nonresonant.
For $\xi\in\mathcal I_m$, define
\begin{equation*}
  A_{m,\xi}
  :=\mathcal J_m^{\alpha\alpha';\beta\beta'}
  E_m^{\alpha\alpha'}E_{m+1}^{\beta\beta'},
  \qquad
  \omega_{m,\xi}:=\omega_{m;\alpha\alpha',\beta\beta'}.
\end{equation*}
The complete off-core perturbation therefore has the decomposition
\begin{equation}
  Q_C(V_C)
  =\sum_{m=1}^{M-1}\sum_{\xi\in\mathcal I_m}
  (A_{m,\xi}+A_{m,\xi}^\dagger),
  \qquad
  [H_{C,0},A_{m,\xi}]=\omega_{m,\xi}A_{m,\xi}.
  \label{eq:supp-rydberg-channels}
\end{equation}
For $\omega_{m,\xi}\ne0$, the anti-Hermitian generator
\begin{equation}
  S_{m,\xi}
  =\frac{A_{m,\xi}^\dagger-A_{m,\xi}}{\omega_{m,\xi}}
  \label{eq:supp-rydberg-generator}
\end{equation}
satisfies $[H_{C,0},S_{m,\xi}]=-(A_{m,\xi}+A_{m,\xi}^\dagger)$.
Assume that every channel in this complete decomposition obeys the uniform active-gap bound $|\omega_{m,\xi}|\ge\gamma_{\mathrm{Ry}}>0$ for all $m$ and $\xi$, with $\gamma_{\mathrm{Ry}}$ independent of the system size.
The full generator then satisfies the strong homological equation,
\begin{equation}
  S_1=\sum_{m=1}^{M-1}\sum_{\xi\in\mathcal I_m}S_{m,\xi},
  \qquad
  [H_{C,0},S_1]=-Q_C(V_C).
  \label{eq:supp-rydberg-homological}
\end{equation}
Summing the channels on each bond gives a two-dimer generator, and odd and even bonds again form two parallel layers.
Under the remaining uniform assumptions of Theorem~\ref*{main-thm:first-order-local-reduction}, the resulting residual satisfies $\varepsilon_R(\lambda)\le C_{\mathrm{Ry}}\lambda^2$ with a volume-independent constant.
If a nonzero off-diagonal channel instead has $\omega_{m,\xi}=0$, the active-gap bound fails and the generator in Eq.~\eqref{eq:supp-rydberg-generator} is undefined for that channel.
The resonant contribution must then be incorporated through a separately justified core enlargement or resonant reduction, or retained as an explicit first-order residual.

\subsection{Quadratic-fermion anchor with weak interactions}
\label{sec:supp-fermion-reduction}

The quadratic-fermion anchor of Sec.~\ref*{main-sec:quadratic-fermion-anchor} has the diagonal occupation core and Gaussian tail
\begin{equation}
  H_{C,0}
  =\sum_{j=1}^{N}\epsilon_j
  \left(n_{C,j}-\frac12\id\right),
  \qquad
  U_0=U_F,
  \label{eq:supp-fermion-reduction-anchor}
\end{equation}
where $n_{C,j}=c_j^\dagger c_j$.
We use the Jordan--Wigner ordering aligned with the physical chain and assume that the encoded tail $U_0$ has a bounded-depth local implementation.
This circuit assumption is required in addition to algebraic Gaussian diagonalizability.

Let $f_j$ denote the annihilation operator of physical mode $j$, set $n_j=f_j^\dagger f_j$, and define $B_j=f_j^\dagger f_{j+1}^\dagger f_{j+2}f_{j+3}$.
As a representative bounded-range interaction, take
\begin{equation}
  \begin{aligned}
    V={}&\sum_{j=1}^{N-1}\kappa_j
    \left(n_j-\frac12\id\right)
    \left(n_{j+1}-\frac12\id\right)\\
    &+\sum_{j=1}^{N-3}
    \left(t_jB_j+t_j^*B_j^\dagger\right),
  \end{aligned}
  \label{eq:supp-fermion-deformation}
\end{equation}
where $\kappa_j\in\mathbb R$ and $t_j\in\mathbb C$ are uniformly bounded.
Both terms conserve fermion parity and have bounded support after the aligned Jordan--Wigner encoding.

To keep the core-frame interaction local, assume separately that the Gaussian mode transformation is banded in the chosen ordering:
\begin{equation}
  \widetilde f_j
  :=U_0^\dagger f_jU_0
  =\sum_{k\in\mathcal N_j}
  \left(u_{jk}c_k+v_{jk}c_k^\dagger\right),
  \label{eq:supp-local-bogoliubov}
\end{equation}
where every neighborhood $\mathcal N_j$ has volume-independent size and diameter.
This bandedness assumption and the bounded-depth encoded circuit assumption control different aspects of the construction: the former preserves the range of $V_C$, while the latter controls the operational tail geometry.
The pulled-back interaction is
\begin{equation}
  \begin{aligned}
    V_C={}&\sum_{j=1}^{N-1}\kappa_j
    \left(\widetilde n_j-\frac12\id\right)
    \left(\widetilde n_{j+1}-\frac12\id\right)\\
    &+\sum_{j=1}^{N-3}
    \left(t_j\widetilde B_j+t_j^*\widetilde B_j^\dagger\right),
  \end{aligned}
  \label{eq:supp-fermion-core-deformation}
\end{equation}
where $\widetilde n_j=\widetilde f_j^\dagger\widetilde f_j$ and $\widetilde B_j=\widetilde f_j^\dagger\widetilde f_{j+1}^\dagger\widetilde f_{j+2}\widetilde f_{j+3}$.

Choose $P_C$ to retain the part diagonal in the simultaneous eigenbasis of the occupations $n_{C,j}$.
Because this dephasing acts independently on the occupation qubits, it does not enlarge the support of a core-frame term.
The retained component has the form
\begin{equation}
  P_C(V_C)
  =e_C\id+\sum_k h_k^{(C)}n_{C,k}
  +\sum_{k<\ell}\kappa_{k\ell}^{(C)}n_{C,k}n_{C,\ell},
  \label{eq:supp-fermion-retained}
\end{equation}
where the coefficients are obtained from diagonal matrix elements of $V_C$ in the vacuum, one-particle, and two-particle occupation states.
The bandedness in Eq.~\eqref{eq:supp-local-bogoliubov} makes the two-body coefficients finite range, and the scalar $e_C\id$ may be omitted from the normalized Gibbs state.
Thus $H_{C,\lambda}^{(1)}=H_{C,0}+\lambda P_C(V_C)$ is a finite-range classical Hamiltonian of binary occupations and can be prepared by the fixed-memory loader of Sec.~\ref{sec:supp-finite-memory}.

Give $V_C$ the bounded-support decomposition inherited from Eqs.~\eqref{eq:supp-fermion-deformation}--\eqref{eq:supp-local-bogoliubov} and decompose $Q_C(V_C)$ into local processes that change definite core occupations.
After pairing each process with its Hermitian-conjugate reverse, write the complete off-core perturbation as
\begin{equation}
  Q_C(V_C)=\sum_\alpha(A_\alpha+A_\alpha^\dagger),
  \qquad
  [H_{C,0},A_\alpha]=\omega_\alpha A_\alpha.
  \label{eq:supp-fermion-channels}
\end{equation}
The sum contains one orientation of every nonzero off-diagonal channel pair.
For example, a pair-transfer channel $A_j=g_jc_j^\dagger c_{j+1}^\dagger c_{j+2}c_{j+3}$ has transition energy $\omega_j=\epsilon_j+\epsilon_{j+1}-\epsilon_{j+2}-\epsilon_{j+3}$.
For $\omega_\alpha\ne0$, the anti-Hermitian generator
\begin{equation}
  S_\alpha=\frac{A_\alpha^\dagger-A_\alpha}{\omega_\alpha}
  \label{eq:supp-fermion-generator}
\end{equation}
satisfies $[H_{C,0},S_\alpha]=-(A_\alpha+A_\alpha^\dagger)$.
Assume that every channel in the complete decomposition in Eq.~\eqref{eq:supp-fermion-channels} obeys the uniform active-gap bound $|\omega_\alpha|\ge\gamma_F>0$, with $\gamma_F$ independent of the system size, and that these channels have bounded incident transition strength, uniformly bounded supports, and a uniformly bounded overlap degree.
The full generator then satisfies the strong homological equation,
\begin{equation}
  S_1=\sum_\alpha S_\alpha,
  \qquad
  [H_{C,0},S_1]=-Q_C(V_C).
  \label{eq:supp-fermion-homological}
\end{equation}
Combining channels with the same support and coloring overlapping supports then gives a volume-independent number of circuit layers.
Under the remaining uniform assumptions of Theorem~\ref*{main-thm:first-order-local-reduction}, the resulting residual satisfies $\varepsilon_R(\lambda)\le C_F\lambda^2$ with a volume-independent constant.
If a fixed physical parity is imposed, the resulting global constraint on the occupation strings must also be included in the core-preparation procedure.

\subsection{Proof of Proposition~\ref*{main-prop:formal-higher-order-reduction}}
\label{sec:supp-higher-order-proof}

\begin{proof}
  For each $j=1,\ldots,k$, introduce the right-commutator superoperator
  \begin{equation*}
    \mathcal D_j(O):=[O,S_j].
  \end{equation*}
  The BCH conjugation formula gives
  \begin{equation*}
    \ee^{-S_\lambda^{(k)}}O\ee^{S_\lambda^{(k)}}
    =
    \exp\!\left(\sum_{j=1}^k\lambda^j\mathcal D_j\right)(O).
  \end{equation*}
  Apply this identity separately to $H_{C,0}$ and to $\lambda V_C$.
  At order $\lambda^\ell$, the only term containing the new generator $S_\ell$ is $\mathcal D_\ell(H_{C,0})=[H_{C,0},S_\ell]$.
  Any additional commutator composed with $\mathcal D_\ell$ raises the order above $\ell$, while $\mathcal D_\ell$ acting on the prefactored term $\lambda V_C$ begins at order $\ell+1$.
  All remaining contributions therefore depend only on the earlier generators $S_1,\ldots,S_{\ell-1}$.

  Collecting those contributions gives
  \begin{equation}
    \begin{split}
      \Xi_\ell={}&\delta_{\ell1}V_C
      +\sum_{r=2}^{\ell}\frac1{r!}
      \sum_{\substack{j_1+\cdots+j_r=\ell\\j_a\ge1}}
      \mathcal D_{j_r}\cdots\mathcal D_{j_1}(H_{C,0})\\
      &+\sum_{r=1}^{\ell-1}\frac1{r!}
      \sum_{\substack{j_1+\cdots+j_r=\ell-1\\j_a\ge1}}
      \mathcal D_{j_r}\cdots\mathcal D_{j_1}(V_C),
      \label{eq:supp-Xi-ell}
    \end{split}
  \end{equation}
  where $\delta_{\ell1}$ is the Kronecker delta and a sum over an empty range vanishes.
  In the first nested-commutator sum, $r\ge2$ forces every $j_a<\ell$.
  In the second, the indices sum to $\ell-1$ and again satisfy $j_a<\ell$.
  This proves the stated decomposition of $G_\ell$ and the dependence of $\Xi_\ell$ on earlier data only.

  Choose the generators recursively.
  Once $S_1,\ldots,S_{\ell-1}$ have been fixed, $\Xi_\ell$ is known and the order-$\ell$ homological equation is linear in $S_\ell$.
  If it is solved, then $Q_C(G_\ell)=0$ and therefore
  \begin{equation*}
    G_\ell=P_C(G_\ell)=\Delta H_C^{(\ell)}\in\mathfrak H_C.
  \end{equation*}
  Repeating this argument for $\ell=1,\ldots,k$ makes every Taylor coefficient through order $k$ core compatible.
  At fixed finite volume, the truncated exponential is analytic in $\lambda$, so Taylor's theorem places all remaining terms in an operator-norm remainder of order $O(\lambda^{k+1})$.
  This proves the proposition.
\end{proof}
\section{Numerical methodology and additional results}
\label{sec:supp-numerics}

This section gives the complete protocol underlying the finite-size numerical demonstrations in Sec.~\ref*{main-sec:numerics} of the main text.
The calculations test the mechanism of the first-order local reduction on one exactly solvable benchmark.

\subsection{{Problem setup, residual propagation, and Gibbs-state evaluation}}
\label{sec:supp-pauli-propagation}

{We use the open-chain graph-stabilizer model and its core-frame representation in Eqs.~\eqref{eq:supp-graph-physical} and \eqref{eq:supp-graph-core}, with uniform $J_i=J=1$ and one-based site labels.
The fields are $h_i=\bar h+(-1)^i\delta$, with $\bar h=1.25$ and, unless explicitly varied, $\delta=0.25$, so they alternate between $1$ and $1.5$.}
{With $\sigma_i^\pm=(X_i\pm\mathrm{i}Y_i)/2$, the bond transitions and their energies are}
\begin{equation}
  {
    \begin{aligned}
      A_{i,\mathrm p}&=\sigma_i^+\sigma_{i+1}^+,
      &\omega_i^{\mathrm p}&=-2(h_i+h_{i+1})=-4\bar h,\\
      A_{i,\mathrm e}&=\sigma_i^+\sigma_{i+1}^-,
      &\omega_i^{\mathrm e}&=-2(h_i-h_{i+1})=-4(-1)^i\delta.
  \end{aligned}}
  \label{eq:supp-numerical-transition-energies}
\end{equation}
{These give $X_iX_{i+1}=A_{i,\mathrm p}+A_{i,\mathrm p}^{\dagger}+A_{i,\mathrm e}+A_{i,\mathrm e}^{\dagger}$.
For $\delta\ne0$, the bond generator $S_i$ is Eq.~\eqref{eq:supp-graph-generator} with $J_i=J$.}
Writing $S_{\rm odd}$ and $S_{\rm even}$ for the sums on the two disjoint bond layers, the implemented correction is the exact scheduled circuit
\begin{equation}
  W_{\lambda,\mathrm{loc}}^{(1)}
  =\exp(\lambda S_{\rm odd})\exp(\lambda S_{\rm even}).
  \label{eq:supp-numerical-correction}
\end{equation}
Each two-qubit exponential is evaluated directly.
{The bare and corrected residuals are defined in Eqs.~\maineqref{eq:numerical-bare-residual} and \maineqref{eq:numerical-corrected-residual} of the main text, and the corrected residual contains all higher-order terms generated by the finite rotation.}

{The dimensionless channel ratios are $\eta_{\mathrm p}=|\lambda J|/(4\bar h)$ and $\eta_{\mathrm e}=|\lambda J|/(4|\delta|)$.}
At the default values, the exchange channel is restrictive and $\max\{\eta_{\mathrm p},\eta_{\mathrm e}\}=|\lambda|$.
{At $\delta=0$, the exchange denominator vanishes and Eq.~\eqref{eq:supp-graph-generator} does not define a finite exchange generator.}

For residual calculations, an operator is stored as a sparse dictionary in the phase-free Pauli basis.
{The expansion and directly evaluated Pauli incident strength are}
\begin{equation}
  {
    \begin{aligned}
      R&=\sum_{P\ne\id}c_PP,\\
      \varepsilon_P(R)&:=\max_{1\le i\le N}
      \sum_{P:\,i\in\supp(P)}|c_P|.
  \end{aligned}}
  \label{eq:supp-pauli-expansion}
\end{equation}
where $P\in\{\id,X,Y,Z\}^{\otimes N}$ and the coefficients are real because $R$ and the basis strings are Hermitian.

The relation between this computable quantity and the centered incident strength in Theorem~\ref*{main-thm:first-order-local-reduction} requires care.
{Group all strings having exactly the same support into a declared interaction term and define its spectral center:}
\begin{equation}
  {
    \begin{aligned}
      r_X&:=\sum_{P:\,\supp(P)=X}c_PP,
      &R&=\sum_{\varnothing\ne X\subseteq\{1,\ldots,N\}}r_X,\\
      c_X^\star&:=\frac{\lambda_{\max}(r_X)+\lambda_{\min}(r_X)}{2}.
  \end{aligned}}
  \label{eq:supp-equal-support-grouping}
\end{equation}
{For the centered spectral half-width, spectral centering can only reduce the operator norm, and the triangle inequality then gives}
\begin{equation}
  {J_X:=\norm{r_X-c_X^\star\id_X}_\infty
    \le \norm{r_X}_\infty
  \le\sum_{P:\,\supp(P)=X}|c_P|.}
  \label{eq:supp-grouped-width-bound}
\end{equation}
{Consequently the intrinsic incident strength of this equal-support decomposition obeys the precise one-sided relation}
\begin{equation}
  {\varepsilon_R:=\max_i\sum_{X\ni i}J_X\le
    \max_i\sum_{X\ni i}
    \sum_{P:\,\supp(P)=X}|c_P|
  =\varepsilon_P(R).}
  \label{eq:supp-epsilonR-epsilonP}
\end{equation}
Equality is not assumed: Pauli components with the same support can cancel in operator norm, and spectral centering may lower the spectral half-width further.
If each Pauli string were instead retained as a separate declared term, its spectral half-width would be $|c_P|$ and the associated declared strength would equal $\varepsilon_P$.

To propagate Eq.~\eqref{eq:supp-pauli-expansion} through the circuit without forming a dense many-body matrix, let $G_i$ be one correction gate on sites $(i,i+1)$, set $(\tau_0,\tau_1,\tau_2,\tau_3)=(\id,X,Y,Z)$, and write $P_{ab}=\tau_a\otimes\tau_b$.
Conjugating a local Pauli operator gives
\begin{equation}
  G_i^\dagger P_{ab}G_i
  =\sum_{c,d=0}^{3}M_{cd,ab}^{(i)}P_{cd}.
  \label{eq:supp-local-pauli-transfer-action}
\end{equation}
The transfer coefficients are
\begin{equation}
  M_{cd,ab}^{(i)}
  =\frac14\Tr\!\left[P_{cd}G_i^\dagger P_{ab}G_i\right].
  \label{eq:supp-local-pauli-transfer}
\end{equation}
The factor $1/4$ is fixed by Pauli orthogonality on the four-dimensional two-qubit Hilbert space.
For a full string $P=P_{<i}\otimes P_{ab}\otimes P_{>i+1}$, the update is
\begin{equation}
  c_PP\longmapsto
  \sum_{c,d=0}^{3}c_PM_{cd,ab}^{(i)}
  P_{<i}\otimes P_{cd}\otimes P_{>i+1}.
  \label{eq:supp-full-pauli-update}
\end{equation}
Coefficients of coincident output strings are added immediately.
Applying the updates across the odd layer and then the even layer evaluates $W_{\lambda,\mathrm{loc}}^{(1)\dagger}(H_{C,0}+\lambda V_C)W_{\lambda,\mathrm{loc}}^{(1)}$ in the order required by Eq.~\eqref{eq:supp-numerical-correction}.
Transfer coefficients are evaluated in floating-point arithmetic, and coefficients of magnitude below $10^{-13}$ are discarded after each gate conjugation.
This sparse procedure reaches $N=32$ for the reported residual calculations.

To evaluate the physical-state errors, we use dense exact diagonalization at $N=8$.
The complete eigensystem is retained and the Boltzmann weights are normalized.
The exact physical Gibbs state is
\begin{equation}
  \rho_{\rm exact}
  =U_G\rhoG{H_{C,0}+\lambda V_C}U_G^\dagger.
  \label{eq:supp-physical-exact-state}
\end{equation}
The bare state is
\begin{equation}
  \rho_{\rm bare}
  =U_G\rhoG{H_{C,0}}U_G^\dagger.
  \label{eq:supp-physical-bare-state}
\end{equation}
The first-order corrected state is
\begin{equation}
  \rho_{\rm corr}
  =U_GW_{\lambda,\mathrm{loc}}^{(1)}\rhoG{H_{C,0}}
  W_{\lambda,\mathrm{loc}}^{(1)\dagger}U_G^\dagger.
  \label{eq:supp-physical-comparison-states}
\end{equation}
Using these physical-frame states, we evaluate the $D_1$ and $\Delta_{ZZ}$ metrics defined in Sec.~\ref*{main-sec:numerics} of the main text.

\subsection{Fit protocol and numerical checks}
\label{sec:supp-fit-protocol}

For the residual scaling, $N=12$ and the ten couplings $\lambda_k=10^{-3+k/4}$, $k=0,\ldots,9$, are used.
Each data set is fitted by ordinary least squares in log space to
\begin{equation}
  \ln y=\ln C+p\ln|\lambda|,
  \label{eq:supp-single-power-fit}
\end{equation}
using all ten points.
{The fitted exponents and their standard errors are reported in Sec.~\ref*{main-sec:numerics} of the main text.}

The size scan fixes $\lambda=0.05$ and uses $N\in\{4,6,8,12,16,24,32\}$.
{The bare value is $0.1$ throughout, whereas the corrected value is $0.0116$ at the reported precision for $N\ge12$ in this scan.}
The largest propagated string in this scan has support cardinality six and diameter five.

For the physical-state scaling, $N=8$, $\beta\in\{0.5,1,2\}$, and the nine couplings $\lambda_k=0.005(1.5)^k$, $k=0,\ldots,8$, are used.
For each observable and each construction, all $27$ values are therefore fitted simultaneously to
\begin{equation}
  \ln y_\beta(\lambda)=b_\beta+p\ln|\lambda|,
  \label{eq:supp-grouped-power-fit}
\end{equation}
with three independent intercepts $b_\beta$ and one common slope $p$.
Note that this is a single grouped regression, rather than an average of three separate slopes.
{As a check on the grouped fits, individual-temperature fits for $D_1$ give bare slopes between $0.9999$ and $1.0009$ and corrected slopes between $1.9975$ and $1.9987$.
For $\Delta_{ZZ}$, the corresponding individual-temperature ranges are $0.9998$--$1.0006$ and $2.9979$--$2.9988$.}

\subsection{Sublattice symmetry and the cubic corrected bond error}
\label{sec:supp-ZZ-symmetry}

{The near-cubic corrected exponent in Fig.~\ref*{main-fig:thermal-accuracy}(b) of the main text follows from a symmetry specific to this bipartite benchmark and observable.}
The physical bond $Z_iZ_{i+1}$ pulls back through $U_G$ to $O_i=X_iX_{i+1}$.
Define
\begin{equation}
  \mathcal Z_{\rm e}:=\prod_{j\ {\rm even}}Z_j.
  \label{eq:supp-even-sublattice-Z}
\end{equation}
Every nearest-neighbor bond contains exactly one even site.
Consequently,
\begin{equation}
  \mathcal Z_{\rm e}H_{C,0}\mathcal Z_{\rm e}=H_{C,0}.
  \label{eq:supp-sublattice-anchor}
\end{equation}
The deformation changes sign:
\begin{equation}
  \mathcal Z_{\rm e}V_C\mathcal Z_{\rm e}=-V_C.
  \label{eq:supp-sublattice-deformation}
\end{equation}
The first-order generators likewise obey
\begin{equation}
  \mathcal Z_{\rm e}S_i\mathcal Z_{\rm e}=-S_i.
  \label{eq:supp-sublattice-generator}
\end{equation}
Finally, the pulled-back bond observable changes sign:
\begin{equation}
  \mathcal Z_{\rm e}O_i\mathcal Z_{\rm e}=-O_i.
  \label{eq:supp-sublattice-transformations}
\end{equation}
The exact core-frame Gibbs state $\rho_{{\rm exact},C}(\lambda)=\rhoG{H_{C,0}+\lambda V_C}$ therefore obeys
\begin{equation}
  \rho_{{\rm exact},C}(-\lambda)
  =\mathcal Z_{\rm e}\rho_{{\rm exact},C}(\lambda)\mathcal Z_{\rm e}.
  \label{eq:supp-exact-state-symmetry}
\end{equation}
Likewise, $\mathcal Z_{\rm e}W_{\lambda,\mathrm{loc}}^{(1)}\mathcal Z_{\rm e}=W_{-\lambda,\mathrm{loc}}^{(1)}$, and invariance of $\rhoG{H_{C,0}}$ gives
\begin{equation}
  \rho_{{\rm corr},C}(-\lambda)
  =\mathcal Z_{\rm e}\rho_{{\rm corr},C}(\lambda)\mathcal Z_{\rm e}.
  \label{eq:supp-corrected-state-symmetry}
\end{equation}
Using cyclicity of the trace and the last relation in Eq.~\eqref{eq:supp-sublattice-transformations}, both $\Tr[O_i\rho_{{\rm exact},C}(\lambda)]$ and $\Tr[O_i\rho_{{\rm corr},C}(\lambda)]$ are odd functions of $\lambda$.
Their difference is therefore odd as well.
The first-order cancellation equation removes its linear term.
The next symmetry-allowed term is cubic.
Hence the corrected bond error is $O(|\lambda|^3)$ in this model, consistent with the fitted exponent.
This argument neither upgrades the full $D_1$ error beyond second order nor supplies a generic third-order local-state theorem.

\subsection{Variational transition circuit and free-energy objective}
\label{sec:supp-variational-protocol}

Outside the controlled regime, the transition directions are retained while their finite angles are optimized.
{The bond gate $W_i(\theta_{\mathrm p},\theta_{\mathrm e})$ is defined in Eq.~\maineqref{eq:numerical-variational-bond-gate} of the main text.}
The pair and exchange generators act on the even- and odd-excitation two-qubit subspaces, respectively, and commute on the same bond.
Thus this exponential is evaluated exactly.
The full variational correction has the same brick-wall order as the prescribed first-order circuit,
\begin{equation}
  W(\theta_{\mathrm p},\theta_{\mathrm e})
  =\prod_{i\ {\rm odd}}W_i(\theta_{\mathrm p},\theta_{\mathrm e})
  \prod_{i\ {\rm even}}W_i(\theta_{\mathrm p},\theta_{\mathrm e}),
  \label{eq:supp-variational-correction}
\end{equation}
with the same two angles on every bond.
{For $\delta\ne0$, the perturbative angles in this convention are}
\begin{equation}
  {\theta_{\mathrm p}^{(1)}=\frac{\lambda J}{4\bar h},
    \qquad
  \theta_{\mathrm e}^{(1)}=\frac{\lambda J}{4\delta}.}
  \label{eq:supp-first-order-angles}
\end{equation}

{The variational core $\rho_C(\alpha_{\mathrm o},\alpha_{\mathrm e})$ is the two-sublattice product state in Eq.~\maineqref{eq:numerical-variational-core} of the main text, with $s(i)=\mathrm o$ on odd sites and $s(i)=\mathrm e$ on even sites.}
At fixed $\beta$, this is the Gibbs state of
\begin{equation}
  H_C(\alpha_{\mathrm o},\alpha_{\mathrm e})
  =-\frac1\beta\sum_{i=1}^{N}
  \ln\!\left(\cot\alpha_{s(i)}\right)Z_i
  \label{eq:supp-variational-core-hamiltonian}
\end{equation}
up to a scalar.
{For $\boldsymbol\vartheta=(\theta_{\mathrm p},\theta_{\mathrm e},\alpha_{\mathrm o},\alpha_{\mathrm e})$, the physical trial state $\rho_{\boldsymbol\vartheta}$ is Eq.~\maineqref{eq:numerical-variational-state} of the main text.}

The pair and exchange factors in the bond gate can be written as elementary two-qubit Pauli rotations.
For $p,q\in\{x,y,z\}$, let $\sigma_x=X$, $\sigma_y=Y$, and $\sigma_z=Z$, and define
\begin{equation}
  R_{pq}(\phi)
  =\exp\!\left[-\frac{\mathrm{i}\phi}{2}\sigma_p\otimes\sigma_q\right].
  \label{eq:supp-variational-pauli-rotation}
\end{equation}
The pair factor is
\begin{equation}
  W_i^{\mathrm p}(\theta_{\mathrm p})
  =R_{xy}(-\theta_{\mathrm p})R_{yx}(-\theta_{\mathrm p}).
  \label{eq:supp-variational-pair-factor}
\end{equation}
Writing $\epsilon_i=(-1)^i$, the exchange factor is
\begin{equation}
  W_i^{\mathrm e}(\theta_{\mathrm e})
  =R_{yx}(-\epsilon_i\theta_{\mathrm e})R_{xy}(\epsilon_i\theta_{\mathrm e}).
  \label{eq:supp-variational-exchange-factor}
\end{equation}
The two rotations within each factor commute, and $W_i=W_i^{\mathrm p}W_i^{\mathrm e}$.
Fig.~\ref{fig:supp-variational-circuit} displays the purification-based preparation and expands one pair and one exchange gate into these rotations.

\begin{figure}[t!]
  \centering
  \resizebox{0.8\linewidth}{!}{%
    \begin{tikzpicture}[
        x=1cm,
        y=0.56cm,
        wire/.style={line width=0.45pt},
        loader/.style={draw=NavyBlue,fill=NavyBlue!9,rounded corners=1.5pt,minimum width=1.66cm,minimum height=0.50cm,align=center,font=\scriptsize},
        corrgate/.style={draw=BurntOrange!85!black,fill=BurntOrange!12,rounded corners=1.5pt,minimum width=0.64cm,minimum height=0.68cm,align=center,font=\scriptsize},
        rotationgate/.style={draw=BurntOrange!85!black,fill=BurntOrange!8,rounded corners=1.2pt,minimum width=1.26cm,minimum height=0.92cm,align=center,font=\tiny},
        hadgate/.style={draw=ForestGreen!65!black,fill=ForestGreen!10,rounded corners=1.5pt,minimum width=0.54cm,minimum height=0.50cm,align=center,font=\scriptsize},
        ctrlone/.style={circle,draw=black,fill=black,inner sep=0pt,minimum size=5.0pt},
        target/.style={
          circle,
          draw=black,
          fill=white,
          inner sep=0pt,
          minimum size=9.0pt,
          path picture={
            \draw[wire] (path picture bounding box.west) -- (path picture bounding box.east);
            \draw[wire] (path picture bounding box.south) -- (path picture bounding box.north);
          }
        },
        czdot/.style={circle,draw=ForestGreen!65!black,fill=ForestGreen!65!black,inner sep=0pt,minimum size=5.0pt},
        stage/.style={densely dashed,rounded corners=2pt,line width=0.45pt}
      ]
      \draw[densely dashed,draw=black,line width=0.50pt]
      (4.83,1.50) -- (4.83,-0.85) -- (3.24,-1.40);
      \draw[densely dashed,draw=black,line width=0.50pt]
      (5.67,1.50) -- (5.67,-0.85) -- (8.11,-1.40);
      \foreach \i/\ya/\parity in {1/7/o,2/6/e,3/5/o,4/4/e}{
        \node[anchor=east,font=\scriptsize] at (0.90,\ya) {$A_{\i}:\ket{0}$};
        \draw[wire] (1.00,\ya) -- (4.28,\ya);
        \node[loader] at (2.10,\ya) {$R_y(2\alpha_{\mathrm{\parity}})$};
      }
      \foreach \i/\ys in {1/3,2/2,3/1,4/0}{
        \node[anchor=east,font=\scriptsize] at (0.90,\ys) {$S_{\i}:\ket{0}$};
        \draw[wire] (1.00,\ys) -- (10.68,\ys);
      }
      \foreach \ya/\ys/\xc in {7/3/3.15,6/2/3.45,5/1/3.75,4/0/4.05}{
        \draw[preaction={draw=white,line width=1.8pt},wire] (\xc,\ya) -- (\xc,\ys);
        \node[ctrlone] at (\xc,\ya) {};
        \node[target] at (\xc,\ys) {};
      }

      \draw[stage,draw=NavyBlue!70] (1.10,-0.58) rectangle (4.30,7.76);
      \node[font=\scriptsize\bfseries,text=NavyBlue,align=center,anchor=south] at (2.74,7.94)
      {variational core loader\\[-1pt]$\alpha_{\mathrm o},\alpha_{\mathrm e}$};

      \node[corrgate] at (4.83,1.5) {$W_2^{\mathrm p}$};
      \node[corrgate] at (5.67,1.5) {$W_2^{\mathrm e}$};
      \node[corrgate] at (6.48,2.5) {$W_1^{\mathrm p}$};
      \node[corrgate] at (6.48,0.5) {$W_3^{\mathrm p}$};
      \node[corrgate] at (7.32,2.5) {$W_1^{\mathrm e}$};
      \node[corrgate] at (7.32,0.5) {$W_3^{\mathrm e}$};
      \draw[stage,draw=BurntOrange!80!black] (4.38,-0.58) rectangle (7.76,3.58);
      \node[font=\scriptsize\bfseries,text=BurntOrange!85!black,align=center,anchor=south] at (6.07,3.76)
      {variational correction\\[-1pt]$\theta_{\mathrm p},\theta_{\mathrm e}$};

      \foreach \ys in {3,2,1,0}{
        \node[hadgate] at (8.36,\ys) {$H$};
      }
      \draw[wire,draw=ForestGreen!65!black] (8.92,2) -- (8.92,1);
      \node[czdot] at (8.92,2) {};
      \node[czdot] at (8.92,1) {};
      \draw[wire,draw=ForestGreen!65!black] (9.72,3) -- (9.72,2);
      \node[czdot] at (9.72,3) {};
      \node[czdot] at (9.72,2) {};
      \draw[wire,draw=ForestGreen!65!black] (9.72,1) -- (9.72,0);
      \node[czdot] at (9.72,1) {};
      \node[czdot] at (9.72,0) {};
      \draw[stage,draw=ForestGreen!60!black] (7.90,-0.58) rectangle (9.99,3.58);
      \node[font=\scriptsize\bfseries,text=ForestGreen!55!black,align=center,anchor=south] at (8.95,3.76)
      {fixed graph tail $U_G$};
      \node[anchor=west,font=\scriptsize] at (10.82,1.5)
      {$\rho_{\boldsymbol\vartheta}$};

      \draw[stage,draw=BurntOrange!65] (1.10,-4.22) rectangle (5.38,-1.40);
      \node[font=\scriptsize\bfseries,text=BurntOrange!85!black] at (3.24,-1.82)
      {$W_2^{\mathrm p}$};
      \node[anchor=east,font=\scriptsize] at (1.64,-2.65) {$S_2$};
      \node[anchor=east,font=\scriptsize] at (1.64,-3.41) {$S_3$};
      \draw[wire] (1.72,-2.65) -- (5.18,-2.65);
      \draw[wire] (1.72,-3.41) -- (5.18,-3.41);
      \node[rotationgate] at (2.59,-3.03) {$R_{xy}(-\theta_{\mathrm p})$};
      \node[rotationgate] at (4.31,-3.03) {$R_{yx}(-\theta_{\mathrm p})$};

      \draw[stage,draw=BurntOrange!65] (5.76,-4.22) rectangle (10.45,-1.40);
      \node[font=\scriptsize\bfseries,text=BurntOrange!85!black] at (8.11,-1.82)
      {$W_2^{\mathrm e}$};
      \node[anchor=east,font=\scriptsize] at (6.36,-2.65) {$S_2$};
      \node[anchor=east,font=\scriptsize] at (6.36,-3.41) {$S_3$};
      \draw[wire] (6.44,-2.65) -- (10.18,-2.65);
      \draw[wire] (6.44,-3.41) -- (10.18,-3.41);
      \node[rotationgate] at (7.44,-3.03) {$R_{yx}(-\epsilon_2\theta_{\mathrm e})$};
      \node[rotationgate] at (9.21,-3.03) {$R_{xy}(\epsilon_2\theta_{\mathrm e})$};
    \end{tikzpicture}%
  }
  \caption{Variational Gibbs-state preparation circuit for the deformed alternating-field graph-stabilizer Hamiltonian, illustrated for four system qubits and four ancilla qubits.
    The ancilla qubits are labeled $A_i$ and the system qubits are labeled $S_i$.
    The alternating angles $\alpha_{\mathrm o}$ and $\alpha_{\mathrm e}$ parameterize the variational core loader, and the parallel CNOTs prepare its purification.
    Superscripts $\mathrm p$ and $\mathrm e$ identify the pair and exchange factors in each bond gate, with the two insets expanding $W_2^{\mathrm p}$ and $W_2^{\mathrm e}$ into Pauli rotations.
  The fixed graph tail applies a Hadamard to every system qubit followed by the two controlled-$Z$ layers, mapping the corrected core-frame state to the physical frame.}
  \label{fig:supp-variational-circuit}
\end{figure}
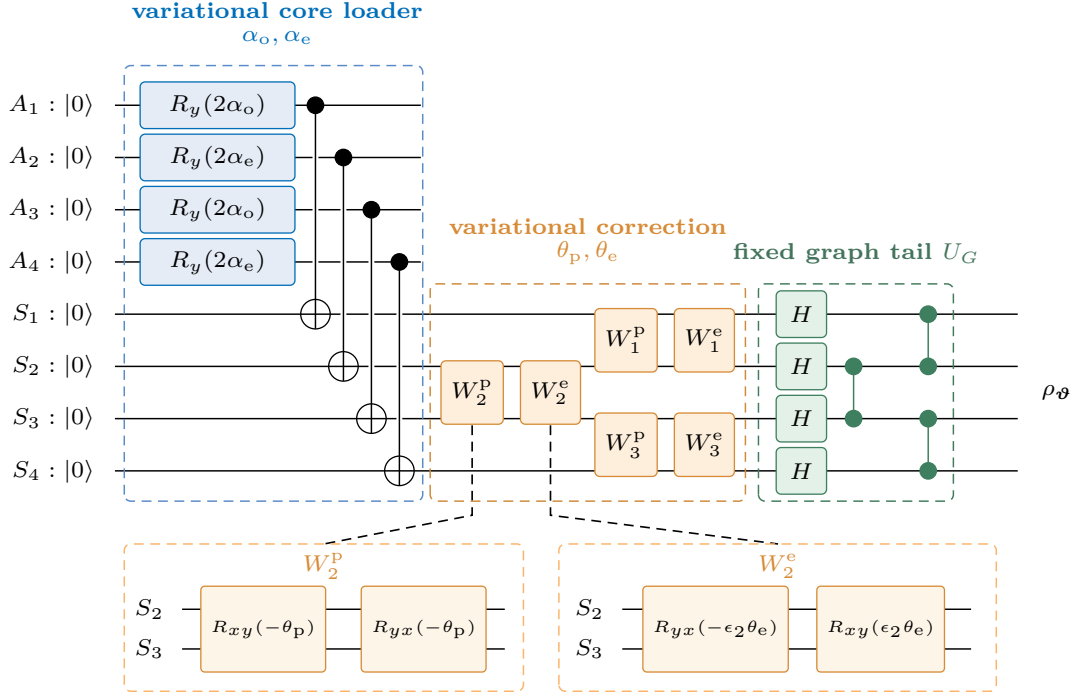

The core loader can be purified with one ancilla per site by applying $R_y(2\alpha_{s(i)})$ to the ancilla followed by a CNOT from ancilla to system.
All subsequent gates act only on the system register and preserve the Schmidt coefficients.
The entropy of the physical trial state is therefore known exactly from the loader probabilities:
\begin{equation}
  \begin{split}
    S(\rho_{\boldsymbol\vartheta})
    =-\sum_{i=1}^{N}\big[&
      \cos^2\!\alpha_{s(i)}\ln(\cos^2\!\alpha_{s(i)})
      \\
    &+\sin^2\!\alpha_{s(i)}\ln(\sin^2\!\alpha_{s(i)})\big].
  \end{split}
  \label{eq:supp-variational-entropy}
\end{equation}
{The optimized objective is the physical free-energy density in Eq.~\maineqref{eq:numerical-variational-free-energy} of the main text.}
In the reported finite-size calculation, the energy trace is evaluated with dense matrices and the entropy by Eq.~\eqref{eq:supp-variational-entropy}.

\subsection{Optimizer, multistart protocol, and failure scans}
\label{sec:supp-optimization-details}

{The free-energy density in Eq.~\maineqref{eq:numerical-variational-free-energy} of the main text is minimized with the L-BFGS-B optimizer, using SciPy's default finite-difference gradient~\cite{main-VirtanenEtAl2020}.}
The transition-angle bounds are
\begin{equation}
  -\pi\le\theta_{\mathrm p},\theta_{\mathrm e}\le\pi.
  \label{eq:supp-transition-angle-bounds}
\end{equation}
The loader-angle bounds are
\begin{equation}
  10^{-6}\le\alpha_{\mathrm o},\alpha_{\mathrm e}
  \le\frac\pi2-10^{-6}.
  \label{eq:supp-optimizer-bounds}
\end{equation}
The loader cutoff avoids the endpoints at which the effective field in Eq.~\eqref{eq:supp-variational-core-hamiltonian} diverges.

Ten initial conditions are used at every parameter point.
The bare start has $\theta_{\mathrm p}=\theta_{\mathrm e}=0$ and the odd-sublattice loader angle
\begin{equation}
  \alpha_{\mathrm o}
  =\arctan\!\left[\ee^{-\beta(\bar h-\delta)}\right].
  \label{eq:supp-bare-odd-initialization}
\end{equation}
The even-sublattice loader angle is
\begin{equation}
  \alpha_{\mathrm e}
  =\arctan\!\left[\ee^{-\beta(\bar h+\delta)}\right].
  \label{eq:supp-bare-initialization}
\end{equation}
{For $\delta\ne0$, the perturbative start uses Eq.~\eqref{eq:supp-first-order-angles} with the same loader angles.}
At exact resonance, the perturbative exchange angle is undefined, so that start uses $\theta_{\mathrm e}=0$ and retains only the nonsingular pair angle.
The other eight starts are pseudorandom points in the admissible parameter domain.
The L-BFGS-B options are \texttt{maxiter=500}, \texttt{ftol=1e-13}, \texttt{gtol=1e-9}, and \texttt{maxls=50}.
Among the ten converged runs, the reported state is the one with the smallest final free-energy value found.
All $200$ runs---ten starts at each of twenty parameter points---reported formal convergence.
At some points, only one start reached the lowest objective within the configured tolerances, so the multistart procedure materially affected the result.

The resonance scan uses $N=8$, $\beta=1$, $\lambda=0.02$, and the nine positive staggerings
\begin{equation}
  \delta\in\{2^{-1},2^{-2},\ldots,2^{-9}\},
  \label{eq:supp-resonance-grid}
\end{equation}
together with a separate exact-resonance calculation at $\delta=0$.
At the two smallest positive staggerings, $\eta_{\mathrm e}=1.28$ and $2.56$, respectively.

Exact resonance is not part of that logarithmic plot and must be interpreted separately.
The first-order exchange generator is singular there, so the reported first-order comparator is the pair-channel-only rotation.
At $\delta=0$, the variational result is
\begin{equation}
  D_1^{\rm var}=3.93\times10^{-3}.
  \label{eq:supp-exact-resonance-variational}
\end{equation}
The pair-only first-order comparator gives
\begin{equation}
  D_1^{\rm pair\text{-}only}=2.81\times10^{-3}.
  \label{eq:supp-exact-resonance-result}
\end{equation}
Thus free-energy optimization does not improve the chosen physical local metric at exact resonance in this benchmark.
This is compatible with the variational principle because the optimized objective is global free energy instead of $D_1$.

The finite-deformation scan instead fixes $N=8$, $\beta=1$, $\delta=0.25$ and uses $\lambda\in\{0.05,0.1,0.2,0.3,0.45,0.65,0.8,1,2,3\}$.
The first-order and bare $D_1$ values first cross in the sampled bracket $0.8\le\lambda\le1$.
Interpolation of their ratio linearly in $\ln\lambda$ gives $\lambda_\times\simeq0.813$.
At $\lambda=1$, the variational, bare, and first-order errors are $0.327$, $0.494$, and $0.557$, respectively.
At $\lambda=3$, they are $0.763$, $1.020$, and $0.990$.

\end{document}